\documentclass[3p,review, 11pt]{elsarticle}

\renewcommand*{\today}{}

\journal{Journal of Theoretical Biology}

\usepackage{pdfpages}
\usepackage{amsthm,amsmath,color, amsfonts,mathrsfs}
\usepackage{subfig}
\usepackage{wrapfig}
\usepackage{float}
\usepackage{caption}
\usepackage{hhline}
\usepackage{array}
\usepackage{delarray}
\usepackage{longtable}
\usepackage[all]{xy}
\newtheorem{thm}{Theorem}[section]
\newtheorem{cor}[thm]{Corollary}

\newtheorem{exmp}[thm]{Example}
\newtheorem{rem}[thm]{Remark}
\newtheorem{defn}[thm]{Definition}

\begin{document}

\begin{frontmatter}

\title{Identifiability of the unrooted species tree topology under the coalescent model with time-reversible substitution processes, site-specific rate variation, and invariable sites}

\author{Julia Chifman}
\address{Department of Cancer Biology, Wake Forest School of Medicine, Winston-Salem, NC, 27157}
\ead{jchifman@wakehealth.edu}

\author{Laura Kubatko\corref{cor1}}
\ead{lkubatko@stat.osu.edu}
\address{Department of Statistics, Department of Evolution, Ecology, and Organismal Biology,
Mathematical Biosciences Institute, The Ohio State University, Columbus, OH 43210}

\cortext[cor1]{Corresponding author}

\begin{abstract}
The inference of the evolutionary history of a collection of organisms is a problem of fundamental importance in evolutionary biology. The abundance of DNA sequence data arising from genome sequencing projects has led to significant challenges in the inference of these phylogenetic relationships. Among these challenges is the inference of the evolutionary history of a collection of species based on sequence information from several distinct genes sampled throughout the genome. It is widely accepted that each individual gene has its own phylogeny, which may not agree with the {\em species tree.} Many possible causes of this gene tree incongruence are known. The best studied is incomplete lineage sorting, which is commonly modeled by the coalescent process. Numerous methods based on the coalescent process have been proposed for estimation of the phylogenetic species tree given DNA sequence data.  However, use of these methods assumes that the phylogenetic species tree can be identified from DNA sequence data at the leaves of the tree, although this has not been formally established. {\em We prove that the unrooted topology of the $n$-leaf phylogenetic species tree is generically identifiable given observed data at the leaves of the tree that are assumed to have arisen from the coalescent process under a time-reversible substitution process with the possibility of site-specific rate variation modeled by the discrete gamma distribution and a proportion of invariable sites.}
\end{abstract}


\end{frontmatter}


\section{Introduction}

The field of evolutionary genetics has benefitted enormously from recent advances in sequencing technology that have led to the availability of DNA sequence information for hundreds or thousands of species.  These data are commonly used to study evolutionary patterns and processes. A fundamental problem in this area is the inference of a {\itshape phylogenetic species tree} that describes the evolutionary relationships among a collection of species for which data have been collected. A species phylogeny is a tree with all internal nodes of degree three, except for a root node of degree two. All edges are treated as directed from the root toward the leaves. The tree represents the biological process of speciation, in which one population splits into two populations which then evolve independently, with no subsequent exchange of genetic material. The root node represents the most ancestral population to all sampled species, while the leaves (also called {\itshape taxa}; singular: {\itshape taxon}) represent present-day populations. An example of a phylogenetic species tree for four species, $a$, $b$, $c$, and $d$,  is shown by the outlined tree in Figure \ref{hist}.

Although the goal is generally to estimate the species phylogeny from available DNA sequence data, these sequence data are only directly informative about the {\itshape gene tree} -- the phylogenetic tree underlying the gene for which the DNA sequences are available. It is well-accepted that gene trees and species trees may not agree with one another (see, e.g., \cite{maddison1997,pamilonei1988}), with many evolutionary processes known to give rise to variability in gene phylogenies within a fixed species phylogeny.  Examples of such processes are incomplete lineage sorting (ILS), hybridization, horizontal gene transfer, and gene duplication and loss \cite{maddison1997}. The best studied of these processes is ILS, which results when two lineages fail to share a most recent common ancestor (MRCA; represented by an internal node in the gene tree) until further back in time than the immediately ancestral population. For example, in Figure \ref{histA}, the gene tree embedded within the species tree represents the phylogenetic history of the lineages sampled from species $a$, $b$, $c$, and $d$, which are denoted by $A$, $B$, $C$, and $D$, respectively. Throughout the text, we use uppercase letters to refer to gene lineages, and the corresponding lowercase letters to refer to the species from which these lineages are sampled. Although it is possible for lineages $C$ and $D$ to share their most recent common ancestor in the population labeled $P_1$, they remain distinct in this population, and instead share their MRCA in population $P_3$, thus providing an example of ILS.  Note that the topology (branching pattern) of the gene tree in Figure \ref{histA} matches the topology of the species tree. Figure \ref{histB} gives another example of ILS, but in this case lineage $A$ coalesces with lineage $C$ in their ancestral population, and the gene tree topology does not match the species tree topology.

One of the reasons that ILS has been well-studied is that it can be modeled by the coalescent process.  The coalescent process can be derived as the large sample limit (as the population size goes to $\infty$) of the Wright-Fisher and other common population genetics models \cite{kingman1982a,kingman1982b,tavare1984}.  The key property of the coalescent model is that the waiting time back into the past for  a pair of lineages to find their MRCA follows an exponential distribution, with a parameter that depends on the sample size.  The coalescent model thus provides a link between the phylogenetic species tree and the set of gene trees embedded within the species tree that give rise to the actual data. For this reason, numerous methods based on the coalescent process have recently been proposed for estimation of the phylogenetic species tree.  One group of methods (e.g., BEST \cite{liupearl2007}, *BEAST \cite{heleddrummond2010}, STEM \cite{kubatkoetal2009}, MP-EST \cite{liuetal2010}) assumes that multi-locus data are available for inference, with the assumption that each locus has a single underlying (unobserved) gene tree.  Alternatively, single nucleotide polymorphism (SNP) data are sometimes used for inference. SNP data represent sites sampled throughout the genome that are known to be variable, with the assumptions that the sites are unlinked and that they each have their own phylogenetic history. The software package SNAPP \cite{bryantetal2012} has recently been developed for species tree estimation from biallelic SNP data.  Use of any of these methods assumes that the phylogenetic species tree can be identified from DNA sequence data at the leaves of the tree, but this has not formally been established  (note, however, that Allman et al. (2011)  \cite{allmanetal2011} have established identifiability given  a collection of gene tree topologies; Allman et al. (2011) \cite{allmanetal2011b} have considered identifiability given clade probabilities; and Liu and Edwards (2009) \cite{liuedwards2009} have established identifiability when the order of ancestral populations, and hence the relationships among all rooted triples, can be consistently estimated).  

Here, we prove that the unrooted topology of the phylogenetic species tree is identifiable given observed SNP data at the leaves of the tree that are assumed to have arisen from the coalescent process. Our results hold for data for which a single observation corresponds to recording which of $\kappa$ possible states occurs at each leaf. These data are modeled by a continuous-time Markov process that specifies the rates of transitions between states along the phylogeny and that satisfies the condition of time-reversibility. We also consider models that allow rate variation across sites as modeled by the discrete gamma distribution, as well as the possibility of invariable sites. For the special case of DNA sequence data, there are four states (i.e., $\kappa=4$) corresponding to the four nucleotides $A$, $C$, $G$, and $T$.  In this case, our results hold for the General Time Reversible (GTR; \cite{tavare1986}) model with discrete-gamma distributed rate variation and a proportion of invariable sites, and all associated sub-models.

In the next section, we give the necessary background on the coalescent process and on the process of mutation for general $\kappa$-state models, pointing out the application to DNA sequence data where relevant. We also examine some common modifications to site-independent sequence substitution models to allow for variation in the rate of evolution across sites.  We then present our main results and show how they are used to establish identifiability in the general case. Based on these results, we propose a method for inferring species-level relationships for empirical data sets consisting of DNA sequence data. We conclude by suggesting extensions of our current work.

\section{Background}  

In this section, we review the models used for both the coalescent process and the mutation process along a phylogenetic tree.  By a {\em topological tree} we mean a tree for which edge length are not specified. If a tree, rooted or unrooted, is endowed with a collection of non-negative edge length, we will say that a tree is {\em metric}. 

Let $S$ denote a rooted, binary, topological $n$-taxon {\em species} tree and let $\boldsymbol{\tau} = (\tau_1, \tau_2, \ldots , \tau_{n-1})$ be a vector of speciation times (looking backward in time), $0 < \tau_j < \infty$. Then  a pair $(S, \boldsymbol{\tau})$ represents a phylogenetic rooted metric $n$-leaf species tree. Similarly, let $G$ be a rooted, binary, $n$-taxon topological {\em gene} tree and let $\mathbf{t} = (t_1, t_2, \ldots , t_{n-1})$ be a vector of coalescent times, where $t_j$ is the time from the $j^{th}$ coalescent event to the next speciation event (looking forward in time), $0 < t_j < \infty$, for $j=1, 2, \ldots , n-1$. Figure \ref{hist} shows examples of gene trees nested within species trees with all of these quantities labeled. In addition, observe that in Figure \ref{histA} (upper left image) the coalescent event that happens in population $P_2$ could also happen in population $P_3$. Thus, a pair  $(G, \mathbf{t})$ will represent a phylogenetic rooted metric $n$-leaf gene tree along with coalescent times that are associated with specific intervals between speciation events on the species tree. The set of all such pairs, $(G, \mathbf{t})$, conditional on the species tree $(S, \boldsymbol{\tau})$ will be denoted by $\mathcal{G}_S$. Notice that both $(S, \boldsymbol{\tau})$ and $(G, \mathbf{t})$  are ultrametric, that is, all leaves are equidistant from the root. In phylogenetics such trees are said to satisfy the {\em molecular clock}. Here we measure time in {\itshape coalescent units}, which are the number of generations scaled by $2N\mu$, where $N$ is the population size and $\mu$ is the mutation rate.

\begin{figure}
\centering
\subfloat[]{\label{histA}\includegraphics[scale=0.4]{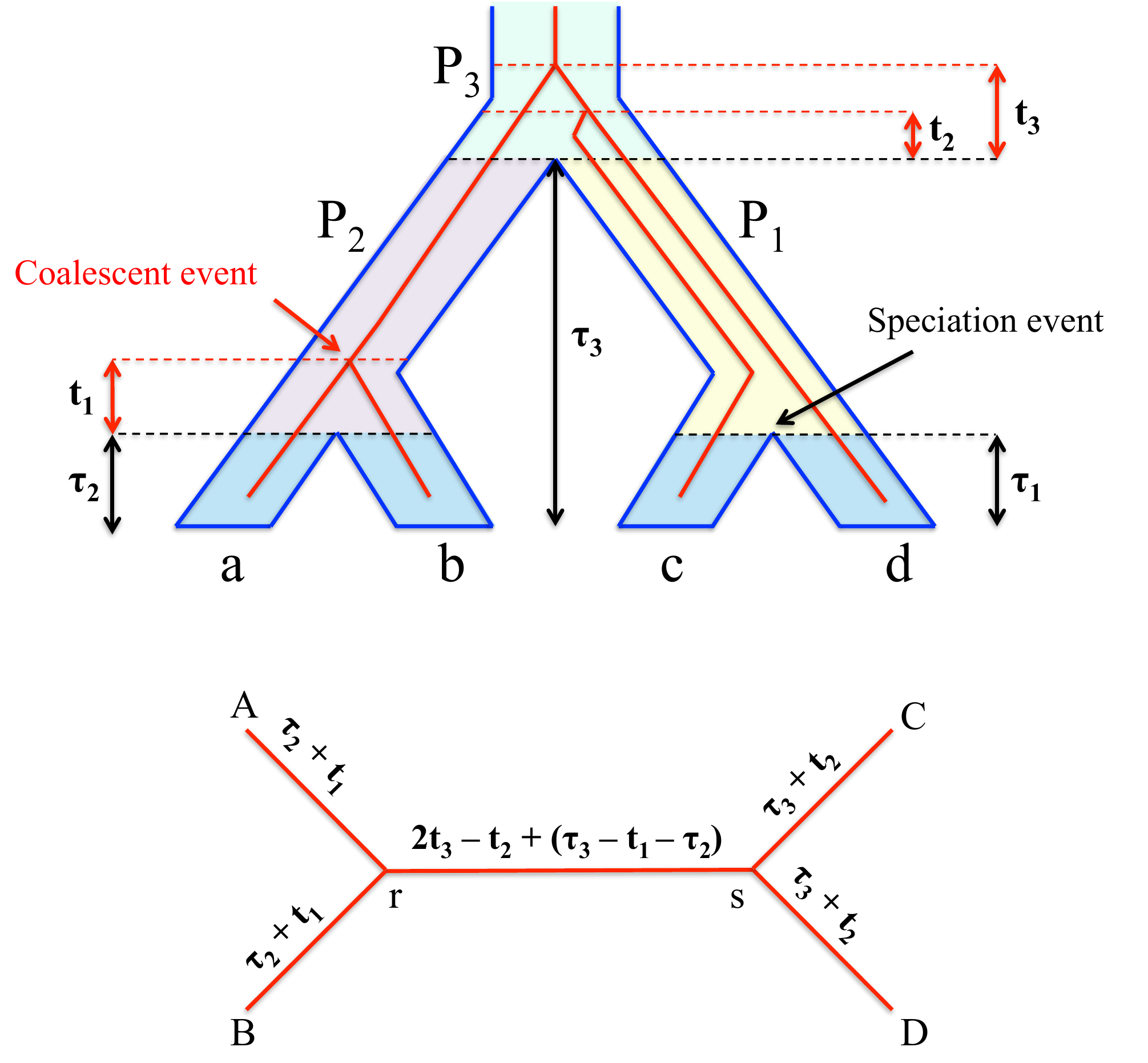}}
\hspace{0.3in}
\subfloat[]{\label{histB}\includegraphics[scale=0.4]{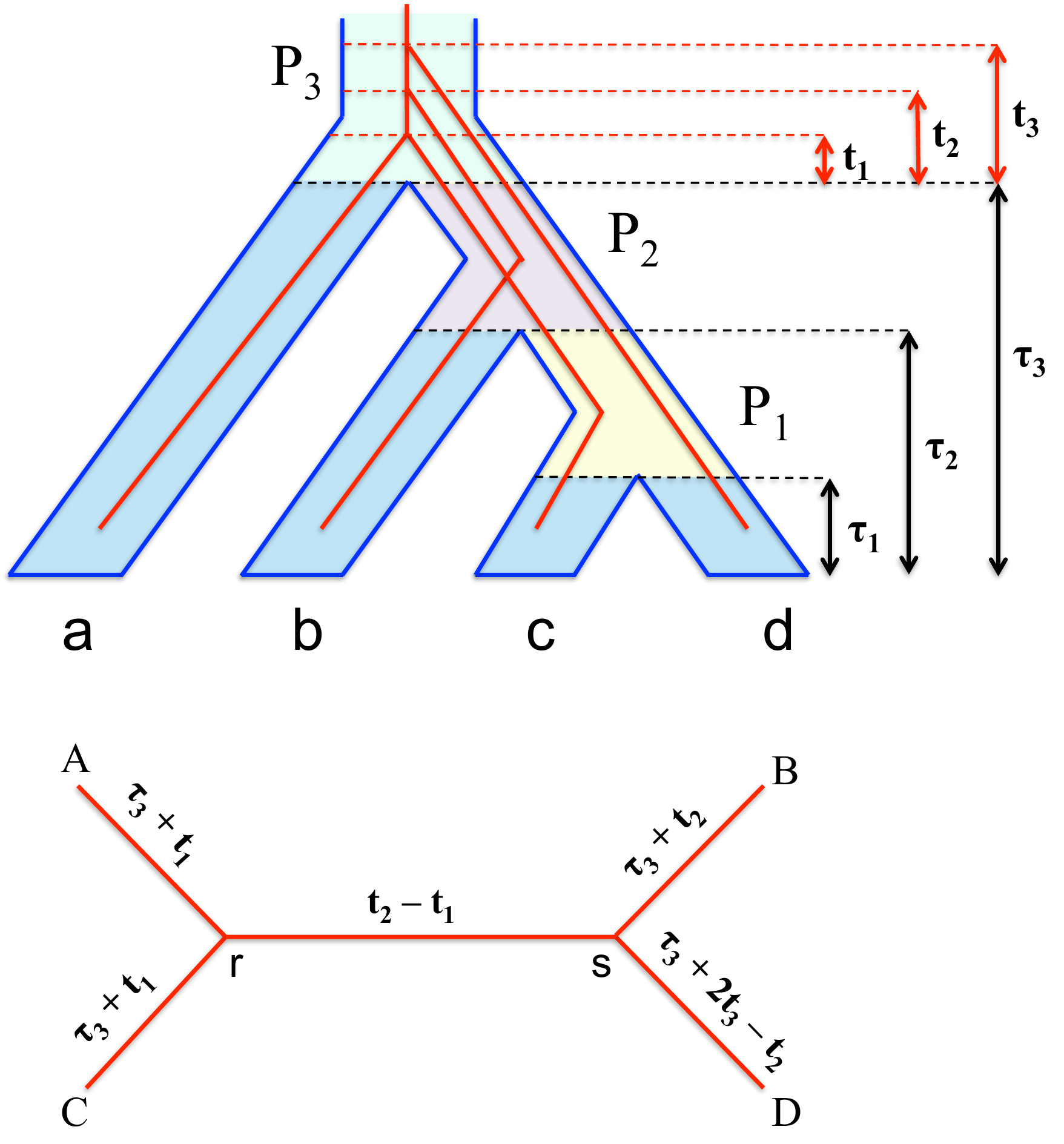}}
\caption{Example of two gene trees (red) nested within a species tree (blue). In (a) the gene tree and the species tree have the same topology, while in (b) they do not agree with one another. In both subfigures, the dotted black horizontal lines represent speciation events in the species tree.  The times of these events, denoted by $\tau_i$, are measured backward from the present time.  Portions of the tree that fall between species events represent ancestral populations, denoted by $P_i$. The dotted horizontal red lines represent coalescent events that occur in the gene trees.  The times of coalescent events, denoted by $t_i$, are measured from the most recent (looking forward in time) speciation event. }
\label{hist}
\end{figure}

\subsection{The Coalescent Process}

The coalescent is a model for the evolutionary history (i.e., sequence of coalescent events) of a sample of lineages within a population backward in time from the present to the past \cite{kingman1982a,kingman1982b,tavare1984}. In particular, given a sample of $j$ gene lineages, the coalescent model specifies that the time $t_j$ until the next pair of lineages coalesces follows an exponential distribution with rate $\left( \binom{j}{2}\frac{2}{\theta} \right)^{-1}$, i.e. $t_j$ has probability density,
\begin{equation}\label{eq:exp.density} 
f(t_j) = \frac{j(j-1)}{2}\frac{2}{\theta} \mbox{exp}\left(-\frac{j(j-1)}{2} \frac{2}{\theta}t_j\right), \mbox{       } t_j > 0,
\end{equation}
where the parameter $\theta$ is the effective population size, which is given by $\theta = 2N\mu$.  

Now consider a population within a species tree, e.g., $P_1$ in Figure \ref{histA}, which corresponds to a time interval of length $\tau = \tau_3- \tau_1$, and let $b$ represent this branch. Following Rannala and Yang (2003), \nocite{rannalayang2003} let $u$ denote the number of lineages ``entering'' the population (e.g., the number of lineages in the population closest to the present time) and let $v$ be the number of lineages ``leaving'' the population, $v \leq u$.  The number of coalescent events within the population is $u-v$, and the density of the time of each coalescent event can be determined via Equation \ref{eq:exp.density} by setting $j$ to be the number of current lineages in the population and following an assumption of the coalescent model that each pair of lineages within a population is equally likely to be the next to coalesce.  This means that when there are $j$ lineages available to coalesce in the population, there are $\binom{j}{2}$ possible pairs that might be the next to coalesce, and the density for the next event should be weighted by the probability of a particular pair coalescing, which is $1/\binom{j}{2}$.  Let $t^b_j$ denote the time from the speciation event that immediately precedes branch $b$ (looking backward in time) to the coalescent event that reduces the number of lineages on branch $b$ from $j$ to $j-1$. Define $t^b_{u+1}$ to be 0, and let $\tau_{b}$ refer to the length of branch $b$.
Then, we can write the joint density of coalescent times $t^b_u, t^b_{u-1}, \ldots , t^b_{v+1}$ within population $P$ on branch $b$ in the case in which $u>v$ as
\begin{eqnarray}\label{eq:joint.density.one.pop}
 \nonumber f_{P_b}(t^b_u, t^b_{u-1}, \ldots , t^b_{v+1})  &  = & \prod_{j=v+1}^{u} \left[ \frac{2}{\theta} \mbox{exp} \left(-\frac{j(j-1)}{\theta}(t^b_j - t^b_{j+1})\right) \right] \\ 
 &  \times &  \mbox{exp} \left(-\frac{v(v-1)}{\theta}(\tau_{b} - t^b_{v+1} ) \right),\\ 
\nonumber & & 0 < t^b_j < \infty, \  j = u+1, u, u-1, \ldots , v+1. 
\end{eqnarray}
The terms inside the product in Equation \ref{eq:joint.density.one.pop} correspond to observed coalescent events, while the final term reflects the fact that when $v \neq 1$, the coalescent event that decreases the number of lineages from $v$ to $v-1$ does not occur within the time remaining in population $P_b$, which is $\tau_{b} - t^b_{v+1}$.  

When $u=v$ for branch $b$, no coalescent events happen on that branch.  The probability of this occurring under the model is 
\begin{equation}\label{eq:prob.no.coal}
P(\mbox{no coalescence among $v$ lineages in population $P_b$}) = \mbox{exp} \left(-\frac{v(v-1)}{\theta}\tau_{b} \right),
\end{equation}
where $\tau_{b}$ is defined as above to be the length of branch $b$.  Note that when $u=v=1$, as is the case for the branches at the tips of the tree, this probability is 1.

As an example, refer again to Figure \ref{histA}.  For population $P_1$, we have $u=v=2$ and the probability of this is given by Equation \ref{eq:prob.no.coal} as 
\begin{equation}
P(\mbox{no coalescence among 2 lineages in population $P_1$}) = \mbox{exp}\left(-\frac{2}{\theta}(\tau_{3} - \tau_{1}) \right).
\end{equation}
For population $P_3$, $u=3$, $v=1$, $t^{P_3}_{u} = t_2$, $t^{P_3}_{u-1} = t^{P_3}_{v+1} = t_3$, and the joint density is 
\begin{equation}
f_{P_3}(t^{P_3}_u, t^{P_3}_{v+1}) = 
f_{P_3}(t_2, t_3)  =
\left\{ \frac{2}{\theta}\mbox{exp}\left(-\frac{2}{\theta}(t_3-t_2)\right) \right\} \left\{\frac{2}{\theta}\mbox{exp} \left( -\frac{6}{\theta}t_2\right) \right\}.
\end{equation}
 For population $P_2$, $u=2$, $v=1$, and $t^{P_2}_{u} = t^{P_2}_{v+1} = t_1$, and the joint density is 
 \begin{equation}
f_{P_2}(t^{P_2}_u)  =  
f_{P_2}(t_1)  =  
\frac{2}{\theta}\mbox{exp} \left( -\frac{2}{\theta}t_1\right).
\end{equation}

Equation \ref{eq:joint.density.one.pop} and Equation \ref{eq:prob.no.coal} describe the coalescent process {\itshape within} a population.  We now wish to apply the coalescent model to the entire phylogenetic species tree in order to derive the probability density for gene trees nested within the species tree. Again following Rannala and Yang (2003), \nocite{rannalayang2003} we note that once the number of lineages entering and leaving a population on a species phylogeny is specified, the coalescent processes within each of the populations are conditionally independent of one another.  Thus, the densities within individual populations can be multiplied to give the overall gene tree density given a particular species tree with speciation time vector $\boldsymbol{\tau}$,
\begin{equation}\label{gtd}
f((G,\mathbf{t})|(S, \boldsymbol{\tau}))= \prod_{b=1}^{n-1} f_{P_b} (t^b_{u_b}, t^b_{u_{b}-1}, \ldots , t^b_{v_b+1}),
\end{equation}
where the index $b$ is over populations (e.g., branches) in the species phylogeny $(S,\boldsymbol{\tau})$, $u_b$ is the number of lineages entering population $P_b$, and $v_b$ is the number of lineages leaving population $P_b$. Note that the collection of $t^b_i$ terms across all branches is equivalent to the vector $\mathbf{t}$ in  $(G,\mathbf{t})$. 
We emphasize that the notation $f((G,\mathbf{t})|(S, \boldsymbol{\tau}))$ for the joint density is chosen to reflect the fact that this is the joint density of the topology and branch lengths of the gene tree $(G, \mathbf{t})$, conditional on the species tree $(S, \boldsymbol{\tau})$. 

\subsection{The Mutation Process}\label{mprocess}

The process of evolutionary change along a phylogeny is commonly modeled by a continuous-time Markov process.  In this section, we describe the general case of a Markov mutation process on $\kappa$ states, and then consider the commonly used models for DNA sequence data for which $\kappa=4$.  We begin by specifying a general time-reversible $\kappa \times \kappa$ instantaneous rate matrix $\mathbf{Q}$ such that entry $q_{ij}, i \neq j$, gives the instantaneous rate of change from state $i$ to state $j$, 
\begin{center}
{\small 
\begin{gather}
\mathbf{Q} =
\begin{pmatrix}\label{qmatrix}
  - & \pi_{2}\mu_1 & \pi_3 \mu_2 & \cdots & \pi_\kappa \mu_{\kappa-1}  \\
  \pi_1 \mu_1  & - & \pi_3 \mu_{\kappa} & \cdots & \pi_\kappa \mu_{2\kappa-3} \\
  \pi_1 \mu_2   & \pi_2 \mu_{\kappa}  & - & \cdots & \pi_\kappa \mu_{3\kappa-6} \\
  \vdots  & \vdots & \vdots  & \ddots & \vdots  \\
  \pi_1 \mu_{\kappa-1}  & \pi_2 \mu_{2\kappa-3} &\pi_3 \mu_{3\kappa-6}  & \cdots & -
 \end{pmatrix}.
\end{gather}
}
\end{center}
In the matrix above, each $\pi_j$ term represents the frequency of state $j$ at equilibrium, with the constraints that $\pi_j > 0$ for $j \in [\kappa]:= \{1,2, \dots , \kappa \}$, and $\sum_{j=1}^{\kappa} \pi_j = 1$.  The model contains $\binom{\kappa}{2}$ additional parameters $\mu_1, \ldots , \mu_{\frac{1}{2}\kappa(\kappa-1)}$ that specify the rates of mutation between states, with the assumption that $\mu_j >0$ for $j=1, 2, \ldots , \kappa(\kappa-1)/2$. The diagonal entries of $\mathbf{Q}$ are set so that the rows sum to zero.  Note that the $\mu_l$ term in the $(i,j)^{th}$ entry is the same as in the $(j,i)^{th}$ entry, so that the model satisfies the condition of time-reversibility, i.e., $\pi_i \mathbf{Q}_{ij} = \pi_j \mathbf{Q}_{ji}$. The intuition of the model is that a pair of states $i$ and $j$ have the same basic rate of mutating from one to the other, represented by $\mu_l$, but the rate of moving to state $j$ depends also on the frequency of state $j$, given by $\pi_j$.

The instantaneous rate matrix is used to compute the transition probability matrix $\mathbf{P}(t)$ such that the $(i,j)^{th}$ entry of $\mathbf{P}(t)$ gives the probability that state $i$ mutates to state $j$ in an interval of time of length $t$.  This is done by solving the matrix differential equation $\mathbf{P}^{'}(t) = \mathbf{Q}\mathbf{P}(t)$ with initial condition $\mathbf{P}(0) = \mathbf{I}$, to give $\mathbf{P}(t) = e^{\mathbf{Q}t}$.    In phylogenetics, it is common to refer to the matrix $\mathbf{P}(t)$ as the {\itshape substitution matrix} rather than the transition probability matrix, because the word ``transition'' has a particular meaning in the process of DNA sequence mutation.

The matrix $\mathbf{Q}$ is the generalization of the 4-state GTR model for DNA sequence data to $\kappa$ states, in the sense that it allows a separate parameter for mutations between each pair of states.  All of the models we consider here can then be thought of as special cases of  $\mathbf{Q}$.  For example, if we specify $\mu_1 = \mu_2 =  \ldots =  \mu_{\frac{1}{2}\kappa(\kappa-1)}$ and $\pi_j = \frac{1}{\kappa}$ for $j \in  [\kappa]$, the resulting model is the $\kappa$-state analog of the Jukes-Cantor model \cite{jukescantor1969}, which has been called the Mk model by Lewis (2001) \cite{lewis2001}. 

When DNA sequence data are available at the tips of the tree, $\kappa=4$ and we let $j =1, 2, 3, 4$ correspond to the four nucleotides $A$, $C$, $G$, and $T$, respectively.  Below we list the restrictions on the parameters in $\mathbf{Q}$ that lead to many of the commonly-used substitution models in empirical phylogenetics.\\

{\small
\noindent {\bf [JC69] Jukes-Cantor model \cite{jukescantor1969}}: $\pi_i = \frac{1}{4}$ for $i=1, 2, 3, 4$; $\mu_j= \mu$ for all $j =1, 2, \ldots 6$.\\

\noindent {\bf [K2P] Kimura's 2-parameter model \cite{kimura1980}}: $\pi_i = \frac{1}{4}$ for $i=1, 2, 3, 4$; $\mu_1=\mu_6>0$; $\mu_2=\mu_3=\mu_4=\mu_5>0$.\\

\noindent {\bf [F81] Felsenstein's 1981 model \cite{felsenstein1981}}: $\mu_j = \mu$ for $j=1, 2, \ldots 6$; $\pi_i \in (0,1)$ for $i=1, 2, 3, 4$ with $\sum_{i=1}^{4} \pi_i = 1$.\\

\noindent {\bf [HKY85] Hasegawa, Kishino, and Yano's 1985 model \cite{hasegawaetal1985}}: $\mu_1=\mu_6>0$; $\mu_2=\mu_3=\mu_4=\mu_5>0$; $\pi_i \in (0,1)$ for $i=1, 2, 3, 4$ with $\sum_{i=1}^{4} \pi_i = 1$.\\

\noindent {\bf [TN93] Tamura and Nei's 1993 model \cite{tamuranei1993}}: $\mu_1 >0$; $\mu_6 > 0$; $\mu_2=\mu_3=\mu_4=\mu_5>0$; $\pi_i \in (0,1)$ for $i=1, 2, 3, 4$ with $\sum_{i=1}^{4} \pi_i = 1$. \\

\noindent {\bf [GTR] General time-reversible model \cite{tavare1986}}: $\mu_j >0$ for $j=1, 2, \ldots 6$; $\pi_i \in (0,1)$ for $i=1, 2, 3, 4$ with $\sum_{i=1}^{4} \pi_i = 1$.\\
}

Two additional modifications of the basic substitution models listed above are commonly used in the analysis of DNA sequence data. Both are designed to model the fact that observed rates of evolution are known to vary across sites in the sequences, with some sites evolving more quickly than others and some sites essentially never evolving.    To capture differential rates of evolution across sites, it is common to incorporate a parameter, $\rho$, for the rate of evolution at a site. Traditionally, the gamma distribution with shape parameter $\alpha > 0$ and rate parameter $\beta > 0$ is used \cite{yang1993}, which is generally denoted as ``+ $\Gamma$''. The density function for the gamma distribution is
\begin{equation}\label{gamma}
g(\rho | \alpha, \beta) = \begin{cases}
   \frac{\beta^\alpha \rho^{\alpha-1} e^{-\rho  \beta}}{ \Gamma(\alpha)}, & \text{if } \rho \geq 0 \\
   0  ,     & \text{if } \rho < 0,
     \end{cases}
  \end{equation}
where $\Gamma(\alpha)$ is the gamma function. In applications, it is common for computational purposes to use a discrete gamma model to approximate the continuous gamma distribution, which involves specifying a finite number of rate classes \cite{yang1994}.  In particular, the range of $\rho$ is cut into $m$ categories such that each category has equal probability $\frac{1}{m}$. For each category $i \in [m]$ a single rate $\rho_i$ is used to assign the rate for that class.  The rate $\rho_i$ is generally chosen to be the mean of the corresponding continuous gamma distribution for that class, computed as follows,
\begin{equation}\label{rho}
\rho_i = \frac{\int_a^b \rho \cdot g(\rho | \alpha, \beta) d \rho}{\int_a^b g(\rho | \alpha, \beta) d \rho}
\end{equation}
(see \cite{yang1994}).
Here, $a$ and $b$ represent the endpoints of the interval for category $i$. If a site has rate parameter $\rho_i$, then its instantaneous rate matrix becomes $\rho_i \mathbf{Q}$, rather than simply $\mathbf{Q}$, and the corresponding change to the substitution matrix $\mathbf{P}(t)$ is made. We denote the substitution matrix that depends on $\rho_i$ by $\mathbf{P}^{\rho_i}(t)$, with $(i,j)^{th}$-element $\mathbf{P}_{ij}^{\rho_i}(t)$.

The second modification models the fact that some sites in a DNA sequence are known not to evolve, and are thus called {\itshape invariable}.  To capture this, a parameter $\delta$ is incorporated that gives the probability that a site is invariable.  If a site is invariable, none of the states observed at the leaves differs from the state at the root. With probability $1-\delta$, the mutation probabilities at the site follow one of the Markov models described above.   This possibility is typically denoted ``+I'' in specifying the overall substitution model. Putting this all together, we represent this model by ``GTR+I+$\Gamma$''. These models have been considered in work establishing identifiability in the gene tree case \cite{allmanrhodes2006,allmanetal2008,allmanrhodes2008b}. Below we give results for ``GTR+I+$\Gamma$'' models for which the gamma-distributed rates are modeled using the discrete gamma approximation described above.

Thus far, we have described the evolutionary model for mutations between states along specific branches of a gene tree.  We now describe how these models are used to compute the probability distribution of the data at the leaves of a phylogenetic gene tree.  The key idea in calculating site pattern probabilities is that the states are not observed at the internal nodes of the tree, and so the probability must be computed by summing over all possible states for these nodes.  Additionally, probabilities along each branch are multiplied together because we assume that the mutation process proceeds independently along each branch.  Finally, we point out that we are utilizing what have been called {\itshape rate matrix models} \cite{allmanrhodes2006} in phylogenetics, since we assume that the Markov mutation process is homogeneous across branches of the phylogeny.

\begin{figure}[h]
\centering
\includegraphics[scale=0.3]{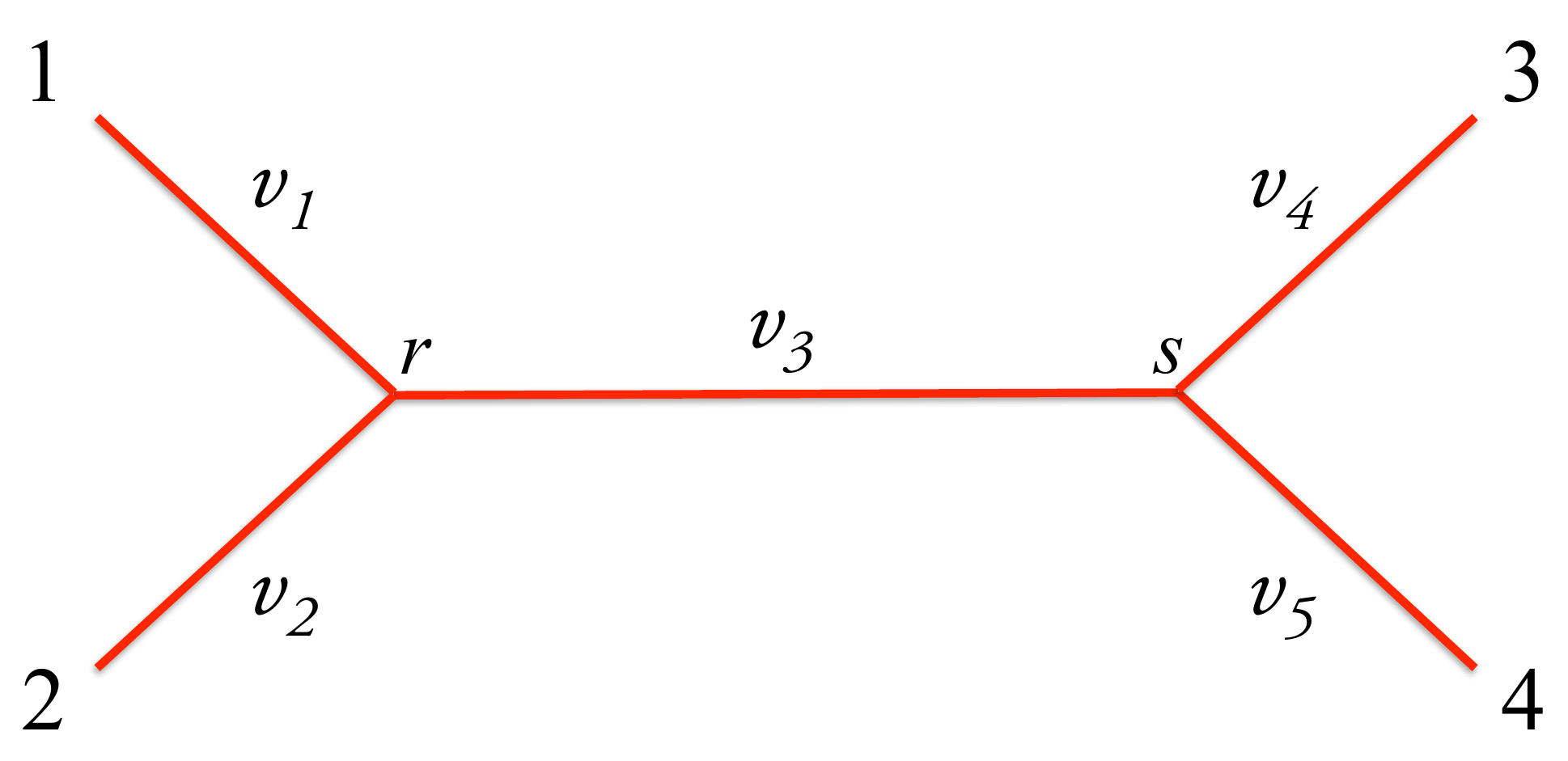}
\caption{Four-leaf unrooted gene tree.}
\label{4leafgenetree}
\end{figure}

Consider a gene tree as labeled in Figure \ref{4leafgenetree} with branch lengths $v_i$. Let $X_y$ be the state observed for lineage $y$, $y = 1, 2, 3, 4$. With $\mathbf{P}^{\rho_i} (v_i)$ denoting the substitution matrix for edge $v_i$ and considering the tree to be rooted at $r$, the {\itshape site pattern probability} on the gene tree for a particular observation $i_1 i_2 i_ 3 i_4$, $i_j \in [\kappa]$, at the tips of the tree with associated site-specific rate $\rho_i$ is given by
\begin{eqnarray*}
\label{eq:site.pattern.prob.rates}
P(X_1 = i_1, X_2 = i_2, X_3 = i_3, X_4 = i_4) 
 =  \sum_{r=1}^{\kappa} \sum_{s=1}^{\kappa} \pi_r \mathbf{P}^{\rho_i}_{r i_1}(v_1) \mathbf{P}^{\rho_i}_{r i_2}(v_2) \mathbf{P}^{\rho_i}_{rs}(v_3) \mathbf{P}^{\rho_i}_{s i_3} (v_4)\mathbf{P}^{\rho_i}_{s i_4}(v_5).
 \end{eqnarray*}
For the gene tree in Figure \ref{histA}, we define lineage 1 to be $A$, 2 to be $B$, 3 to be $C$, and 4 to be $D$, and we then have $v_1= \tau_2 + t_1, v_2 = \tau_2 + t_1 , v_3 = 2t_3-t_2+(\tau_3-t_1-\tau_2) , v_4=\tau_3+t_2$, and  $v_5 =\tau_3+t_2 $.  For the gene tree in Figure \ref{histB}, we define lineage 1 to be $A$, 2 to be $C$, 3 to be $B$, and 4 to be $D$, and we have $v_1 = v_2 = \tau_3+t_1$, $v_3 = t_2-t_1$, $v_4=\tau_3+t_2$, and $v_5 = \tau_3+2t_3-t_2$.  In general, lineages 1 and 2 are defined to be those on either side of the root, lineage 3 is the next most closely related to these in the gene tree embedded within the species tree, and the remaining lineage is assigned label 4. 

For each category $i \in [m]$ we will denote by $p^{\rho_i}_{i_1 i_2 i_3 i_4 |(G,\mathbf{t})}$ the site pattern probability on the gene tree  $(G,\mathbf{t})$. Since the rates $\rho$ are unobserved in practice, we sum over them with respect to discrete gamma distribution to obtain,
\begin{equation}
\label{eq:site.pattern.prob.rates2}
p^{\Gamma}_{i_1 i_2 i_3 i_4 |(G,\mathbf{t})} = \frac{1}{m}\sum_{i=1}^{m} p^{\rho_i}_{i_1 i_2 i_3 i_4 |(G,\mathbf{t})}
\end{equation}
(see \cite{yang1994}).
Define the 4-dimensional $\kappa \times \kappa \times \kappa \times \kappa$ array  $P^{\Gamma}_{(G,\mathbf{t})}$ as the probability distribution on all possible site patterns for the gene tree $(G,\mathbf{t})$ under a model that includes rate variation across sites. The dimensions of $P^{\Gamma}_{(G,\mathbf{t})}$  correspond to the ordered taxa of the gene tree and the entries are indexed by the states at the leaves.

 To incorporate the possibility of invariable sites, define $P^{I} = diag(\boldsymbol{\pi})$ to be a 4-dimensional array whose off-diagonal entries are 0 and whose diagonal entries are the elements of the vector $\mathbf{\pi} = (\pi_1, \cdots, \pi_{\kappa})$.  Then, the probability distribution $P_{(G,\mathbf{t})}$ in all sites patterns is given by 
 \begin{equation}
 \label{eq:site.pattern.prob}
 P_{(G,\mathbf{t})} = (1-\delta) P^{\Gamma}_{(G,\mathbf{t})} + \delta P^{I}.
 \end{equation}
The entries of $P_{(G,\mathbf{t})}$ will be denoted by $p_{i_1 i_2 i_3 i_4 |(G,\mathbf{t})}$ and we emphasize that each entry is an analytic function.

\section{The Site Pattern Probability Distribution Under the Coalescent Model}\label{model} %
Our goal is to establish identifiability of the $n$-leaf phylogenetic species tree given data at the leaves of the tree. To achieve this, we first prove identifiability of the 4-taxon species tree, and thus dedicate this section to the description of the coalescent model for a fixed species tree $(S,\boldsymbol{\tau})$ on four leaves with a set of taxa labeled by $a,b,c$ and $d$. We point out that everything in this section is easily extendable to the $n$-taxon case. 

To compute the site pattern probabilities on the species tree, we will use gene tree site pattern probabilities along with the gene tree density given in Equation \ref{gtd}.
To distinguish site pattern probabilities arising on the species tree under the coalescent model from those arising on a gene tree, we will use the notation 
$p^*_{i_1 i_2 i_3 i_4 | ((G,\mathbf{t}), (S, \boldsymbol{\tau}))}, i_j \in [\kappa] $ to denote site pattern probabilities that come from the species phylogeny $(S, \boldsymbol{\tau})$ with embedded gene tree $(G,\mathbf{t})$.

Consider first the case in which $G=S$, i.e., the gene tree and species tree have the same topology (see, for example, Figure \ref{histA}). In that case, for a fixed observation $i_1, i_2, i_3,$ and $i_4$ at the leaves of a species tree, we find
\begin{eqnarray*}
\label{eq:site.pattern.prob.gene.coal}
p^*_{i_1 i_2 i_3 i_4 | ((G,\mathbf{t}), (S, \boldsymbol{\tau}))} & = & P(X_a = i_1, X_b = i_2, X_c = i_3, X_d = i_4) \\
& = & \int_{\mathbf{t}} p_{i_1 i_2 i_3 i_4 |(G,\mathbf{t})} f((G,\mathbf{t})|(S,\boldsymbol{\tau})) d\mathbf{t}.
\end{eqnarray*}

When $G \neq S$, the labels at the leaves of the species tree may correspond to a different ordering of the labels of the lineages in the gene tree.  For example, in Figure \ref{histB}, species $a$ and $b$ are sister leaves in the species tree, but lineages $A$ and $C$ are sisters in the gene tree. Hence, for a gene tree $(G,\mathbf{t})$ and species tree $(S, \boldsymbol{\tau})$ in Figure \ref{histB}, we compute
\begin{eqnarray*}
p^*_{i_1 i_2 i_3 i_4 | ((G,\mathbf{t}), (S, \boldsymbol{\tau}))} & = & P(X_a = i_1, X_b = i_2, X_c = i_3, X_d = i_4) \\
& = & \int_{\mathbf{t}} p_{i_1 i_3 i_2 i_4 |(G,\mathbf{t})} f((G,\mathbf{t})|(S,\boldsymbol{\tau}))d\mathbf{t}.
\end{eqnarray*}

It is clear that many observations at the leaves of the species tree will not always result in the same observations at the tips of the embedded gene tree. We will use $\sigma(i_1, i_2, i_3, i_4)$ to represent the observation at the leaves of the gene tree, where $\sigma$ is a permutation of $i_1, i_2, i_3,$ and $i_4$ for a fixed gene tree topology nested within a species tree. In general we write 
\begin{equation}
\label{eq:site.pattern.prob.gene.coal2}
p^*_{i_1 i_2 i_3 i_4 | ((G,\mathbf{t}), (S, \boldsymbol{\tau}))} 
 =  \int_{\mathbf{t}} p_{\sigma(i_1, i_2, i_3, i_4) |(G,\mathbf{t})} f((G,\mathbf{t})|(S,\boldsymbol{\tau})) d\mathbf{t}.
\end{equation}

Equation \ref{eq:site.pattern.prob.gene.coal2} gives the probability of the event $\{ X_a = i_1, X_b = i_2, X_c = i_3, X_d = i_4 \}$ for a fixed gene tree $(G,\mathbf{t})$  and species tree $(S, \boldsymbol{\tau})$. Our goal is to obtain these probabilities conditioning only on the species tree $(S,\boldsymbol{\tau})$. Since the true gene tree is unobserved, we must consider all possible gene trees that are consistent with the given species tree, and weight the probability of the site pattern of interest appropriately by the probability of each gene tree under the coalescent model. We note that the limits of integration will depend on where the coalescent events happen. In addition, recall that $\mathcal{G}_S$ denotes the set of all gene trees $(G,\mathbf{t})$ conditional on species tree $(S,\boldsymbol{\tau})$.
This leads to the following expression for the probability that $\{ X_a = i_1, X_b = i_2, X_c = i_3, X_d = i_4 \}$ for species tree $(S, \boldsymbol{\tau})$,
\begin{align}
\label{eq:site.pattern.prob.coal}
\nonumber p^*_{i_1 i_2 i_3 i_4 | (S, \boldsymbol{\tau})} &= \sum_{(G,\mathbf{t}) \in \mathcal{G}_S} p^*_{i_1 i_2 i_3 i_4 | ((G,\mathbf{t}), (S, \boldsymbol{\tau}))} \\ 
&=\sum_{(G,\mathbf{t}) \in \mathcal{G}_S} \int_{\mathbf{t}} p_{\sigma(i_1 i_2 i_3 i_4 ) |(G,\mathbf{t})} f((G,\mathbf{t})|(S,\boldsymbol{\tau})) d\mathbf{t}.
\end{align}

We denote the probability distribution on all possible site patterns given species tree $(S,\boldsymbol{\tau})$ by the 4-dimensional $\kappa \times \kappa \times \kappa \times \kappa$ array $P^*_{(S,\boldsymbol{\tau})}$ and let  $p^*_{i_1 i_2 i_3 i_4 |(S,\boldsymbol{\tau})}, i_j \in [\kappa] $ denote the site pattern probabilities as given in Equation \ref{eq:site.pattern.prob.coal}. We stress again that in general each $ p^*_{i_1 i_2 i_3 i_4 |(S,\boldsymbol{\tau})} $ is not a polynomial but an {\em analytic function}. 

Let $\mathcal{C}_{GTR(\kappa)}$ be the $\kappa$-state analytic GTR+I+$\Gamma$ model 
under the coalescent for a 4-leaf species tree with the discrete gamma distribution used to model rate variation across sites. The model parameters are the topology of the species tree $S$, the matrix $\mathbf{Q}$ as described by (\ref{qmatrix}), the vector of speciation times $\boldsymbol{\tau}=(\tau_1, \tau_2, \tau_3)$ for $S$, the effective population size $\theta$, the proportion of invariable sites $\delta$, and the parameters of the discrete gamma distribution $\alpha$ and $\beta$.
For a fixed species tree $S$ we will denote the continuous parameter space of $\mathcal{C}_{GTR(\kappa)}$  by an open set $U_{S} \subseteq \mathbb{R}^M$. Then the parameterization map $\psi_{S}$, for the analytic model, giving the probability distribution of the variables at the leaves of the 4-taxon species tree $S$, is 
\begin{align}\label{map}
\psi_{S}: U_{S}  &\to \Delta^{\kappa^4-1} \\
u &\mapsto P^*_{(S,\boldsymbol{\tau})},\nonumber
\end{align}
where $\Delta^{\kappa^4-1} \subseteq [0,1]^{\kappa^4}$ is the probability simplex $$\Delta^{\kappa^4-1} :=\{(p_1,p_2,\dots,p_{\kappa^4}) \in \mathbb{R}^{\kappa^4} | \sum_{i=1}^{\kappa^4} p_i=1 \ \text{and} \  p_i \geq 0 \ \text{for all} \  i\}.$$
The image of the map, $\mathcal{C}_S := \psi_S(U)\subseteq \Delta^{\kappa^4-1}$, is the  {\em coalescent phylogenetic model}. One can easily see that the model can be extended to any $n$-taxon species tree. 
 
 \subsection{A few remarks about embedded gene trees within a species tree}\label{pairing}
Let $S$ be a four-taxon symmetric $((a,b),(c,d))$ or asymmetric $(a,(b,(c,d)))$ species tree with a cherry $(c,d)$ as in Figure \ref{hist}. Recalling that the gene trees arising under the coalescent model will all satisfy the molecular clock (i.e., the distance from each tip to the root is identical), it might be obvious to some readers that for any observation $i_1 i_2 i_3 i_4$ at the leaves of $(S, \boldsymbol{\tau})$, $i_1, i_2, i_3, i_4 \in [\kappa]$ 
\begin{equation}\label{eq:cherry}
p^*_{i_1 i_2 i_3 i_4 | (S,\boldsymbol{\tau})} = p^*_{i_1 i_2 i_4 i_3 |(S,\boldsymbol{\tau})}.
\end{equation}
However, if the above relationship is conditioned on a particular gene tree for which $(C,D)$ is not a cherry, then the equation will be false, as the example below shows.

\begin{exmp}
Let the gene tree $G$ arising from the species tree $(a,(b,(c,d)))$ be $(((A,C),B),D)$, as in Figure \ref{histB}. Then according to Equation \ref{eq:site.pattern.prob.gene.coal2} for any observation $i_1 i_2 i_3 i_4$ at the leaves of $(S, \boldsymbol{\tau})$ we get
\begin{equation*}
p^*_{i_1 i_2 i_3 i_4 | ((G, \mathbf{t}), (S, \boldsymbol{\tau}))} = \int_{\mathbf{t}} p_{i_1 i_3 i_2 i_4 |(G,\mathbf{t})} f((G,\mathbf{t})|(S,\boldsymbol{\tau})) d\mathbf{t},									
\end{equation*}
\begin{equation*}
p^*_{i_1 i_2 i_4 i_3 | ((G, \mathbf{t}), (S, \boldsymbol{\tau}))} 	= \int_{\mathbf{t}} p_{i_1 i_4 i_2 i_3 |(G,\mathbf{t})} f((G,\mathbf{t})|(S,\boldsymbol{\tau})) d\mathbf{t}.								
\end{equation*}
Since $p_{i_1 i_3 i_2 i_4 |(G,\mathbf{t})} \neq p_{i_1 i_4 i_2 i_3 |(G,\mathbf{t})}$ (see Equations \ref{eq:site.pattern.prob.rates}-\ref{eq:site.pattern.prob} and Figure \ref{hist}), we see that $p^*_{i_1 i_2 i_3 i_4 | ((G, \mathbf{t}), (S, \boldsymbol{\tau}))} 	\neq p^*_{i_1 i_2 i_4 i_3 | ((G, \mathbf{t}), (S, \boldsymbol{\tau}))} $.
\end{exmp}

This example demonstrates that in general for many individual gene trees an equation analogous to Equation  \ref{eq:cherry} does not hold, yet when the sum is taken over all possible gene trees $(G,\mathbf{t})$ conditional on species tree $(S,\boldsymbol{\tau})$ the result does hold. In this section we would like to clarify why and how the above Equation \ref{eq:cherry} is true under the coalescent model. 
\vspace{0.1in}

\noindent{\bfseries Case 1:}  Consider a collection of gene trees with $(C,D)$ as a cherry, denoted by $\mathcal{G}_{CD}$. Then for each $(G, \mathbf{t}) \in \mathcal{G}_{CD}$ it is true that
\begin{equation*}
 p^*_{i_1 i_2 i_3 i_4 | ((G, \mathbf{t}),(S,\boldsymbol{\tau}))} = p^*_{i_1 i_2 i_4 i_3 |((G, \mathbf{t}),(S,\boldsymbol{\tau}))},
 \end{equation*}
for all $i_1, i_2, i_3, i_4 \in [\kappa]$.
\vspace{0.1in}

\noindent{\bfseries Case 2:} Now consider the set of gene trees for which $(A,B)$ is a cherry and $(C,D)$ is not a cherry. There are only two topological gene trees of this form:  $G_1 = (((A,B),C),D)$ and $G_2 = (((A,B),D),C)$.  Since $(A,B)$ is a cherry, the coalescent event joining $A$ and $B$ can happen in population $P_2$ or $P_3$ for the symmetric tree in Figure \ref{histA}, while for the asymmetric species tree in Figure \ref{histB} $A$ and $B$ can only coalesce  in population $P_3$. This is true for both gene trees, $G_1$ and $G_2$.  
Let $\mathbf{t}$ represent a vector of coalescent times for which $A$ and $B$ coalesce in the same population for both gene trees, whereas the other two coalescent events necessarily happen in population $P_3$.
Now the probability densities for these two gene trees, $(G_1,\mathbf{t})$ and $(G_2,\mathbf{t})$, are identical under the coalescent model, i.e. $f((G_1,\mathbf{t})|(S,\boldsymbol{\tau})) = f((G_2,\mathbf{t})|(S,\boldsymbol{\tau}))$.
Then, for all $i_1, i_2, i_3, i_4 \in [\kappa]$ 
\begin{eqnarray}
\label{eq:case21} p^*_{i_1 i_2 i_3 i_4 |((G_1,\mathbf{t}),(S,\boldsymbol{\tau}))} & =  \int_{\mathbf{t}} p_{i_1 i_2 i_3 i_4|(G_1,\mathbf{t})} f((G_1,\mathbf{t})|(S,\boldsymbol{\tau})) d\mathbf{t},\\
\label{eq:case22} p^*_{i_1 i_2 i_4 i_3 |((G_{2}, \mathbf{t}),(S,\boldsymbol{\tau}))} & =  \int_{\mathbf{t}} p_{i_1 i_2 i_4 i_3|(G_2,\mathbf{t})} f((G_{2},\mathbf{t})|(S,\boldsymbol{\tau})) d\mathbf{t}.
\end{eqnarray}

\noindent Now, note that 
\begin{equation}\label{eq:case23}
p_{i_1 i_2 i_3 i_4|(G_{u},\mathbf{t})} = p_{i_1 i_2 i_4 i_3|(G_{v},\mathbf{t})}, 
\end{equation}
for $u, v \in \{1, 2\}, u \neq v$, which gives
\begin{equation}\label{eq.matchedhistories}
p^*_{i_1 i_2 i_3 i_4 |((G_1,\mathbf{t}),(S,\boldsymbol{\tau}))} + p^*_{ i_1 i_2 i_3 i_4 |((G_{2}, \mathbf{t}),(S,\boldsymbol{\tau}))} = p^*_{i_1 i_2 i_4 i_3 |((G_1,\mathbf{t}),(S,\boldsymbol{\tau}))} + p^*_{i_1 i_2 i_4 i_3 |((G_{2}, \mathbf{t}),(S,\boldsymbol{\tau}))}. 
\end{equation}

\noindent{\bfseries Case 3:} Consider the set of gene trees for which neither $(A,B)$ nor $(C,D)$ are cherries. Then at least one of the following pairs are cherries: $$\{ (A,C), (A,D), (B,C), (B,D) \}.$$ First, suppose that $(A,C)$ is a cherry (the case that $(A,D)$ is a cherry is analogous).
\vspace{0.07in}

\noindent {(i).} First, suppose that $(B,D)$ is also a cherry, and denote this gene tree by $G_1 = ((A,C),(B,D))$. For the symmetric species tree all coalescent events happen in population $P_3$ for this gene tree (Figure \ref{histA}). Let $t_1$ be the coalescent time for lineages $A$ and $C$, and $t_2$ be the coalescent time for lineages $B$ and $D$. One can see that there are two metric gene trees, $(G_1, \mathbf{t})_{t_1 > t_2}$ and $(G_1, \mathbf{t})_{t_1 < t_2}$, where $\mathbf{t}$ is the vector of coalescent times. For the asymmetric species tree, (Figure \ref{histB}) there are three metric gene trees: one in which lineages $B$ and $D$ coalesce in population $P_2$, and the other two (analogous to the symmetric case) in which all coalescent events happen in $P_3$.

Now consider a second gene tree, $G_2$, in which $(A,D)$ and $(B,C)$ are cherries. Similar to our discussion above, one can show that if the species tree is symmetric there are two metric gene trees, while for the asymmetric species tree there are three such pairs. Let $\mathbf{t}$ be a vector of coalescent times consistent with both trees, $G_1$ and $G_2$. For example, if $(S,\boldsymbol{\tau})$ is the asymmetric tree and the lineages $B$ and $D$ coalesce in $P_2$ for $G_1$ then lineages $B$ and $C$ also coalesce in $P_2$ for $G_2$, and the other two events occur in population $P_3$. As in Case 2, the probability densities of $(G_1, \mathbf{t})$ and $(G_2, \mathbf{t})$ are identical.  Then the relationship in Equation \ref{eq.matchedhistories} holds.

\vspace{0.07in}

\noindent {(ii).} Now suppose that  $(B,D)$ is not a cherry. Then two gene trees are possible: $((A,C),D),B)$ and $((A,C),B),D)$. Let $G_1$ be the tree $((A,C),D),B)$ (the other case is analogous), and consider the tree $G_2 = ((A,D),C),B)$. For both the symmetric and the asymmetric species trees, $G_1$ and $G_2$ have all coalescent events happening in $P_3$ and thus result only in one gene tree pair  $(G_i, \mathbf{t})$ for each $i \in \{1,2\}$. The same arguments as in the previous cases can be used to show that Equation \ref{eq.matchedhistories} holds in this case as well. Finally, we note that the cases for which $(B,C)$, and analogously $(B,D)$, are cherries can be handled in exactly the same manner as above. This concludes Case 3.

\vspace{0.1in}

It is then straightforward to check that all of the possible histories for both the symmetric and asymmetric species trees are included in exactly one of the cases described above. Now, recall that the probability of any observation $i_1 i_2 i_3 i_4$ at the leaves of $S$ is the sum over all possible gene trees. Thus, Equation \ref{eq:cherry} is true under the coalescent model. If the species tree $(S,\boldsymbol{\tau})$ is symmetric then in addition we also get 
$$p^*_{i_1 i_2 i_3 i_4 | (S,\boldsymbol{\tau})} = p^*_{i_2 i_1 i_3 i_4 |(S,\boldsymbol{\tau})},$$ for any observation $i_1 i_2 i_3 i_4$ at the leaves of $S$, $i_1, i_2, i_3, i_4 \in [\kappa]$.

\section{Splits, Flattenings, Invariants and Identifiability Background}\label{background}
In this section we provide definitions and results that we will use in subsequent sections. Let $\mathcal{L}$ denote a set of taxa for a tree (species or gene), e.g. if $S$ is a species tree in Figure \ref{hist}, then $\mathcal{L} = \{a,b,c,d\}$.
\begin{defn}
A {\bf split} of a set of taxa $\mathcal{L}$ is a bipartition of $\mathcal{L}$ into two non-empty subsets $L_1$ and $L_2$, denoted $L_1|L_2$. A split $L_1|L_2$ is {\bf valid} for tree $T$ if the subtrees containing the taxa in $L_1$ and in $L_2$  do not intersect. 
\end{defn}

In general, the definition of a split is used for unrooted trees. To avoid any ambiguity, by a split of a rooted 4-taxon species tree (symmetric or asymmetric) we will always mean a bipartition of $\mathcal{L}$ into two equal subsets, i.e. $|L_1|=|L_2|=2$. 

In particular, for the  four-leaf species trees $((a,b),(c,d))$ or $(a,(b,(c,d)))$ (Figure \ref{hist}) there are three splits according to our discussion above. The split $L_1|L_2 = ab|cd$ is valid and splits $ ac|bd$ and $ ad|bc$ are not valid. For the gene tree in Figure \ref{histB} the split $AC|BD$ is valid, while $AB|CD$ and $AD|BC$ are not. This also demonstrates that a valid split for a species tree will not always result in the same valid split for an embedded gene tree.
Splits of a set of taxa provide a natural way to rearrange the $n$-dimensional $\kappa \times \cdots \times \kappa$ array $P$ as a matrix.
\begin{defn}
Let $L_1|L_2$ be any split of a set of taxa $\mathcal{L}$. A \textbf{flattening} of $P$, denoted by $\text{Flat}_{L_1|L_2}(P)$, is an $\kappa^{|L_1|} \times \kappa^{|L_2|}$ matrix,  whose rows are indexed by possible states for the leaves in $L_1$ and columns by possible states in $L_2$. The entries of $\text{Flat}_{L_1|L_2}(P)$ correspond to the probability of the site pattern specified by the row and column indices. 
\end{defn}
For example, for $\kappa=4$ and a species tree $S$ as in Figure \ref{hist}, the $16 \times 16$ flattening of $P^*_{(S,\boldsymbol{\tau})}$ for a split $L_1|L_2 = ad|bc$ is
\begin{center}
{\small 
\begin{gather*}
{\text{Flat}_{ad|bc}(P^*_{(S,\boldsymbol{\tau})})} =
\begin{pmatrix}
p^*_{AAAA}	& p^*_{AACA} 	& p^*_{AAGA} 	& p^*_{AATA} 	& p^*_{ACAA} 	& \cdots 	& p^*_{ATTA} \\
p^*_{AAAC}  	& p^*_{AACC} 	& p^*_{AAGC} 	& p^*_{AATC} 	& p^*_{ACAC} 	& \cdots 	& p^*_{ATTC} \\
p^*_{AAAG} 	& p^*_{AACG} 	& p^*_{AAGG} 	& p^*_{AATG} 	& p^*_{ACAG} 	& \cdots 	& p^*_{ATTG} \\
p^*_{AAAT} 	& p^*_{AACT} 	& p^*_{AAGT} 	& p^*_{AATT} 	& p^*_{ACAT} 	& \cdots 	& p^*_{ATTT} \\
p^*_{CAAA} 	& p^*_{CACA} 	& p^*_{CAGA} 	& p^*_{CATA} 	& p^*_{CCAA} 	& \cdots 	& p^*_{CTTA} \\
\vdots  		& \vdots 		& \vdots  		& \vdots  		& \vdots  		& \ddots 	& \vdots  \\
p^*_{TAAT} 	& p^*_{TACT} 	& p^*_{TAGT} 	& p^*_{TATT} 	& p^*_{TCAT} 	& \cdots 	& p^*_{TTTT} 
\end{pmatrix}.
\end{gather*}
}
\end{center}
\vspace{0.1in}

A very important concept that will be tied with the minors of the matrix $\text{Flat}_{L_1|L_2}(P)$ is that of {\em invariants}. Consider an asymmetric species tree $S=(a,(b,(c,d)))$, where $c$ and $d$ are sister leaves. Then, $$ p^*_{\star \star i j |(S,\boldsymbol{\tau})} -p^*_{\star \star j i |(S,\boldsymbol{\tau})} = 0,$$ for all $i,j \in [\kappa]$, where $\star$ indicates any value in $[\kappa]$ (see Section \ref{pairing}). This is an example of a {\em linear invariant}. It is termed {\em linear} since it is a linear function in the site pattern probabilities.  More precisely, an {\em invariant} is a function in the site pattern probabilities that vanishes when evaluated on any distribution arising from the model. Linear invariants for gene trees have been studied by several authors \cite{Lake1987,Cavender1989,Fu1992,Fu1995}. In general, linear invariants for a gene tree are not necessarily invariants for the same species tree. Intuitively this makes sense, since when the sum is taken over all gene trees embedded within a species tree it is quite obvious that some linear invariants will vanish on one gene tree topology but not the other.

Let $\mathcal{R}$ be a collection of analytic functions $f_1, f_2, \dots , f_r$ defined on the common domain $D \subseteq \mathbb{R}^n$. An {\em analytic variety}  $V(\mathcal{R})$ is a simultaneous zero-set of functions in $\mathcal{R}$,
\begin{equation*}
V(\mathcal{R})=\{\mathbf{a} \in D | f_i(\mathbf{a})=0, 1 \leq i \leq r\}.
\end{equation*}

Fix a coalescent phylogenetic model $\mathcal{C}_S \subseteq \Delta^{\kappa^4-1}$ as defined in Section \ref{model}. Then an analytic function $f$ is called a {\em coalescent phylogenetic invariant} if $f(\mathbf{a})=0$ for all $\mathbf{a} \in \mathcal{C}_S$.   

A very powerful and important concept about the zero-sets of analytic varieties will play an important role in our arguments about identifiability. It is well-known that if a real function $f$ is analytic on the connected open set $U \subseteq \mathbb{R}^n$ and its zero set is of positive measure, then $f \equiv 0$. Let a set $U \subseteq \mathbb{R}^n$ be of full dimension and set $W:=U \cap V(\mathcal{R})$, where $\mathcal{R}$ is a finite set of nonconstant analytic functions defined on $U$. If $W$ is a proper subvariety  of $U$ then it must be of dimension strictly less than $n$. In particular, $W$ is necessarily of Lebesgue measure zero. 

A key issue when formulating any statistical model is whether the parameters of that model are identifiable. Classically, identifiability means the following: suppose that $\mathcal{M}(\Theta) = \{P_{\theta}: \theta \in \Theta \}$ defines a family  of probability distributions with parameters $\theta$.  The model $\mathcal{M}(\Theta)$ is said to be {\em identifiable} if the mapping $\theta \mapsto P_{\theta}$ is injective. This definition is rather too strict as for many models this map will not be injective. Describing a set of parameters on which a model is non-identifiable and excluding it from a domain is not always possible, especially for identifiability of numerical parameters.  Nonetheless, identifiability for parametric models, both algebraic and analytic, can be proved by demonstrating that the set of parameters on which the model is non-identifiable is a subset of a proper subvariety of measure zero. In this case we say that model parameters are {\em generically identifiable}. 

Generic identifiability of a gene tree topology and numerical parameters for many evolutionary models is well-studied and established. A series of papers by E. Allman and J. Rhodes and collaborators \cite{allmanetal2008,allmanrhodes2008b,allmanrhodes2009,allmanetal2010b} have used an algebraic framework to establish identifiability of the unrooted gene tree and associated model parameters for substitution models as general as the General Markov model with a proportion of sites invariable (GM+I) as well as for the general time reversible model with rate variation following the gamma distribution (GTR+$\Gamma$).

For the general $\kappa$-state Markov model $\mathcal{M}$ a flattening of a tensor $P$ is used to prove generic identifiability of a gene tree topology. Briefly, let $P$ be a joint distribution arising from $\mathcal{M}_T$ on an $n$-leaf gene tree $T$ then: (i) if $L_1|L_2$ is a valid split for $T$, rank$(\text{Flat}_{L_1|L_2}(P)) \leq \kappa$, and (ii) if $L_1|L_2$ is not a valid split for $T$, {\em generically} rank$(\text{Flat}_{L_1|L_2}(P)) > \kappa$. In particular, for a valid split on $T$ the $(\kappa+1)$-minors of the matrix $\text{Flat}_{L_1|L_2}(P)$, called {\em edge invariants}, all vanish on $\mathcal{M_T}$. 

\section{Main Results}

To prove identifiability of a 4-leaf species tree from site pattern probabilities we will use precise formulas for the generalized Jukes-Cantor $k$-state model under the coalescent for symmetric and asymmetric 4-leaf species trees, which are described fully in Section \ref{JC69}. In Theorem \ref{thm:quartets}, we establish an identifiability result for unrooted 4-leaf species trees (note that the rooted symmetric and the rooted asymmetric species trees yield a single unrooted species tree topology), making a distinction between symmetric and asymmetric trees when necessary.  Additional arguments extend this result to $n$-leaf species trees, as stated in Corollary \ref{cor:bigtrees}.

In the next theorem we are going to identify the parameter space $U_{S}$ with a full dimensional subset of $\mathbb{R}^M$.
In particular, $U_S$ is an open connected subset of $\mathbb{R}^M$.  Also, recall that the parameterization map $\psi_{S}$ defined by (\ref{map}) is given by analytic functions. In Section  \ref{background} we have stated a result from the theory of analytic functions of several complex variables that will play an important role in our next argument. 
For clarity we rephrase this fact as follows: {\em if a real function $f$ is analytic on the connected open set $U \subseteq \mathbb{R}^n$ and $f$ is not identically zero then the set $Z(f) = \{x \in U | f(x) = 0\}$ has $n$-dimensional Lebesgue measure zero.} For more information about analytic functions of several complex variables and the proof of the above fact, the reader is encouraged to consult \cite{gunning2009analytic}.

We will need a few additional results to prove Theorem \ref{thm:quartets}. Let $A = (a_{ij})$ be an $n \times n$ real symmetric matrix. The well-known {\em Sylvester's criterion} states that  $A$ is positive definite if and only if all of its principal minors are positive. Positive definite matrices are invertible, the sum of positive definite matrices is positive definite, and multiplication of a positive definite matrix by a real positive number is also positive definite. In addition, recall that a matrix $A$ is {\em strictly diagonally dominant} if $|a_{ii}| > \sum_{i \neq j} |a_{ij}|$ for all $i$. A symmetric strictly diagonally dominant matrix $A$ with real non-negative diagonal entries is positive definite. 

Next, we would like to make several observations about the discrete gamma distribution that we will use in the proof of Theorem \ref{thm:quartets}. From the definition of the rates $\rho_i$ (see Equation (\ref{rho})) one can see that $0 <\rho_1 < \rho_2 < \cdots < \rho_m$, where $m$ is the number of categories. For any fixed $m$ and $\alpha =1$ we show that we can always find $\beta>0$ such that $\rho_m < n$ for some positive number $n$. With $\alpha =1$ the rate $\rho_m$ becomes
\begin{equation}\label{rho_m}
\rho_m = \frac{\int_a^\infty \rho \cdot \beta \cdot e^{-\rho  \beta} d \rho}{\int_a^\infty \beta \cdot e^{-\rho  \beta} d \rho}.
\end{equation}
Recall that each category has equal probability, hence the left-most endpoint of the last category corresponds to a cumulative probability of $(m-1)/m$. Now we compute the left-most endpoint $a$ for the last category:
\begin{align*}
\int_0^a \beta \cdot e^{-\rho  \beta} d\rho &= \frac{m-1}{m}\\
1-e^{-a  \beta}&= \frac{m-1}{m}\\
a &= \frac{\ln(m)}{\beta}.
\end{align*}
Using this endpoint and substituting into Equation (\ref{rho_m}) we find that 
\begin{equation*}
\rho_m = \frac{1+\ln(m)}{\beta}.
\end{equation*}
Therefore, for any fixed number of categories $m$ and $\alpha=1$ if $\beta > \frac{1+\ln(m)}{n}$, for some positive real number $n$, then $\rho_m <n$. In particular, all rates $\rho_i$ are bounded by $n$. Of course, in practice such a choice of parameters may not be biologically meaningful, but for our proof we only need the existence of a single point to show generic identifiability.

\begin{thm}\label{thm:quartets}
Let $S$ be a four-taxon symmetric $((a,b),(c,d))$ or asymmetric $(a,(b,(c,d)))$ species tree with a cherry $(c,d)$ and let $L_1 | L_2$ be one of the splits of taxa, $ab|cd$, $ac|bd$ or $ad|bc$.  
Consider the  $\mathcal{C}_{GTR(\kappa)}$ $\kappa$-state analytic \textup{GTR+I+$\Gamma$} model under the coalescent for a species tree $S$ with the discrete gamma distribution used to model rate variation across sites. 
Identify the parameter space $U_{S}$ with a full dimensional subset of $\mathbb{R}^M$. 
\begin{enumerate}
\item If $L_1 | L_2$ is a valid split for $S$, then for all distributions $P^*_{(S,\boldsymbol{\tau})}$ arising from the model 
$$\textup{rank}(\textup{Flat}_{L_1 | L_2}(P^*_{(S,\boldsymbol{\tau})})) \leq \binom{\kappa+1}{2}.$$
\item If $L_1 | L_2$ is not a valid split for $S$, then for generic distributions $P^*_{(S,\boldsymbol{\tau})}$ arising from the model 
$$\textup{rank}(\textup{Flat}_{L_1 | L_2}(P^*_{(S,\boldsymbol{\tau})})) > \binom{\kappa+1}{2}.$$
\end{enumerate}
\end{thm}
\begin{proof}
Suppose that $L_1 | L_2$ is a valid split for $S$, that is $L_1 | L_2 = ab|cd$. From our discussion in Section \ref{pairing} it is clear that for any distribution $P^*_{(S,\boldsymbol{\tau})}$ arising from the $\kappa$-state model $\mathcal{C}_{GTR(\kappa)}$ the entries $((i,j),(k,l))$ and $((i,j),(l,k))$ of the $\kappa^2 \times \kappa^2$ matrix $\text{Flat}_{ab | cd}(P^*_{(S,\boldsymbol{\tau})})$ are equal, for all $k \neq l \in [\kappa]$. In particular, the columns labeled by the $cd$-indices $kl$ and $lk$ for distinct $k , l \in [\kappa]$ are identical, e.g. columns labeled by 12 and 21 are equal.  Since there are $\binom{\kappa}{2}$ such pairs, we get that
$$\text{rank}(\text{Flat}_{ab | cd}(P^*_{(S,\boldsymbol{\tau})})) \leq \kappa^2-\binom{\kappa}{2} = \binom{\kappa+1}{2},$$
which establishes (1).

Now suppose that $L_1 | L_2$ is not a valid split for $S$.  Let $X_{ac|bd}$ and $X_{ad|bc}$  denote the sets of $(\binom{\kappa+1}{2} +1)$-minors of the  $\kappa^2 \times \kappa^2$ matrices  $\text{Flat}_{ac | bd}$ and $\text{Flat}_{ad | bc}$, respectively.  
Let  $V_{ac|bd}$ be the zero set of $X_{ac|bd}$ and $V_{ad|bc}$ be the zero set of $X_{ad|bc}$.  

Define $W$ as a zero set of real functions $f \circ \psi_{S}$ that are analytic on the connected set $U_S$, where $f$ vanishes on $V_{ac|bd} \cup V_{ad|bc}$. It is clear that $W$ is an analytic subvariety and  a subset of $U_{S}$. We are going to show that $W$ is a proper subvariety of $U_{S}$ by finding $u \in U_{S} \setminus W$ with $\psi_S(u) \notin V_{ac|bd} \cup V_{ad|bc}$. In particular, we show that the ranks of $\text{Flat}_{ac | bd}$ and $\text{Flat}_{ad | bc}$ are greater than $\binom{\kappa+1}{2}$ for this choice of parameters.

Recall that the continuous model parameters are the vector of speciation times $\boldsymbol{\tau}=(\tau_1, \tau_2, \tau_3)$ for $S$ as depicted in Figure \ref{hist}, the effective population size $\theta$, the matrix $\mathbf{Q}$ (the rates of mutation between states $\boldsymbol{\mu}=(\mu_1, \mu_2, \dots , \mu_{\frac{1}{2}\kappa(\kappa-1)})$ and relative frequencies of the states at equilibrium $\boldsymbol{\pi}=(\pi_1, \pi_2, \dots, \pi_\kappa)$), the proportion of invariable sites $\delta$, and the parameters $\alpha$ and $\beta$ for the discrete gamma distribution.
Note that for a fixed species tree $S$ the matrices  $\text{Flat}_{ac | bd}(P^*_{(S,\boldsymbol{\tau})})$ and $\text{Flat}_{ad | bc}(P^*_{(S,\boldsymbol{\tau})})$ will be identical, since for any observation at the leaves of $S$ for $i_1, i_2, i_3, i_4 \in [\kappa]$ we have $$p^*_{i_1 i_2 i_3 i_4 | (S,\boldsymbol{\tau})} = p^*_{i_1 i_2 i_4 i_3 |(S,\boldsymbol{\tau})}.$$

Thus, for simplicity let $F:=\text{Flat}_{ac | bd}(P^*_{(S,\boldsymbol{\tau})})=\text{Flat}_{ad | bc}(P^*_{(S,\boldsymbol{\tau})})$.  Notice that we can express flattening $F$ as $F = \frac{1}{m} \sum_{i = 1}^{m} F_{\rho_i}$, where $\rho_i>0$ is the rate for category $i \in [m]$ associated with the discrete gamma distribution and $F_{\rho_i}$ is the flattening associated with the fixed rate $\rho_i$.  Also, note that all entries in the matrix $F_{\rho_i}$ are positive. To differentiate between site pattern probabilities on the species tree and site pattern probabilities for a specific rate $\rho_i$ on the same species tree, we use the notation $p^*_{i_1 i_2 i_3 i_4 | (S,\boldsymbol{\tau})}$ and $p^{\rho_i,*}_{i_1 i_2 i_3 i_4 | (S,\boldsymbol{\tau})}$, respectively. Let the proportion of invariable sites be such that $\delta=0$. 
 Next, let $\mathbf{Q}$ be a generalized Jukes-Cantor matrix as described by Equation (\ref{jcmatrix}) in Section \ref{JC69} and let $\tau_2 = \tau_3$ and $\tau_1 = \frac{\tau_3}{10}$. Since $\tau_2 = \tau_3$ then the distributions  $P^*_{(S,\boldsymbol{\tau})}$ will be the same for both symmetric and asymmetric species trees. Furthermore, let $\tau_3 = 1$, $\theta = 0.1$, and $\mu_i = 0.1$ for all $i \in \{1, \dots , \frac{1}{2}\kappa(\kappa-1) \}$.

First let $\kappa =2$. By using precise formulas for the site pattern probabilities for the generalized Jukes-Cantor as described in Section \ref{JC69} and Supplement  A and the choice of parameters made above, we find that all principal minors of a symmetric matrix $F_{\rho_i}$ are positive for any $\rho_i >0$ (see {\em Mathematica} supplementary file for computations). This implies that  $F_{\rho_i}$ is positive definite matrix. Thus we conclude that  $F = \frac{1}{m} \sum_{i = 1}^{m} F_{\rho_i}$ is positive definite and hence invertible. In particular, we have that generically rank$(F) = \kappa^2$.

For $\kappa \geq 3$ we further specify parameters associated with the discrete gamma distribution. Let $\alpha =1$ and $\beta > \frac{1+\ln(m)}{2}$ for any fixed number $m$ of categories. First select all rows and columns from the matrix $F_{\rho_i}$ labeled by $ac, bd$-indices with distinct $a$ and $c$ (and with distinct $b$ and $d$). Denote this $\kappa(\kappa-1) \times \kappa(\kappa-1)$ submatrix by $F_{\rho_i}^*$. Note that for $\kappa \geq 4$, $\kappa(\kappa-1) > \binom{\kappa+1}{2}$, while for $\kappa = 3$ this is clearly false. Thus, at an appropriate point we will have two sub-cases: one for $\kappa \geq 4$ and one for $\kappa = 3$ (for which we will add an additional column and row to the submatrix  $F_{\rho_i}^*$). By the symmetry of the matrix $\mathbf{Q}$ the matrix $F_{\rho_i}^*$ will have six distinct entries for $\kappa \geq 4$ (see Supplement  A), 
$$p^{\rho_i,*}_{xxyy}, \ p^{\rho_i,*}_{xxyz}, \ p^{\rho_i,*}_{xyxy}, \ p^{\rho_i,*}_{xyxz}, \ p^{\rho_i,*}_{yzxx}, \ \text{and} \  p^{\rho_i,*}_{xyzw}. $$ 
Note that for $\kappa=3$ there is no entry $p^{\rho_i,*}_{xyzw}$.  The diagonal $(\kappa-1) \times (\kappa -1)$ blocks of $F_{\rho_i}^*$ will be of the following form,
{\small \begin{gather*}
\begin{pmatrix}
p^{\rho_i,*}_{xxyy}  & p^{\rho_i,*}_{xxyz} & p^{\rho_i,*}_{xxyz} & \cdots & p^{\rho_i,*}_{xxyz} \\
p^{\rho_i,*}_{xxyz} & p^{\rho_i,*}_{xxyy}  & p^{\rho_i,*}_{xxyz} &\cdots & p^{\rho_i,*}_{xxyz} \\
p^{\rho_i,*}_{xxyz} & p^{\rho_i,*}_{xxyz} & p^{\rho_i,*}_{xxyy}  &\cdots & p^{\rho_i,*}_{xxyz} \\
\vdots  & \vdots & \vdots  & \ddots & \vdots  \\
p^{\rho_i,*}_{xxyz} & p^{\rho_i,*}_{xxyz} & p^{\rho_i,*}_{xxyz} &\cdots & p^{\rho_i,*}_{xxyy}  
\end{pmatrix}.
\end{gather*}}
 
The off-diagonal $(\kappa-1) \times (\kappa -1)$  blocks of $F_{\rho_i}^*$ can be obtained from the block described below by appropriately permuting rows and columns.
{\small $$\bordermatrix{
				&12 			& 13 	 		& 14 			& 15 	 		& \ldots 	& 1\kappa\cr
                21&		p^{\rho_i,*}_{xyxy} 	& p^{\rho_i,*}_{xyxz} 	& p^{\rho_i,*}_{xyxz} 	& p^{\rho_i,*}_{xyxz}	& \ldots	&p^{\rho_i,*}_{xyxz}	\cr
                23& 		p^{\rho_i,*}_{xyxz}  	& p^{\rho_i,*}_{yzxx} 	& p^{\rho_i,*}_{xyzw}	& p^{\rho_i,*}_{xyzw}	& \ldots	&p^{\rho_i,*}_{xyzw}\cr
                24& 		p^{\rho_i,*}_{xyxz} 	& p^{\rho_i,*}_{xyzw} 	& p^{\rho_i,*}_{yzxx}	& p^{\rho_i,*}_{xyzw}	& \ldots	&p^{\rho_i,*}_{xyzw}\cr
                25		& p^{\rho_i,*}_{xyxz}  	& p^{\rho_i,*}_{xyzw} 	& p^{\rho_i,*}_{xyzw} 	& p^{\rho_i,*}_{yzxx}	& \ldots	&p^{\rho_i,*}_{xyzw}\cr
                26		& p^{\rho_i,*}_{xyxz}	& p^{\rho_i,*}_{xyzw} 	& p^{\rho_i,*}_{xyzw}	& p^{\rho_i,*}_{xyzw}	& \ldots	&p^{\rho_i,*}_{xyzw}\cr
                \vdots	& \vdots 		& \vdots	 	& \vdots		& \vdots		& \ddots	& \vdots \cr
                2\kappa	& p^{\rho_i,*}_{xyxz}  	& p^{\rho_i,*}_{xyzw}     & p^{\rho_i,*}_{xyzw}	& p^{\rho_i,*}_{xyzw}	& \ldots	&p^{\rho_i,*}_{yzxx}}. $$}
The submatrix $F_{\rho_i}^*$ is symmetric. In addition, it is straightforward to show that row (column) sums are all the same and they are equal to
\begin{equation}
p^{\rho_i,*}_{xxyy} + p^{\rho_i,*}_{xyxy} + (\kappa - 2)(p^{\rho_i,*}_{xxyz} + 2p^{\rho_i,*}_{xyxz}+p^{\rho_i,*}_{yzxx}+(\kappa - 3)p^{\rho_i,*}_{xyzw}).
\end{equation}
Recall that $\alpha=1$ and $\beta >  \frac{1+\ln(m)}{2}$, where $m$ is the number of categories. From our earlier discussion this implies that $0<\rho_1 < \rho_2 < \cdots < \rho_m <2$.
First, let $\kappa \geq 4$. To show diagonal dominance of the matrix $F_{\rho_i}^*$ we need to show that 
\begin{equation}\label{diagonal}
p^{\rho_i,*}_{xxyy} -( p^{\rho_i,*}_{xyxy} + (\kappa - 2)(p^{\rho_i,*}_{xxyz} + 2p^{\rho_i,*}_{xyxz}+p^{\rho_i,*}_{yzxx}+(\kappa - 3)p^{\rho_i,*}_{xyzw})) > 0.
\end{equation}
With the choice of parameters made above and using precise formulas for the site pattern probabilities for the generalized Jukes-Cantor model as described in Section \ref{JC69} and Supplement  A we show that Equation (\ref{diagonal}) can be simplified to the following form
\begin{equation}\label{diagonal2}
\frac{f_1(\rho_i)}{\kappa^2}+\frac{f_2(\rho_i)}{\kappa^3}+\frac{f_3(\rho_i)}{\kappa^4},
\end{equation}
where each $f_j$ is the function of $\rho_i$. We find that if $\rho_i \in(0,2)$ then $f_j(\rho_i) >0$ for $j \in \{1,2,3\}$. Therefore, Equation (\ref{diagonal}) is strictly positive and we conclude that $F_{\rho_i}^*$ is strictly diagonally dominant for any $\rho_i \in(0,2)$ and hence positive definite. This implies that $F^* = \frac{1}{m} \sum_{i = 1}^{m} F_{\rho_i}^*$ is positive definite and invertible. Thus, generically rank$(F^*) = \kappa(\kappa -1)$  (see {\em Mathematica} supplementary files for computations).
 
Now, let $\kappa=3$. In this case the submatrix $F_{\rho_i}^*$ is $6 \times 6$ thus we need to add at least one row and column from the flattening $F_{\rho_i}$. Choose the last row and the last column from $F_{\rho_i}$ labeled by $\kappa \kappa$-index and insert it as the $7^{th}$ row and column into $F_{\rho_i}^*$ by removing entries with 11 and 22 indices. Call this new submatrix $F_{\rho_i}^{**}$. Again, using parameters chosen above and precise formulas we find that for any $\rho_i \in (0,2)$ the $6 \times 6$ principal submatrix of  $F_{\rho_i}^{**}$, which is just $F_{\rho_i}^*$, is strictly diagonally dominant and also we show that det$(F_{\rho_i}^{**}) > 0$.  Now, this implies that all principal minors of $F_{\rho_i}^{**}$ are positive for any $\rho_i \in (0,2)$ (see {\em Mathematica} supplementary file for computations). Thus, we conclude that  $F^{**} = \frac{1}{m} \sum_{i = 1}^{m} F_{\rho_i}^{**}$ is positive definite and hence invertible. In particular, generically rank$(F^{**}) = \kappa (\kappa -1) +1$. Note that the interval for $\rho_i$ is not tight, we simply chose the smallest interval to accommodate all cases.
 
Since there exists at least one parameter choice in $U_{S}$ that does not lie in $W$, then $W$ is a proper analytic subvariety of $U_{S}$ for a 4-leaf species tree $S$ and hence of dimension strictly less than the dimension of $U_{S}$. We conclude that for a non-valid split $L_1|L_2$ $$\textup{rank}(\textup{Flat}_{L_1 | L_2}(P^*_{(S,\boldsymbol{\tau})})) > \binom{\kappa+1}{2},$$  for generic  distributions $P^*_{(S,\boldsymbol{\tau})}$ arising from the model, establishing (2) . 
\end{proof}
\begin{rem}
Note that the proof of Theorem \ref{thm:quartets}(2) actually establishes that if $L_1 | L_2$ is not a valid split for $S$, then for generic distributions $P^*_{(S,\boldsymbol{\tau})}$ arising from the model and $\kappa \geq 4$
$$\textup{rank}(\textup{Flat}_{L_1 | L_2}(P^*_{(S,\boldsymbol{\tau})})) \geq \kappa(\kappa-1).$$
\end{rem}

\begin{rem}
Recall that the probability distribution $P^*_{(S,\boldsymbol{\tau})}$ arises from a fixed 4-leaf species tree topology $S$. Thus, we can use the vanishing of the the $(\binom{\kappa+1}{2}+1)$-minors  of the $\textup{Flat}_{L_1 | L_2}(P^*_{(S,\boldsymbol{\tau})})$ to identify $S$ for generic parameters.  In particular, using the notation of the Theorem \ref{thm:quartets}, we conclude that the unrooted 4-leaf species tree $S$ is identifiable from $P^*_{(S,\boldsymbol{\tau})} = \psi_S(u)$ for any parameters $u \in U_{S}\setminus W$.
\end{rem}
In the single 4-leaf gene tree case (for the general $\kappa$-state Markov model) under the molecular clock assumption, even though the columns of the flattening labeled by the $cd$-indices $kl$ and $lk$ for distinct $k , l \in [\kappa]$ are identical, the rank of the flattening for a valid split is less than or equal to $\kappa$. Thus, one might wonder if the rank of the flattening for the species tree under the coalescent model might be less than or equal to $\kappa$ for all distributions arising from the model $\mathcal{C}_{GTR(\kappa)}$. For a valid split and using our computations for JC69 $\kappa$-state model under the coalescent (Section \ref{JC69}) we have found several $(\kappa+1) \times (\kappa+1)$ minors for $\kappa = 4$ that do not vanish identically. 

Now we turn our attention to the identifiability of the unrooted $n$-taxon species tree from the induced quartets. 
\begin{rem}
Let $S_n$ be an $n$-taxon species tree. We are going to show that the distribution for $S_n$ can be marginalized to $S_{n-1}$. We demonstrate the idea on the 5-taxon species tree $S_5$; the $n$-taxon case follows easily from the one described below. 
\end{rem}
\begin{figure}[h]
\centering
\includegraphics[scale=0.35]{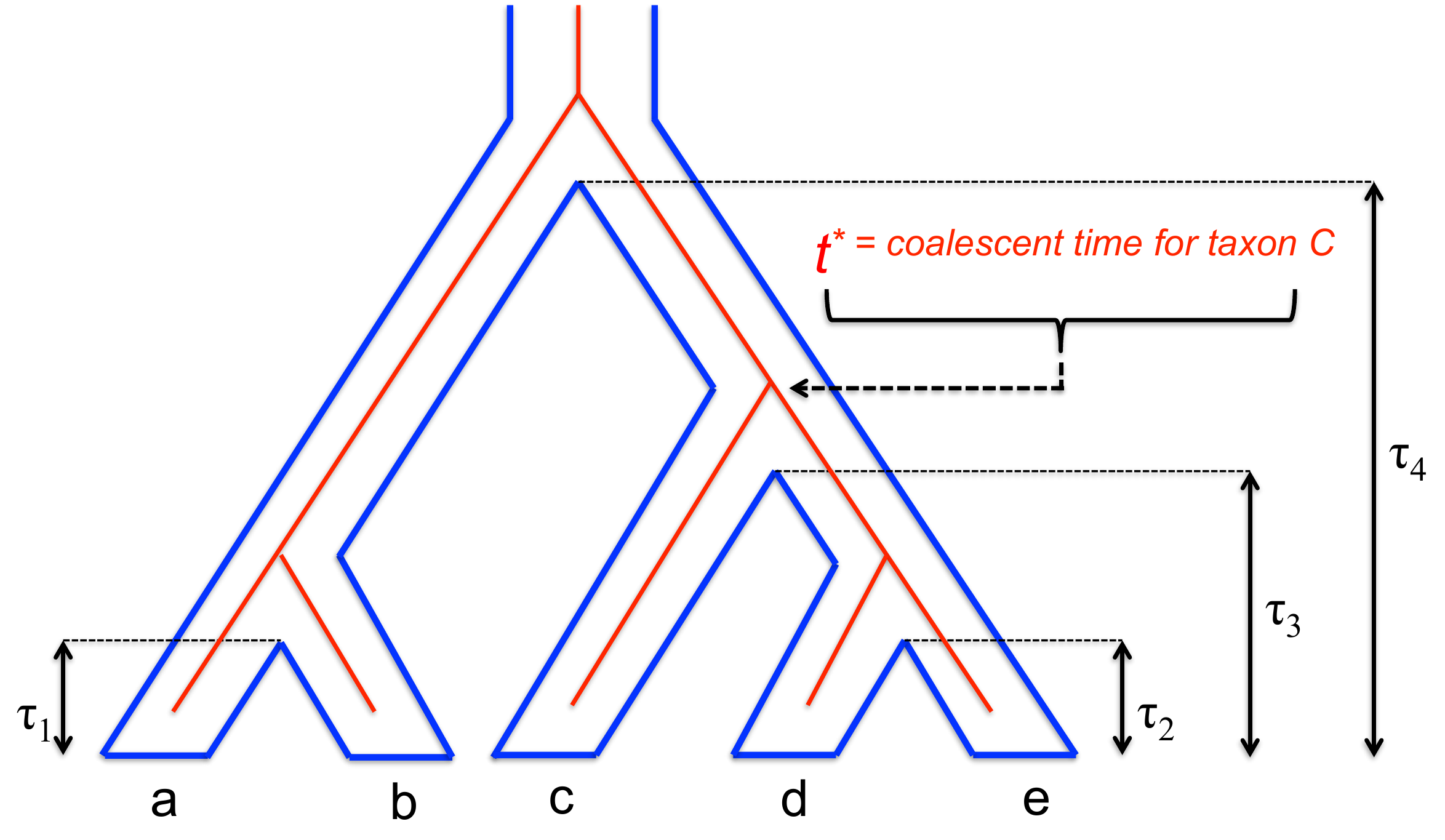}
\caption{A 5-taxon species tree with one example gene tree embedded.}
\label{5leaf}
\end{figure}
Let $\mathcal{Q}(S_5)$ be a collection of 4-leaf species trees that are induced by $S_5$. Consider $S_5 = ((a,b),(c,(d,e)))$ as in Figure \ref{5leaf} and let $S \in \mathcal{Q}(S_5)$ be $((a,b),(d,e))$. We would like to show that we can express each site pattern probability for an induced quartet tree $S$ as,
\begin{equation}\label{eq:induced_pattern}
p^*_{i_1 i_2 i_4 i_5 |(S,\boldsymbol{\tau})} = \sum_{x \in [\kappa]} p^*_{i_1 i_2 x i_4 i_5 |(S_5,\boldsymbol{\tau_\star})}.
\end{equation}
In the equation above $\boldsymbol{\tau_\star} = (\tau_1,\tau_2,\tau_3,\tau_4)$ denotes the speciation times for $S_5$ (Figure \ref{5leaf}) and $\boldsymbol{\tau} = (\tau_1, \tau_2, \tau_4)$ denotes the speciation times for the induced quartet tree $S$.  Let $t^*$ be the coalescent time for taxon $C$ and define $\mathbf{t} := \mathbf{t_\star} \setminus t^*$, where $\mathbf{t_\star}$ is a vector of coalescent times for a 5-taxon gene tree $G_5$ conditional on $S_5$. Thus, $\mathbf{t}$ is a vector of coalescent times for a 4-taxon gene tree $G$ conditional on $S$.  Let $\mathcal{G}_4$ and $\mathcal{G}_5$ be sets of all metric gene trees $(G, \mathbf{t})$ and  $(G_5, \mathbf{t_\star})$, respectively. Now, using Equation \ref{eq:site.pattern.prob.coal} as applied to a 5-taxon species tree $S_5$, the right-hand side of Equation  \ref{eq:induced_pattern} is equal to,

{\small 
\begin{align}\label{eq:calculation1}
\nonumber \sum_{x \in [\kappa]} p^*_{i_1 i_2 x i_4 i_5 |(S_5,\boldsymbol{\tau_\star})} 
\nonumber	&= \sum_{x \in [\kappa]} \sum_{(G_{5},\mathbf{t_\star})\in\mathcal{G}_5}\int_{\mathbf{t_\star}} p_{\sigma(i_1 i_2 x  i_4 i_5 ) |(G_{5},\mathbf{t_\star})} f{((G_{5},\mathbf{t_\star})|(S_5,\boldsymbol{\tau_\star}))}d\mathbf{t_\star}\\
\nonumber	&=\sum_{(G_{5},\mathbf{t_\star})\in\mathcal{G}_5}\int_{\mathbf{t_\star}} \sum_{x \in [\kappa]} p_{\sigma(i_1 i_2 x  i_4 i_5 ) |(G_{5},\mathbf{t_\star})} f((G_{5},\mathbf{t_\star})|(S_5,\boldsymbol{\tau_\star}))d\mathbf{t_\star}\\
	&=\sum_{(G_{5},\mathbf{t_\star})\in\mathcal{G}_5}\int_{\mathbf{t_\star}} p_{\sigma(i_1 i_2 i_4 i_5 ) |(G,\mathbf{t})} f((G_{5},\mathbf{t_\star})|(S_5,\boldsymbol{\tau_\star}))d\mathbf{t_\star}\\
\nonumber	&=\sum_{(G,\mathbf{t})\in\mathcal{G}_4}\sum_{(G,\mathbf{t}) | (G_{5},\mathbf{t_\star})\in\mathcal{G}_5}\int_{\mathbf{t_\star}} p_{\sigma(i_1 i_2 i_4 i_5 ) |(G,\mathbf{t})} f((G_{5},\mathbf{t_\star})|(S_5,\boldsymbol{\tau_\star}))d\mathbf{t_\star}.
\end{align}}
The second sum in the last expression, $\sum_{(G,\mathbf{t}) | (G_{5},\mathbf{t_\star})\in\mathcal{G}_5}$, means that for each 4-taxon gene tree $(G,\mathbf{t})$ we sum over all placements of the $5^{\text{th}}$ taxon. Essentially, we replace a sum over all five-taxon trees with a sum over all four-taxon trees that includes, for each four-taxon tree, a sum over all placements of the fifth taxon. Define $B^*$ to be the branch of $S_5$ on which $t^*$ occurs, and let $B$ be the set of all branches of $S_5$ regardless of where $t^*$ occurs. Then according to  Equation \ref{gtd} we can write the gene tree density given $S_5$ as
\begin{align}\label{gtd_5leaf}
\nonumber f((G_{5},\mathbf{t_\star})|(S_5, \boldsymbol{\tau_\star}))
&= \prod_{b\in B} f_{P_b} (t_{u_b}, t_{u_{b}-1}, \ldots t_{v_b+1})\\ 
&=\left(\prod_{b\in B \setminus B^*}  f_{P_b} (\mathbf{t}) \right) \cdot  f_{P_{B^*}}(\mathbf{t},t^*).
\end{align} 
Substituting \ref{gtd_5leaf} into the last expression of \ref{eq:calculation1} we get
{\small
\begin{align*}
\sum_{x \in [\kappa]} &p^*_{i_1 i_2 x i_4 i_5 |(S_5,\boldsymbol{\tau_\star})} \\
&=\sum_{(G,\mathbf{t})\in\mathcal{G}_4}\sum_{{(G,\mathbf{t}) | (G_{5},\mathbf{t_\star})\in\mathcal{G}_5}}\int_{\mathbf{t_\star}} p_{\sigma(i_1 i_2 i_4 i_5 ) |(G,\mathbf{t})} f((G_{5},\mathbf{t_\star})|(S_5,\boldsymbol{\tau_\star}))d\mathbf{t_\star}\\
&=\sum_{(G,\mathbf{t})\in\mathcal{G}_4}\sum_{{(G,\mathbf{t}) | (G_{5},\mathbf{t_\star})\in\mathcal{G}_5}}\int_{\mathbf{t}} \int_{t^*} p_{\sigma(i_1 i_2 i_4 i_5 ) |(G,\mathbf{t})}\left(\prod_{b\in B \setminus B^*}  f_{P_b} (\mathbf{t})\right)f_{P_{B^*}}(\mathbf{t},t^*)dt^*d\mathbf{t}\\
&=\sum_{(G,\mathbf{t})\in\mathcal{G}_4}\int_{\mathbf{t}}p_{\sigma(i_1 i_2 i_4 i_5 ) |(G,\mathbf{t})}\prod_{b\in B \setminus B^*}  f_{P_b} (\mathbf{t})\left(\sum_{{(G,\mathbf{t}) | (G_{5},\mathbf{t_\star})}\in\mathcal{G}_5}\int_{t^*}  f_{P_{B^*}}(\mathbf{t},t^*)dt^*\right)d\mathbf{t}\\
&=\sum_{(G,\mathbf{t})\in\mathcal{G}_4}\int_{\mathbf{t}}p_{\sigma(i_1 i_2 i_4 i_5 ) |(G,\mathbf{t})}f((G,\mathbf{t})|(S, \boldsymbol{\tau}))d\mathbf{t}\\
&=p^*_{i_1 i_2 i_4 i_5 |(S,\boldsymbol{\tau})}.
\end{align*}}
Let $U_{S_5}$ denote the continuous parameter space for the species tree $S_5$ and $U_S$ the continuous parameter space for the induced 4-leaf tree $S$. From the argument above it is easy to see that we have a commutative diagram of analytic maps with $\alpha_{S}$ surjective and $\beta_{S}$ a marginalization map, e.g. we sum over indices for taxon C (Figure \ref{5leaf} and Equation \ref{eq:induced_pattern}).
\begin{equation*}
  \xymatrix@R+2em@C+2em{
  U_{S_5}\ar[r]^-{\psi_{S_5}} \ar[d]_-{\alpha_S} & \Delta^{\kappa^5-1} \ar[d]^-{\beta_S} \\
   U_S \ar[r]_-{\psi_S} & \Delta^{\kappa^4-1}
  }
 \end{equation*}
It is straightforward to extend all of these ideas to the $n$-taxon case. By applying an argument similar to that of {\em Corollary 3} on page 1109 in Allman and Rhodes (2006) \cite{allmanrhodes2006}, we get the following result.
\begin{cor}\label{cor:bigtrees}
Let $\mathcal{C}_{S_n}$ denote the coalescent phylogenetic model for the $n$-taxon species tree $S_n$. Then the unrooted species tree $S_n$ is identifiable for generic parameters.
\end{cor}

\section{Generalized Jukes-Cantor Coalescent $\kappa$-state Model}\label{JC69}
To show generic identifiability in Theorem \ref{thm:quartets} we have used precise formulas for the generalized Jukes-Cantor $\kappa$-state model under the coalescent for a $4$-taxon species tree. In this section we describe computations of the JC69 model for a symmetric and asymmetric $4$-leaf species trees.
 
As was mentioned in Section \ref{mprocess}, if  we let the rates of mutation  and relative frequencies of the states at equilibrium be $\mu:=\mu_1 =  \dots = \mu_{\frac{1}{2}\kappa(\kappa-1)}$ and $\boldsymbol{\pi}=(\frac{1}{\kappa},\frac{1}{\kappa}, \cdots ,\frac{1}{\kappa})$ in equation (\ref{qmatrix}),  respectively, then the resulting model is the Mk model as described in \cite{lewis2001}, which is a generalized Jukes-Cantor $\kappa$-state model. In addition, let the proportion of invariable sites, $\delta$, be 0.
The $\kappa \times \kappa$ instantaneous rate matrix $\mathbf{Q}$ for the $\kappa$-state JC69 model is
\begin{center}
{\small 
\begin{gather}
\rho_i \mathbf{Q} =\rho_i \alpha
\begin{pmatrix}\label{jcmatrix}
  1-\kappa & 1 & \cdots & 1 \\
  1 & 1-\kappa & \cdots & 1 \\
  \vdots  & \vdots  & \ddots & \vdots  \\
  1 & 1 & \cdots & 1-\kappa
 \end{pmatrix},
\end{gather}
}
\end{center}
where $\alpha = \frac{\mu}{\kappa}$ is the instantaneous rate of any transition between states and $\rho_i > 0$ is the rate of evolution at a site for category $i \in [m]$ associated with the discrete gamma distribution. The transition probability matrix $\mathbf{P}^{\rho_i}(t)=e^{\mathbf{Q}\rho_i t}$ takes the form, 
\begin{eqnarray}\label{jc69}
\mathbf{P}^{\rho_i}_{ij}(t) & = &\left\{
\begin{array}{ll}
\frac{1}{\kappa} + \frac{\kappa-1}{\kappa} \mbox{e}^{-\rho_i \mu t } & i=j, \\
\vspace{0.005in}\\
\frac{1}{\kappa} - \frac{1}{\kappa} \mbox{e}^{ -\rho_i \mu t } & i \neq j. 
\end{array}\right.
\end{eqnarray}

The symmetry of the Jukes-Cantor model makes it possible to write each site pattern probability for a 4-leaf species tree concisely for any  $\kappa \geq 2$. 

\subsection{Site pattern probabilities: gene tree}

Consider a 4-leaf unrooted gene tree labeled as in Figure \ref{4leafgenetree}.  Recall that the site pattern probability for a particular observation $i_1 i_2 i_ 3 i_4$, $i_j \in [\kappa]$, and the branch length $v_e$, $e \in\{1,2,3,4,5\}$  is given by
\begin{equation}\label{genejc69}
p_{i_1 i_2 i_3 i_4}^{\rho_i} = \sum_{r=1}^{\kappa} \sum_{s=1}^{\kappa} \pi_r \mathbf{P}^{\rho_i}_{r i_1}(v_1) \mathbf{P}^{\rho_i}_{r i_2}(v_2) \mathbf{P}^{\rho_i}_{rs}(v_3) \mathbf{P}^{\rho_i}_{s i_3} (v_4)\mathbf{P}^{\rho_i}_{s i_4}(v_5).
\end{equation}
To make notation simpler, we let $p^{e}_{ii}:=\mathbf{P}^{\rho_i}_{ii}(v_e)$ be the probability of no change in the state over a time interval of length $v_e$ and  let $p^{e}_{ij}:=\mathbf{P}^{\rho_i}_{ij}(v_e)$ be the probability of a state change along the edge $e\in\{1,2,3,4,5\}$ with length $v_e$. Then the probability of the pattern $xxxx$, where $x \in [\kappa]$ is

{\footnotesize 
\begin{flalign}\label{genexxxx}
p_{xxxx}^{\rho_i}&=\frac{1}{\kappa}(p^{1}_{ii}p^{2}_{ii}p^{3}_{ii}p^{4}_{ii}p^{5}_{ii}+(\kappa-1)p^{1}_{ii}p^{2}_{ii}p^{3}_{ij}p^{4}_{ij}p^{5}_{ij}+(\kappa-1)p^{1}_{ij}p^{2}_{ij}p^{3}_{ij}p^{4}_{ii}p^{5}_{ii} \\ \nonumber
&+(\kappa-1)p^{1}_{ij}p^{2}_{ij}p^{3}_{ii}p^{4}_{ij}p^{5}_{ij}+(\kappa-1)(\kappa-2)p^{1}_{ij}p^{2}_{ij}p^{3}_{ij}p^{4}_{ij}p^{5}_{ij}).
\end{flalign}}
Indeed, this is easily computed by observing that when $r=s=x$ in Equation (\ref{genejc69}) then we will have exactly one term of the form $$p^{1}_{ii}p^{2}_{ii}p^{3}_{ii}p^{4}_{ii}p^{5}_{ii}.$$ For $r=s\neq x$ we will have  $(\kappa-1)$ terms of the form $$p^{1}_{ij}p^{2}_{ij}p^{3}_{ii}p^{4}_{ij}p^{5}_{ij}.$$ The other three cases for $r \neq s$ are as follows:
\begin{itemize}
\item  $r=x$ and $s \in [\kappa]\backslash \{r\}$, then we have $(\kappa-1)$ terms of the form $p^{1}_{ii}p^{2}_{ii}p^{3}_{ij}p^{4}_{ij}p^{5}_{ij}$;
\vspace{0.01in}
\item  $s=x$ and $r \in [\kappa]\backslash \{s\}$, then we have $(\kappa-1)$ terms of the form $p^{1}_{ij}p^{2}_{ij}p^{3}_{ij}p^{4}_{ii}p^{5}_{ii}$;
\vspace{0.01in}
\item  $r,s \neq x$, then we have $(\kappa-1)(\kappa-2)$ terms of the form $p^{1}_{ij}p^{2}_{ij}p^{3}_{ij}p^{4}_{ij}p^{5}_{ij}$.
\end{itemize}
The complete list of all site pattern probabilities on the single gene tree can be found in the Supplement  A. 

\subsection{Site pattern probabilities: species tree}

Now we are ready to describe computations for two 4-leaf species trees,  the symmetric tree $((a,b),(c,d))$ and the asymmetric tree $(a,(b,(c,d)))$, under the generalized JC69 coalescent $\kappa$-state model. For these species trees we have 25 gene tree pairs $(G,\mathbf{t})$ for a symmetric species tree and 34 gene tree pairs $(G,\mathbf{t})$ for an asymmetric species tree. Vectors $\boldsymbol{\tau}=(\tau_1,\tau_2,\tau_3)$ and $\mathbf{t}=(t_1,t_2,t_3)$ denote speciation and coalescent times, respectively. 
The gene tree $(G,\mathbf{t})$ in Figure \ref{histA}, call it  $(G_A, \mathbf{t})$, is symmetric and is in agreement with the tree $S$.  It has one coalescent event happening along the branch of length $\tau_3-\tau_1$ and two other events above the root of the symmetric species tree $S$. In contrast, the gene tree in Figure \ref{histB}, call it $(G_B, \mathbf{t})$, is asymmetric with lineages $A$ and $C$ being sister leaves and all coalescent events happening above the root of the asymmetric species tree $S$. Using Equation (\ref{gtd}) we compute the gene tree densities for $(G_A, \mathbf{t})$ and $(G_B, \mathbf{t})$.

\paragraph{Gene tree density for $(G_A, \mathbf{t})$}
\begin{flalign*}
f((G_A,\mathbf{t})|(S,\boldsymbol{\tau}))&= \Big(\frac{2}{\theta}\Big)^3 \mbox{e}^{\frac{-2} {\theta}t_1} \mbox{e}^{\frac{-6}{\theta} t_2} \mbox{e}^{\frac{-2}{\theta} (t_3-t_2)}\mbox{e}^{\frac{-2}{\theta} (\tau_3-\tau_1)}.
\end{flalign*}    
 
 \paragraph{Gene tree density for $(G_B, \mathbf{t})$} 
 \begin{equation*}
f((G_B,\mathbf{t})|(S,\boldsymbol{\tau}))  = \Big(\frac{2}{\theta}\Big)^ 3 \mbox{e}^{\frac{-2}{\theta}(\tau_2- \tau_1)} \mbox{e}^{\frac{-6}{\theta}(\tau_3 - \tau_2)} \mbox{e}^{\frac{-12 }{\theta}t_1} \mbox{e}^{\frac{-6}{\theta}(t_2 - t_1)} \mbox{e}^{\frac{-2}{\theta}(t_3 - t_2)}. 
\end{equation*}

The site pattern probability for the pattern $xxxx$ and gene trees $(G_A, \mathbf{t})$ and $(G_B, \mathbf{t})$ on symmetric (Figure \ref{histA})  and asymmetric (Figure \ref{histB}) species trees is now easily computed by appropriately substituting gene tree branch lengths as listed in Figure \ref{hist} into Equations (\ref{jc69}), (\ref{genejc69}), and (\ref{genexxxx}), multiplied by the gene tree density function and integrated with respect to $\mathbf{t}$. Notice the different limits of integration for each gene tree. As was mentioned previously, the limits depend on the location of coalescent events:

\paragraph{Figure \ref{histA}}
\begin{equation*}
p^{\rho_i,*}_{xxxx | ((G_A,\mathbf{t})|(S,\boldsymbol{\tau}))}= \int_0^{\infty} \int_0^{t_3} \int_0^{\tau_3-\tau_2} p_{xxxx | (G_A,\mathbf{t})}^{\rho_i} \cdot f{((G_A,\mathbf{t})|(S,\boldsymbol{\tau}))}  dt_1 dt_2 dt_3.
\end{equation*}
\paragraph{Figure \ref{histB}}
\begin{equation*}
p^{\rho_i,*}_{xxxx | ((G_B,\mathbf{t})|(S,\boldsymbol{\tau}))}= \int_0^{\infty} \int_0^{t_3} \int_0^{t_2} p_{xxxx | (G_B,\mathbf{t})}^{\rho_i}  \cdot f((G_B,\mathbf{t})|(S,\boldsymbol{\tau}))  dt_1 dt_2 dt_3.
\end{equation*}

We perform these computations for all 25 gene tree pairs for the symmetric species tree and for all 34 pairs for the asymmetric species tree, sum them up and arrive at the following site pattern probabilities for observation $xxxx$. The complete list for all patterns can be found in Supplement  A. We list site pattern probabilities in the parameterized form for a cleaner output.  For each $i \in \{1,2,3\}$ let $x_{i}=\mbox{e}^{-\tau_{i}}$, then
\vspace{0.07in}

\noindent {\em (1) Symmetric species 4-leaf tree:}
{\footnotesize \begin{align*}
p^{\rho_i,*}_{xxxx | (S,\boldsymbol{\tau})}&=\frac{1}{\kappa^4}+\frac{(\kappa-1) x_1^{2 \rho_i \mu}}{\kappa^4 (1+\rho_i \mu \theta)}+\frac{(\kappa-1) x_2^{2\rho_i \mu}}{\kappa^4 (1+\rho_i \mu \theta)}+\frac{(\kappa-1)^2 x_1^{2\rho_i \mu} x_2^{2\rho_i \mu}}{\kappa^4 (1+\rho_i \mu \theta)^2}+\frac{4 (\kappa-1) x_3^{2\rho_i \mu}}{\kappa^4 (1+\rho_i \mu \theta)}\\
&+\frac{4 (\kappa-2) (\kappa-1) x_1^{\rho_i \mu} x_3^{2\rho_i \mu}}{\kappa^4 (1+\rho_i \mu \theta) (2+\rho_i \mu \theta)}+\frac{4 (\kappa-2) (\kappa-1) x_2^{\rho_i \mu} x_3^{2\rho_i \mu}}{\kappa^4 (1+\rho_i \mu \theta) (2+\rho_i \mu \theta)}+\frac{4 (\kappa-2)^2 (\kappa-1) x_1^{\rho_i \mu} x_2^{\rho_i \mu} x_3^{2\rho_i \mu}}{\kappa^4 (1+\rho_i \mu \theta) (2+\rho_i \mu \theta)^2}\\
&+\frac{2 (\kappa-1)\rho_i \mu \theta (\kappa+\rho_i \mu \theta) (\kappa+(\kappa-1)\rho_i \mu \theta) x_1^{\frac{-2}{\theta}} x_2^{\frac{-2}{\theta}} x_3^{4 \left(\rho_i \mu+\frac{1}{\theta}\right)}}{\kappa^4 (1+ \rho_i \mu \theta)^2 (2+\rho_i \mu \theta)^2 (3+\rho_i \mu \theta)}.
\end{align*}}

\noindent {\em (2) Asymmetric species 4-leaf tree:}
{\footnotesize \begin{align*}
p^{\rho_i,*}_{xxxx | (S,\boldsymbol{\tau})}&=\frac{1}{\kappa^4}+\frac{(\kappa-1) x_1^{2 \rho_i \mu}}{\kappa^4 (1+\rho_i \mu \theta)}+\frac{2 (\kappa-1) x_2^{2\rho_i \mu}}{\kappa^4 (1+\rho_i \mu \theta)}+\frac{2 (\kappa-2) (\kappa-1) x_1^{\rho_i \mu} x_2^{2\rho_i \mu}}{\kappa^4 (1+\rho_i \mu \theta) (2+\rho_i \mu \theta)}\\
&+\frac{3 (\kappa-1) x_3^{2\rho_i \mu}}{\kappa^4 (1+\rho_i \mu \theta)}+\frac{2 (\kappa-2) (\kappa-1) x_1^{\rho_i \mu} x_3^{2\rho_i \mu}}{\kappa^4 (1+\rho_i \mu \theta) (2+\rho_i \mu \theta)}+\frac{(\kappa-1)^2 x_1^{2\rho_i \mu} x_3^{2\rho_i \mu}}{\kappa^4 (1+\rho_i \mu \theta)^2}\\
&+\frac{4 (\kappa-2) (\kappa-1) x_2^{\rho_i \mu} x_3^{2\rho_i \mu}}{\kappa^4 (1+\rho_i \mu \theta) (2+\rho_i \mu \theta)}+\frac{4 (\kappa-2)^2 (\kappa-1) x_1^{\rho_i \mu} x_2^{\rho_i \mu} x_3^{2\rho_i \mu}}{\kappa^4 (1+\rho_i \mu \theta) (2+\rho_i \mu \theta)^2}\\
&+\frac{2 (\kappa-1) \rho_i \mu \theta (k+\rho_i \mu \theta) (k+(\kappa-1)\rho_i \mu \theta) x_1^{\frac{-2}{\theta}} x_2^{2 \left(\rho_i \mu+\frac{1}{\theta}\right)} x_3^{2\rho_i \mu}}{\kappa^4 (1+\rho_i \mu \theta)^2 (2+\rho_i \mu \theta)^2 (3+\rho_i \mu \theta)}.
\end{align*}}

\section{Species Tree Inference}
The results presented here can be used to develop methodology for inferring species trees given DNA sequence data at the tips of the tree in cases where the coalescent model is appropriate.  Indeed, we have recently developed and done some preliminary testing with one such method \cite{chifmankubatko2014}, and we briefly describe the basic ideas behind the method here.  These ideas are modeled after the work of Eriksson \cite{Eriksson2005} for the single gene case. In particular, consider a data set consisting of $R$ unlinked single nucleotide polymorphisms (SNPs) for a collection of $n$ species.   Each SNP site is assumed to have its own gene tree, and by ``unlinked'' we mean that SNPs on the same chromosome are located far enough apart that their gene trees are independent given the species tree.  In other words, we assume a data set of $R$ observations arising under the model in expression (\ref{map}).

For a given subset of four species from the $n$ species under study, we can consider the three possible splits and their associated flattenings.  In the empirical setting, we do not observe the flattening matrices, but these matrices can be estimated using the counts of the observed site patterns in the $R$ data points.  The question of interest is then which of the three possible splits gives a flattening that is closest to a rank 10 matrix.  The relevant rank is 10, because $\kappa = 4$ for DNA sequence data and  $\binom{4+1}{2}=10$. To assess this, we compute the distance to the nearest rank 10 matrix in the Frobenius norm by defining the SVD score  for a split $L_1|L_2$, {\em SVD}$(L_1|L_2)$, to be 
\begin{equation}
\text{{\em SVD}}(L_1|L_2) =  \sqrt{\sum_{i=11}^{16} \hat{\sigma}_i^2}
\end{equation}
\noindent where $\hat{\sigma_i}$ are the singular values of ${Flat_{L_1|L_2}( \hat{P} )}$, for $i \in \{11, \dots, 16\}$, and ${Flat_{L_1|L_2}( \hat{P} )}$ refers to the flattening matrix for the split $L_1|L_2$ estimated from the data.  The split yielding the smallest score of the three possible splits can be inferred to be the valid split for the collection of four taxa.  We show \cite{chifmankubatko2014} that the SVD score has good ability to infer the correct split under a variety of conditions.  

We note, also, that although the results presented in this paper apply to unlinked SNP data, many existing data sets consist of collections of genes for which the entire DNA sequence is available for each species. This setting is different than that considered here, in that each site in the DNA sequence for a specific gene is generally assumed to have arisen from the {\itshape same} gene tree.  When the number of genes is large, these multi-locus data sets will approximately satisfy the model, and we expect that the inference method proposed here will still perform well.  We tested this using simulation \cite{chifmankubatko2014} and found that the method was still very effective at distinguishing the valid split.

To estimate the entire species tree, our proposed algorithm works by sampling collections of four taxa at random from those included in the data set. For each sampled collection of four taxa, the valid split is inferred by computation of the SVD score for the three possible splits.  The collection of valid four-taxon splits can then be used in a quartet assembly algorithm to construct the species tree.  Many possible algorithms for quartet assembly have been proposed \cite{strimmervonhaessler1996,strimmeretal1997,snirrao2012}. We find that the method of Snir and Rao (2010) \cite{snirrao2012} is a fast and effective technique for constructing the species tree estimate from the collection of inferred splits. 

Based on the good performance of this method in our initial study, we are currently working on expanding the scope of the data sets to which it can be applied. In particular, we will consider data sets for which (i) multiple individuals are sampled within a species; (ii) there are ambiguities in the observed nucleotide at a site, and  (iii) sequence data other than DNA, such as codon or amino acid data, have been collected.  Overall, the method shows good promise in addressing this important problem in phylogenetic inference under the coalescent process, and we believe that it can be effectively used for practical phylogenetic inference.

\section{Conclusions}
Within the last 15  years, numerous methods for inferring species-level phylogenies from genome-scale data have been proposed, and several are commonly used in empirical practice (e.g., \cite{liupearl2007}, \cite{heleddrummond2010}, \cite{kubatkoetal2009}, \cite{liuetal2010}, \cite{bryantetal2012}).  However, identifiability of the species tree from sequence data has never been formally established. We have shown here that the  unrooted $n$-taxon species tree topology is identifiable from the distribution of site pattern probabilities under the coalescent model with a general $\kappa$-state time-reversible substitution model for the mutational process that allows for a proportion of sites to be invariable and that allows for rate variation across sites under a discrete gamma model.  This is an important step in understanding phylogenetic estimation under more complicated models of DNA sequence evolution, and it has already led to the development of a promising new method for inference \cite{chifmankubatko2014}.

Several important problems in this area remain to be  addressed. For example, open questions include whether the root of the species tree is identifiable, and whether other parameters in the species tree (e.g., branch lengths and effective population sizes) and in the substitution model (e.g., parameters in the transition probability matrix) are identifiable under the model. We consider these questions in future work.

\section{Acknowledgements}
We thank Elizabeth Allman for helpful comments on an earlier draft of this manuscript, and an anonymous reviewer who noticed that our method of proof covered the case of models for site-specific rate variation. We acknowledge support from the National Science Foundation under award  DMS-1106706 (to J. C. and L. K.); from the Mathematical Biosciences Institute at The Ohio State University (J. C. and L. K.); and from NIH Cancer Biology Training Grant T32-CA079448 at Wake Forest School of Medicine (J.C.). Both authors contributed equally to this work.

\newpage
\pagestyle{empty}

\bibliographystyle{elsarticle-num}
\bibliography{bibfile}

\begin{thebibliography}{10}
\expandafter\ifx\csname url\endcsname\relax
  \def\url#1{\texttt{#1}}\fi
\expandafter\ifx\csname urlprefix\endcsname\relax\def\urlprefix{URL }\fi
\expandafter\ifx\csname href\endcsname\relax
  \def\href#1#2{#2} \def\path#1{#1}\fi

\bibitem{maddison1997}
W.~P. Maddison, {Gene trees in species trees}, Syst. Biol. 46 (1997) 523--536.

\bibitem{pamilonei1988}
P.~Pamilo, M.~Nei, {Relationships between gene trees and species trees}, Mol.
  Biol. Evol. 5~(5) (1988) 568--583.

\bibitem{kingman1982a}
J.~F.~C. Kingman, {On the genealogy of large populations}, J. Appl. Prob. 19A
  (1982) 27--43.

\bibitem{kingman1982b}
J.~F.~C. Kingman, {The coalescent}, Stoch. Proc. Appl. 13 (1982) 235--248.

\bibitem{tavare1984}
S.~{Tavar{\'{e}}}, {Line-of-descent and genealogical processes, and their
  applications in population genetics models.}, Theor. Popul. Biol. 26 (1984)
  119--164.

\bibitem{liupearl2007}
L.~Liu, D.~Pearl, Species trees from gene trees: reconstructing {B}ayesian
  posterior distributions of a species phylogeny using estimated gene tree
  distributions, Syst. Biol. 56 (2007) 504--514.

\bibitem{heleddrummond2010}
J.~Heled, A.~J. Drummond, Bayesian inference of species trees from multilocus
  data, Mol. Biol. Evol. 27~(3) (2010) 570--580.

\bibitem{kubatkoetal2009}
L.~S. Kubatko, B.~C. Carstens, L.~L. Knolwes, {STEM}: Species tree estimation
  using maximum likelihood for gene trees under coalescence, Bioinformatics
  25~(7) (2009) 971--973.

\bibitem{liuetal2010}
L.~Liu, L.~Yu, S.~Edwards, A maximum pseudo-likelihood approach for estimating
  species trees under the coalescent model, BMC Evol. Biol. 10~(302).

\bibitem{bryantetal2012}
D.~Bryant, R.~Bouckaert, J.~Felsenstein, N.~Rosenberg, A.~RoyChoudhury,
  Inferring species trees directly from biallelic genetic markers: bypassing
  gene trees in a full coalescent analysis, Mol. Biol. Evol. 29~(8) (2012)
  1917--1932.

\bibitem{allmanetal2011}
E.~S. Allman, J.~H. Degnan, J.~A. Rhodes, Identifying the rooted species tree
  from the distribution of unrooted gene trees under the coalescent, J. Math.
  Biol. 62 (2011) 833--862.

\bibitem{allmanetal2011b}
E.~S. Allman, J.~H. Degnan, J.~A. Rhodes,
  \href{http://dx.doi.org/10.1016/j.jtbi.2011.08.006}{Determining species tree
  topologies from clade probabilities under the coalescent}, J. Theoret. Biol.
  289 (2011) 96--106.
\newblock \href {http://dx.doi.org/10.1016/j.jtbi.2011.08.006}
  {\path{doi:10.1016/j.jtbi.2011.08.006}}.
\newline\urlprefix\url{http://dx.doi.org/10.1016/j.jtbi.2011.08.006}

\bibitem{liuedwards2009}
L.~Liu, S.~V. Edwards, Phylogenetic analysis in the anomaly zone, Systematic
  Biology 58~(4) (2009) 452--460.

\bibitem{tavare1986}
S.~Tavare, Some probabilistic and statistical problems in the analysis of {DNA}
  sequences, Lectures on Mathematics in the Life Sciences (American
  Mathematical Society) 17 (1986) 57--86.

\bibitem{rannalayang2003}
B.~Rannala, Z.~Yang, Likelihood and \text{B}ayes estimation of ancestral
  population sizes in hominoids using data from multiple loci, Genetics 164
  (2003) 1645--1656.

\bibitem{jukescantor1969}
T.~Jukes, C.~R. Cantor, Evolution of protein molecules, in: H.~N. Munro (Ed.),
  Mammalian protein metabolism, Academic Press, New York, 1969, pp. 21--123.

\bibitem{lewis2001}
P.~O. Lewis, A likelihood approach to estimating phylogeny from discrete
  morphological character data, Systematic Biology 50~(6) (2001) 913--925.
\newblock \href {http://dx.doi.org/10.1080/106351501753462876}
  {\path{doi:10.1080/106351501753462876}}.

\bibitem{kimura1980}
M.~Kimura, A simple method for estimating evolutionary rate of base
  substitution through comparative studies of nucleotide sequences, J. Mol.
  Evol. 16 (1980) 111--120.

\bibitem{felsenstein1981}
J.~Felsenstein, Evolutionary trees from {DNA} sequences: {A} maximum likelihood
  approach, J. Mol. Evol. 17~(6) (1981) 368--76.

\bibitem{hasegawaetal1985}
M.~Hasegawa, H.~Kishino, T.~Yano, Dating of human-ape splitting by a molecular
  clock of mitochondrial {DNA}, Journal of Molecular Evolution 22~(2) (1985)
  160--174.

\bibitem{tamuranei1993}
K.~Tamura, M.~Nei, Estimation of the number of nucleotide substitutions in the
  control region of mitochondrial {DNA} in humans and chimpanzees, Mol. Biol.
  Evol. 10~(3) (1993) 512--526.

\bibitem{yang1993}
Z.~Yang, Maximum likelihood estimation of phylogeny from {DNA} sequences when
  substitution rates differ over sites, Molecular Biology and Evolution 10
  (1993) 1396--1401.

\bibitem{yang1994}
Z.~Yang, \href{http://dx.doi.org/10.1007/BF00160154}{Maximum likelihood
  phylogenetic estimation from {DNA} sequences with variable rates over sites:
  Approximate methods}, Journal of Molecular Evolution 39~(3) (1994) 306--314.
\newblock \href {http://dx.doi.org/10.1007/BF00160154}
  {\path{doi:10.1007/BF00160154}}.
\newline\urlprefix\url{http://dx.doi.org/10.1007/BF00160154}

\bibitem{allmanrhodes2006}
E.~S. Allman, J.~A. Rhodes, The identifiability of tree topology for
  phylogenetic models, including covarion and mixture models, Journal of
  Computational Biology 13~(5) (2006) 1101--1113.

\bibitem{allmanetal2008}
E.~S. Allman, C.~An\'{e}, J.~A. Rhodes, Identifiability of a {M}arkovian model
  of molecular evolution with gamma-distributed rates, Advances in Applied
  Probability 40 (2008) 228--249.

\bibitem{allmanrhodes2008b}
E.~S. Allman, J.~A. Rhodes, Identifying evolutionary trees and substitution
  parameters for the general {M}arkov model with invariable sites, Mathematical
  Biosciences 211 (2008) 18--33.

\bibitem{Lake1987}
J.~A. Lake, A rate independent technique for analysis of nucleic acid
  sequences: Evolutionary parsimony, Molecular Biology and Evolution 4~(2)
  (1987) 167--191.

\bibitem{Cavender1989}
J.~A. Cavender, Mechanized derivation of linear invariants, Molecular Biology
  and Evolution 6~(3) (1989) 301--316.

\bibitem{Fu1992}
Y.~Fu, W.~Li, Construction of linear invariants in phylogenetic inference,
  Mathematical Biosciences 109~(2) (1992) 201 -- 228.

\bibitem{Fu1995}
Y.~Fu, Linear invariants under {J}ukes' and {C}antor's one-parameter model,
  Journal of Theoretical Biology 173~(4) (1995) 339--352.

\bibitem{allmanrhodes2009}
E.~S. Allman, J.~A. Rhodes, The identifiability of covarion models in
  phylogenetics, IEEE/ACM Transactions in Computational Biology and
  Bioinformatics 6 (2009) 76--88.

\bibitem{allmanetal2010b}
E.~S. Allman, S.~Petrovic, J.~A. Rhodes, S.~Sullivant, Identifiability of
  2-tree mixtures for group-based models, IEEE/ACM Transactions in
  Computational Biology and Bioinformatics 8~(3) (2011) 710--722.

\bibitem{gunning2009analytic}
R.~Gunning, H.~Rossi,
  \href{http://books.google.com/books?id=L0zJmamx5AAC}{Analytic Functions of
  Several Complex Variables}, {AMS} Chelsea Publishing, {AMS} Chelsea Pub.,
  2009.
\newline\urlprefix\url{http://books.google.com/books?id=L0zJmamx5AAC}

\bibitem{chifmankubatko2014}
J.~Chifman, L.~Kubatko, Quartet inference from {SNP} data under the coalescent
  model, Bioinformatics 30~(23) (2014) 3317--3324.
\newblock \href {http://dx.doi.org/10.1093/bioinformatics/btu53}
  {\path{doi:10.1093/bioinformatics/btu53}}.

\bibitem{Eriksson2005}
N.~Eriksson, {Tree construction using Singular Value Decomposition}, in:
  L.~Pachter, B.~Sturmfels (Eds.), Algebraic Statistics for Computational
  Biology, Cambridge University Press, 2005, Ch.~19, pp. 347--358.

\bibitem{strimmervonhaessler1996}
K.~Strimmer, A.~von Haeseler, Quartet puzzling: A quartet maximum likelihood
  method for reconstructing tree topologies, Mol. Biol. Evol. 13 (1996)
  964--969.

\bibitem{strimmeretal1997}
K.~Strimmer, N.~Goldman, A.~von Haeseler, Bayesian probabilities and quartet
  puzzling, Mol. Biol. Evol. 14 (1997) 210--213.

\bibitem{snirrao2012}
S.~Snir, S.~Rao, Quartet {M}ax{C}ut: A fast algorithm for amalgamating quartet
  trees, Mol. Phylogen. Evol. 62 (2012) 1--8.

\end{thebibliography}



\section*{Supplementary material}

Supplement A

Main article title: Identifiability of the unrooted species tree topology under the coalescent model with time-reversible substitution processes, site-specific rate variation, and invariable sites

Julia Chifman
*Wake Forest School of Medicine*

Laura Kubatko
*The Ohio State University*

# 1 List of Site Pattern Probabilities for Generalized JC69 coalescent $\kappa$-state model

In this supplement we provide a list of site pattern probabilities for gene and species trees as described in Section 6 of the main article.

## 1.1 Four-leaf gene tree site pattern probabilities

$$p_{xxxx}^{\rho_i} = \frac{1}{\kappa}(p_{ii}^1 p_{ii}^2 p_{ii}^3 p_{ii}^4 p_{ii}^5 + (\kappa-1)p_{ii}^1 p_{ii}^2 p_{ij}^3 p_{ij}^4 p_{ij}^5 + (\kappa-1)p_{ij}^1 p_{ij}^2 p_{ij}^3 p_{ii}^4 p_{ii}^5 + (\kappa-1)p_{ij}^1 p_{ij}^2 p_{ii}^3 p_{ij}^4 p_{ij}^5$$
$$+ (\kappa-1)(\kappa-2)p_{ij}^1 p_{ij}^2 p_{ij}^3 p_{ij}^4 p_{ij}^5)$$

$$p_{xxxy}^{\rho_i} = \frac{1}{\kappa}(p_{ii}^1 p_{ii}^2 p_{ii}^3 p_{ii}^4 p_{ij}^5 + p_{ii}^1 p_{ii}^2 p_{ij}^3 p_{ij}^4 p_{ii}^5 + p_{ij}^1 p_{ij}^2 p_{ii}^3 p_{ij}^4 p_{ii}^5 + (\kappa-2)p_{ii}^1 p_{ii}^2 p_{ij}^3 p_{ij}^4 p_{ij}^5 + (\kappa-1)p_{ij}^1 p_{ij}^2 p_{ij}^3 p_{ii}^4 p_{ij}^5$$
$$+ (\kappa-2)p_{ij}^1 p_{ij}^2 p_{ij}^3 p_{ij}^4 p_{ii}^5 + (\kappa-2)p_{ij}^1 p_{ij}^2 p_{ii}^3 p_{ij}^4 p_{ij}^5 + (\kappa-2)^2 p_{ij}^1 p_{ij}^2 p_{ij}^3 p_{ij}^4 p_{ij}^5)$$

$$p_{xxyx}^{\rho_i} = \frac{1}{\kappa}(p_{ii}^1 p_{ii}^2 p_{ii}^3 p_{ij}^4 p_{ii}^5 + p_{ii}^1 p_{ii}^2 p_{ij}^3 p_{ii}^4 p_{ij}^5 + p_{ij}^1 p_{ij}^2 p_{ij}^3 p_{ii}^4 p_{ii}^5 + (\kappa-2)p_{ii}^1 p_{ii}^2 p_{ij}^3 p_{ij}^4 p_{ij}^5 + (\kappa-1)p_{ij}^1 p_{ij}^2 p_{ij}^3 p_{ij}^4 p_{ii}^5$$
$$+ (\kappa-2)p_{ij}^1 p_{ij}^2 p_{ij}^3 p_{ii}^4 p_{ij}^5 + (\kappa-2)p_{ij}^1 p_{ij}^2 p_{ii}^3 p_{ij}^4 p_{ij}^5 + (\kappa-2)^2 p_{ij}^1 p_{ij}^2 p_{ij}^3 p_{ij}^4 p_{ij}^5)$$

$$p_{xyxx}^{\rho_i} = \frac{1}{\kappa}(p_{ii}^1 p_{ij}^2 p_{ii}^3 p_{ii}^4 p_{ii}^5 + p_{ii}^1 p_{ij}^2 p_{ij}^3 p_{ii}^4 p_{ii}^5 + p_{ii}^1 p_{ij}^2 p_{ii}^3 p_{ij}^4 p_{ij}^5 + (\kappa-2)p_{ii}^1 p_{ij}^2 p_{ij}^3 p_{ij}^4 p_{ij}^5 + (\kappa-1)p_{ii}^1 p_{ij}^2 p_{ij}^3 p_{ij}^4 p_{ij}^5$$
$$+ (\kappa-2)p_{ij}^1 p_{ij}^2 p_{ij}^3 p_{ii}^4 p_{ii}^5 + (\kappa-2)p_{ij}^1 p_{ij}^2 p_{ii}^3 p_{ij}^4 p_{ij}^5 + (\kappa-2)^2 p_{ij}^1 p_{ij}^2 p_{ij}^3 p_{ij}^4 p_{ij}^5)$$

$$p_{yxxx}^{\rho_i} = \frac{1}{\kappa}(p_{ij}^1 p_{ii}^2 p_{ii}^3 p_{ii}^4 p_{ii}^5 + p_{ii}^1 p_{ij}^2 p_{ij}^3 p_{ii}^4 p_{ii}^5 + p_{ii}^1 p_{ij}^2 p_{ii}^3 p_{ij}^4 p_{ij}^5 + (\kappa-2)p_{ii}^1 p_{ij}^2 p_{ij}^3 p_{ij}^4 p_{ij}^5 + (\kappa-1)p_{ij}^1 p_{ii}^2 p_{ij}^3 p_{ij}^4 p_{ij}^5$$
$$+ (\kappa-2)p_{ij}^1 p_{ij}^2 p_{ij}^3 p_{ii}^4 p_{ii}^5 + (\kappa-2)p_{ij}^1 p_{ij}^2 p_{ii}^3 p_{ij}^4 p_{ij}^5 + (\kappa-2)^2 p_{ij}^1 p_{ij}^2 p_{ij}^3 p_{ij}^4 p_{ij}^5)$$

$$p_{xxyy}^{\rho_i} = \frac{1}{\kappa}(p_{ii}^1 p_{ii}^2 p_{ii}^3 p_{ij}^4 p_{ij}^5 + p_{ij}^1 p_{ij}^2 p_{ij}^3 p_{ii}^4 p_{ii}^5 + p_{ii}^1 p_{ii}^2 p_{ij}^3 p_{ii}^4 p_{ii}^5 + (\kappa-2)p_{ii}^1 p_{ii}^2 p_{ij}^3 p_{ij}^4 p_{ij}^5 + (\kappa-2)p_{ij}^1 p_{ij}^2 p_{ij}^3 p_{ii}^4 p_{ii}^5$$
$$+ (\kappa-2)p_{ij}^1 p_{ij}^2 p_{ii}^3 p_{ij}^4 p_{ij}^5 + (\kappa^2-3\kappa+3)p_{ij}^1 p_{ij}^2 p_{ij}^3 p_{ij}^4 p_{ij}^5)$$

1

$$p^{\rho_i}_{xyxy} = \frac{1}{\kappa}(p^1_{ii}p^2_{ij}p^3_{ii}p^4_{ii}p^5_{ij} + p^1_{ii}p^2_{ij}p^3_{ii}p^4_{ij}p^5_{ii} + p^1_{ij}p^2_{ii}p^3_{ij}p^4_{ii}p^5_{ij} + p^1_{ij}p^2_{ii}p^3_{ii}p^4_{ij}p^5_{ii} + (\kappa-2)p^1_{ii}p^2_{ij}p^3_{ij}p^4_{ij}p^5_{ij}$$
$$+ (\kappa-2)p^1_{ij}p^2_{ii}p^3_{ij}p^4_{ij}p^5_{ij} + (\kappa-2)p^1_{ij}p^2_{ij}p^3_{ii}p^4_{ij}p^5_{ij} + (\kappa-2)p^1_{ij}p^2_{ij}p^3_{ij}p^4_{ii}p^5_{ij} + (\kappa-2)p^1_{ij}p^2_{ij}p^3_{ii}p^4_{ij}p^5_{ij}$$
$$+ (\kappa-3)(\kappa-2)p^1_{ij}p^2_{ij}p^3_{ij}p^4_{ij}p^5_{ij})$$

$$p^{\rho_i}_{yxxy} = \frac{1}{\kappa}(p^1_{ii}p^2_{ij}p^3_{ii}p^4_{ij}p^5_{ii} + p^1_{ii}p^2_{ij}p^3_{ij}p^4_{ii}p^5_{ij} + p^1_{ij}p^2_{ii}p^3_{ij}p^4_{ij}p^5_{ii} + p^1_{ij}p^2_{ii}p^3_{ii}p^4_{ij}p^5_{ij} + (\kappa-2)p^1_{ii}p^2_{ij}p^3_{ij}p^4_{ij}p^5_{ij}$$
$$+ (\kappa-2)p^1_{ij}p^2_{ii}p^3_{ij}p^4_{ij}p^5_{ij} + (\kappa-2)p^1_{ij}p^2_{ij}p^3_{ij}p^4_{ii}p^5_{ij} + (\kappa-2)p^1_{ij}p^2_{ij}p^3_{ij}p^4_{ij}p^5_{ii} + (\kappa-2)p^1_{ij}p^2_{ij}p^3_{ii}p^4_{ij}p^5_{ij}$$
$$+ (\kappa-3)(\kappa-2)p^1_{ij}p^2_{ij}p^3_{ij}p^4_{ij}p^5_{ij})$$

$$p^{\rho_i}_{xxyz} = \frac{1}{\kappa}(p^1_{ii}p^2_{ij}p^3_{ii}p^4_{ij}p^5_{ij} + p^1_{ii}p^2_{ii}p^3_{ij}p^4_{ii}p^5_{ij} + p^1_{ii}p^2_{ii}p^3_{ij}p^4_{ij}p^5_{ii} + p^1_{ij}p^2_{ij}p^3_{ii}p^4_{ii}p^5_{ij} + p^1_{ij}p^2_{ij}p^3_{ii}p^4_{ij}p^5_{ii}$$
$$+ (\kappa-3)p^1_{ii}p^2_{ij}p^3_{ij}p^4_{ij}p^5_{ij} + (\kappa-3)p^1_{ij}p^2_{ij}p^3_{ij}p^4_{ij}p^5_{ij} + (\kappa-2)p^1_{ij}p^2_{ij}p^3_{ij}p^4_{ij}p^5_{ii} + (\kappa-2)p^1_{ij}p^2_{ij}p^3_{ij}p^4_{ii}p^5_{ij}$$
$$+ (\kappa^2-4\kappa+5)p^1_{ij}p^2_{ij}p^3_{ij}p^4_{ij}p^5_{ij})$$

$$p^{\rho_i}_{yzxx} = \frac{1}{\kappa}(p^1_{ij}p^2_{ij}p^3_{ii}p^4_{ii}p^5_{ii} + p^1_{ii}p^2_{ij}p^3_{ij}p^4_{ii}p^5_{ii} + p^1_{ii}p^2_{ij}p^3_{ii}p^4_{ij}p^5_{ij} + p^1_{ij}p^2_{ii}p^3_{ij}p^4_{ii}p^5_{ij} + p^1_{ij}p^2_{ii}p^3_{ij}p^4_{ij}p^5_{ij}$$
$$+ (\kappa-3)p^1_{ij}p^2_{ij}p^3_{ij}p^4_{ii}p^5_{ii} + (\kappa-3)p^1_{ij}p^2_{ij}p^3_{ii}p^4_{ij}p^5_{ij} + (\kappa-2)p^1_{ii}p^2_{ij}p^3_{ij}p^4_{ij}p^5_{ij} + (\kappa-2)p^1_{ij}p^2_{ii}p^3_{ij}p^4_{ij}p^5_{ij}$$
$$+ (\kappa^2-4\kappa+5)p^1_{ij}p^2_{ij}p^3_{ij}p^4_{ij}p^5_{ij})$$

$$p^{\rho_i}_{xyxz} = \frac{1}{\kappa}(p^1_{ii}p^2_{ij}p^3_{ii}p^4_{ii}p^5_{ij} + p^1_{ii}p^2_{ij}p^3_{ij}p^4_{ij}p^5_{ii} + p^1_{ij}p^2_{ii}p^3_{ij}p^4_{ii}p^5_{ij} + p^1_{ij}p^2_{ii}p^3_{ii}p^4_{ij}p^5_{ij} + p^1_{ij}p^2_{ii}p^3_{ij}p^4_{ij}p^5_{ii} + p^1_{ij}p^2_{ij}p^3_{ii}p^4_{ij}p^5_{ii}$$
$$+ (\kappa-2)p^1_{ii}p^2_{ij}p^3_{ij}p^4_{ij}p^5_{ij} + (\kappa-2)p^1_{ij}p^2_{ij}p^3_{ij}p^4_{ii}p^5_{ij} + (\kappa-3)p^1_{ij}p^2_{ii}p^3_{ij}p^4_{ij}p^5_{ij} + (\kappa-3)p^1_{ij}p^2_{ij}p^3_{ij}p^4_{ij}p^5_{ii}$$
$$+ (\kappa-3)p^1_{ij}p^2_{ij}p^3_{ii}p^4_{ij}p^5_{ij} + (\kappa^2-5\kappa+7)p^1_{ij}p^2_{ij}p^3_{ij}p^4_{ij}p^5_{ij})$$

$$p^{\rho_i}_{yxxz} = \frac{1}{\kappa}(p^2_{ii}p^1_{ij}p^3_{ii}p^4_{ii}p^5_{ij} + p^2_{ii}p^1_{ij}p^3_{ij}p^4_{ij}p^5_{ii} + p^2_{ij}p^1_{ii}p^3_{ij}p^4_{ii}p^5_{ij} + p^2_{ij}p^1_{ii}p^3_{ii}p^4_{ij}p^5_{ij} + p^2_{ij}p^1_{ii}p^3_{ij}p^4_{ij}p^5_{ii} + p^2_{ij}p^1_{ij}p^3_{ii}p^4_{ij}p^5_{ii}$$
$$+ (\kappa-2)p^2_{ii}p^1_{ij}p^3_{ij}p^4_{ij}p^5_{ij} + (\kappa-2)p^2_{ij}p^1_{ij}p^3_{ij}p^4_{ii}p^5_{ij} + (\kappa-3)p^2_{ij}p^1_{ii}p^3_{ij}p^4_{ij}p^5_{ij} + (\kappa-3)p^2_{ij}p^1_{ij}p^3_{ij}p^4_{ij}p^5_{ii}$$
$$+ (\kappa-3)p^2_{ij}p^1_{ij}p^3_{ii}p^4_{ij}p^5_{ij} + (\kappa^2-5\kappa+7)p^1_{ij}p^2_{ij}p^3_{ij}p^4_{ij}p^5_{ij})$$

$$p^{\rho_i}_{xyzx} = \frac{1}{\kappa}(p^1_{ii}p^2_{ij}p^3_{ii}p^5_{ii}p^4_{ij} + p^1_{ii}p^2_{ij}p^3_{ij}p^5_{ij}p^4_{ii} + p^1_{ij}p^2_{ii}p^3_{ij}p^5_{ii}p^4_{ij} + p^1_{ij}p^2_{ii}p^3_{ii}p^5_{ij}p^4_{ij} + p^1_{ij}p^2_{ii}p^3_{ij}p^5_{ij}p^4_{ii} + p^1_{ij}p^2_{ij}p^3_{ii}p^5_{ij}p^4_{ii}$$
$$+ (\kappa-2)p^1_{ii}p^2_{ij}p^3_{ij}p^5_{ij}p^4_{ij} + (\kappa-2)p^1_{ij}p^2_{ij}p^3_{ij}p^5_{ii}p^4_{ij} + (\kappa-3)p^1_{ij}p^2_{ii}p^3_{ij}p^5_{ij}p^4_{ij} + (\kappa-3)p^1_{ij}p^2_{ij}p^3_{ij}p^5_{ij}p^4_{ii}$$
$$+ (\kappa-3)p^1_{ij}p^2_{ij}p^3_{ii}p^5_{ij}p^4_{ij} + (\kappa^2-5\kappa+7)p^1_{ij}p^2_{ij}p^3_{ij}p^4_{ij}p^5_{ij})$$

$$p^{\rho_i}_{yxzx} = \frac{1}{\kappa}(p^2_{ii}p^1_{ij}p^3_{ii}p^5_{ii}p^4_{ij} + p^2_{ii}p^1_{ij}p^3_{ij}p^5_{ij}p^4_{ii} + p^2_{ij}p^1_{ii}p^3_{ij}p^5_{ii}p^4_{ij} + p^2_{ij}p^1_{ii}p^3_{ii}p^5_{ij}p^4_{ij} + p^2_{ij}p^1_{ii}p^3_{ij}p^5_{ij}p^4_{ii} + p^2_{ij}p^1_{ij}p^3_{ii}p^5_{ij}p^4_{ii}$$
$$+ (\kappa-2)p^2_{ii}p^1_{ij}p^3_{ij}p^5_{ij}p^4_{ij} + (\kappa-2)p^2_{ij}p^1_{ij}p^3_{ij}p^5_{ii}p^4_{ij} + (\kappa-3)p^2_{ij}p^1_{ii}p^3_{ij}p^5_{ij}p^4_{ij} + (\kappa-3)p^2_{ij}p^1_{ij}p^3_{ij}p^5_{ij}p^4_{ii}$$
$$+ (\kappa-3)p^2_{ij}p^1_{ij}p^3_{ii}p^5_{ij}p^4_{ij} + (\kappa^2-5\kappa+7)p^1_{ij}p^2_{ij}p^3_{ij}p^4_{ij}p^5_{ij})$$

$$p^{\rho_i}_{xyzw} = \frac{1}{\kappa}(p^1_{ii}p^2_{ij}p^3_{ij}p^4_{ij}p^5_{ij} + p^1_{ii}p^2_{ij}p^3_{ij}p^4_{ii}p^5_{ij} + p^1_{ii}p^2_{ij}p^3_{ij}p^4_{ij}p^5_{ii} + p^1_{ij}p^2_{ii}p^3_{ii}p^4_{ij}p^5_{ij} + p^1_{ij}p^2_{ii}p^3_{ij}p^4_{ii}p^5_{ij} + p^1_{ij}p^2_{ii}p^3_{ij}p^4_{ij}p^5_{ii}$$
$$+ p^1_{ij}p^2_{ij}p^3_{ii}p^4_{ii}p^5_{ij} + p^1_{ij}p^2_{ij}p^3_{ii}p^4_{ij}p^5_{ii} + (\kappa-3)p^1_{ii}p^2_{ij}p^3_{ij}p^4_{ij}p^5_{ij} + (\kappa-3)p^1_{ij}p^2_{ii}p^3_{ij}p^4_{ij}p^5_{ij} + (\kappa-3)p^1_{ij}p^2_{ij}p^3_{ij}p^4_{ij}p^5_{ii}$$
$$+ (\kappa-3)p^1_{ij}p^2_{ij}p^3_{ij}p^4_{ii}p^5_{ij} + (\kappa-4)p^1_{ij}p^2_{ij}p^3_{ii}p^4_{ij}p^5_{ij} + (\kappa^2-5\kappa+8)p^1_{ij}p^2_{ij}p^3_{ij}p^4_{ij}p^5_{ij})$$

## 1.2 Symmetric four-leaf species tree

Each site pattern probability for a symmetric species 4-leaf tree will be of the following form:

$$p_{i_1 i_2 i_3 i_4 | (S,\tau)}^{\rho_i,*} = c_0 + c_1 x_1^{2\rho_i \mu} + c_2 x_2^{2\rho_i \mu} + c_3 x_1^{2\rho_i \mu} x_2^{2\rho_i \mu} + c_4 x_3^{2\rho_i \mu} + c_5 x_1^{\rho_i \mu} x_3^{2\rho_i \mu} + c_6 x_2^{\rho_i \mu} x_3^{2\rho_i \mu}$$
$$+ c_7 x_1^{\rho_i \mu} x_2^{\rho_i \mu} x_3^{\rho_i \mu} + c_8 x_1^{-\frac{2}{\theta}} x_2^{-\frac{2}{\theta}} x_3^{4\left(\rho_i \mu + \frac{1}{\theta}\right)},$$

where $x_i = e^{-\tau_i}$, $i \in \{1, 2, 3\}$, $\kappa$ is the number of states, $\theta > 0$ is the effective population size and $\mu > 0$ is the instantaneous rate of any transition between states. Table 1 lists all the coefficients $c_i$ for each pattern. Under the molecular clock assumption, rooted balanced 4-leaf species tree $((a,b),(c,d))$ will have 9 distinct site patterns probabilities

$$p_{xxxx}^{\rho_i,*}, p_{xxxy}^{\rho_i,*} = p_{xxyx}^{\rho_i,*}, p_{xyxx}^{\rho_i,*} = p_{yxxx}^{\rho_i,*}, p_{xyxy}^{\rho_i,*} = p_{yxxy}^{\rho_i,*}, p_{xxyy}^{\rho_i,*},$$
$$p_{xyxz}^{\rho_i,*} = p_{yxxz}^{\rho_i,*} = p_{xyzx}^{\rho_i,*} = p_{yxzx}^{\rho_i,*}, p_{xxyz}^{\rho_i,*}, p_{yzxx}^{\rho_i,*}, p_{xyzw}^{\rho_i,*}.$$

Table 1: Coefficients for site pattern probabilities: symmetric species 4-leaf tree

|       | xxxx | xxxy | xyxx |
|-------|------|------|------|
| $c_0$ | $\frac{1}{\kappa^4}$ | $\frac{1}{\kappa^4}$ | $\frac{1}{\kappa^4}$ |
| $c_1$ | $\frac{(\kappa-1)}{\kappa^4(1+\rho_i\mu\theta)}$ | $-\frac{1}{\kappa^4(1+\rho_i\mu\theta)}$ | $\frac{(\kappa-1)}{\kappa^4(1+\rho_i\mu\theta)}$ |
| $c_2$ | $\frac{(\kappa-1)}{\kappa^4(1+\rho_i\mu\theta)}$ | $\frac{(\kappa-1)}{\kappa^4(1+\rho_i\mu\theta)}$ | $-\frac{1}{\kappa^4(1+\rho_i\mu\theta)}$ |
| $c_3$ | $\frac{(\kappa-1)^2}{\kappa^4(1+\rho_i\mu\theta)^2}$ | $-\frac{(\kappa-1)}{\kappa^4(1+\rho_i\mu\theta)^2}$ | $-\frac{(\kappa-1)}{\kappa^4(1+\rho_i\mu\theta)^2}$ |
| $c_4$ | $\frac{4(\kappa-1)}{\kappa^4(1+\rho_i\mu\theta)}$ | $\frac{2(\kappa-2)}{\kappa^4(1+\rho_i\mu\theta)}$ | $\frac{2(\kappa-2)}{\kappa^4(1+\rho_i\mu\theta)}$ |
| $c_5$ | $\frac{4(\kappa-2)(\kappa-1)}{\kappa^4(1+\rho_i\mu\theta)(2+\rho_i\mu\theta)}$ | $-\frac{4(\kappa-2)}{\kappa^4(1+\rho_i\mu\theta)(2+\rho_i\mu\theta)}$ | $\frac{2(\kappa-2)^2}{\kappa^4(1+\rho_i\mu\theta)(2+\rho_i\mu\theta)}$ |
| $c_6$ | $\frac{4(\kappa-2)(\kappa-1)}{\kappa^4(1+\rho_i\mu\theta)(2+\rho_i\mu\theta)}$ | $\frac{2(\kappa-2)^2}{\kappa^4(1+\rho_i\mu\theta)(2+\rho_i\mu\theta)}$ | $-\frac{4(\kappa-2)}{\kappa^4(1+\rho_i\mu\theta)(2+\rho_i\mu\theta)}$ |
| $c_7$ | $\frac{4(\kappa-2)^2(\kappa-1)}{\kappa^4(1+\rho_i\mu\theta)(2+\rho_i\mu\theta)^2}$ | $-\frac{4(\kappa-2)^2}{\kappa^4(1+\rho_i\mu\theta)(2+\rho_i\mu\theta)^2}$ | $-\frac{4(\kappa-2)^2}{\kappa^4(1+\rho_i\mu\theta)(2+\rho_i\mu\theta)^2}$ |
| $c_8$ | $\frac{2(\kappa-1)\rho_i\mu\theta(\kappa+\rho_i\mu\theta)(\kappa+(\kappa-1)\rho_i\mu\theta)}{\kappa^4(1+\rho_i\mu\theta)^2(2+\rho_i\mu\theta)^2(3+\rho_i\mu\theta)}$ | $-\frac{2\rho_i\mu\theta(\kappa+\rho_i\mu\theta)(\kappa+(\kappa-1)\rho_i\mu\theta)}{\kappa^4(1+\rho_i\mu\theta)^2(2+\rho_i\mu\theta)^2(3+\rho_i\mu\theta)}$ | $-\frac{2\rho_i\mu\theta(\kappa+\rho_i\mu\theta)(\kappa+(\kappa-1)\rho_i\mu\theta)}{\kappa^4(1+\rho_i\mu\theta)^2(2+\rho_i\mu\theta)^2(3+\rho_i\mu\theta)}$ |

|       | xyxy | xxyy | xxyz |
|-------|------|------|------|
| $c_0$ | $\frac{1}{\kappa^4}$ | $\frac{1}{\kappa^4}$ | $\frac{1}{\kappa^4}$ |
| $c_1$ | $-\frac{1}{\kappa^4(1+\rho_i\mu\theta)}$ | $\frac{(\kappa-1)}{\kappa^4(1+\rho_i\mu\theta)}$ | $-\frac{1}{\kappa^4(1+\rho_i\mu\theta)}$ |
| $c_2$ | $-\frac{1}{\kappa^4(1+\rho_i\mu\theta)}$ | $\frac{(\kappa-1)}{\kappa^4(1+\rho_i\mu\theta)}$ | $\frac{(\kappa-1)}{\kappa^4(1+\rho_i\mu\theta)}$ |
| $c_3$ | $\frac{1}{\kappa^4(1+\rho_i\mu\theta)^2}$ | $\frac{(\kappa-1)^2}{\kappa^4(1+\rho_i\mu\theta)^2}$ | $-\frac{(\kappa-1)}{\kappa^4(1+\rho_i\mu\theta)^2}$ |
| $c_4$ | $\frac{2(\kappa-2)}{\kappa^4(1+\rho_i\mu\theta)}$ | $-\frac{4}{\kappa^4(1+\rho_i\mu\theta)}$ | $-\frac{4}{\kappa^4(1+\rho_i\mu\theta)}$ |
| $c_5$ | $-\frac{4(\kappa-2)}{\kappa^4(1+\rho_i\mu\theta)(2+\rho_i\mu\theta)}$ | $-\frac{4(\kappa-2)}{\kappa^4(1+\rho_i\mu\theta)(2+\rho_i\mu\theta)}$ | $\frac{8}{\kappa^4(1+\rho_i\mu\theta)(2+\rho_i\mu\theta)}$ |
| $c_6$ | $-\frac{4(\kappa-2)}{\kappa^4(1+\rho_i\mu\theta)(2+\rho_i\mu\theta)}$ | $-\frac{4(\kappa-2)}{\kappa^4(1+\rho_i\mu\theta)(2+\rho_i\mu\theta)}$ | $-\frac{4(\kappa-2)}{\kappa^4(1+\rho_i\mu\theta)(2+\rho_i\mu\theta)}$ |
| $c_7$ | $\frac{8(\kappa-2)}{\kappa^4(1+\rho_i\mu\theta)(2+\rho_i\mu\theta)^2}$ | $-\frac{4(\kappa-2)^2}{\kappa^4(1+\rho_i\mu\theta)(2+\rho_i\mu\theta)^2}$ | $\frac{8(\kappa-2)}{\kappa^4(1+\rho_i\mu\theta)(2+\rho_i\mu\theta)^2}$ |
| $c_8$ | $\frac{\rho_i\mu\theta(2\kappa^2+\kappa(3\kappa-2)\rho_i\mu\theta+(2+(\kappa-2)\kappa)\rho_i\mu^2\theta^2)}{\kappa^4(1+\rho_i\mu\theta)^2(2+\rho_i\mu\theta)^2(3+\rho_i\mu\theta)}$ | $\frac{2\rho_i\mu\theta(\kappa+\rho_i\mu\theta)^2}{\kappa^4(1+\rho_i\mu\theta)^2(2+\rho_i\mu\theta)^2(3+\rho_i\mu\theta)}$ | $\frac{2\rho_i\mu^2\theta^2(\kappa+\rho_i\mu\theta)}{\kappa^4(1+\rho_i\mu\theta)^2(2+\rho_i\mu\theta)^2(3+\rho_i\mu\theta)}$ |

|       | yzxx | xyzx | xyzw |
|-------|------|------|------|
| $c_0$ | $\frac{1}{\kappa^4}$ | $\frac{1}{\kappa^4}$ | $\frac{1}{\kappa^4}$ |
| $c_1$ | $\frac{(\kappa-1)}{\kappa^4(1+\rho_i\mu\theta)}$ | $-\frac{1}{\kappa^4(1+\rho_i\mu\theta)}$ | $-\frac{1}{\kappa^4(1+\rho_i\mu\theta)}$ |
| $c_2$ | $-\frac{1}{\kappa^4(1+\rho_i\mu\theta)}$ | $-\frac{1}{\kappa^4(1+\rho_i\mu\theta)}$ | $-\frac{1}{\kappa^4(1+\rho_i\mu\theta)}$ |
| $c_3$ | $-\frac{(\kappa-1)}{\kappa^4(1+\rho_i\mu\theta)^2}$ | $\frac{1}{\kappa^4(1+\rho_i\mu\theta)^2}$ | $\frac{1}{\kappa^4(1+\rho_i\mu\theta)^2}$ |
| $c_4$ | $-\frac{4}{\kappa^4(1+\rho_i\mu\theta)}$ | $\frac{(\kappa-4)}{\kappa^4(1+\rho_i\mu\theta)}$ | $-\frac{4}{\kappa^4(1+\rho_i\mu\theta)}$ |
| $c_5$ | $-\frac{4(\kappa-2)}{\kappa^4(1+\rho_i\mu\theta)(2+\rho_i\mu\theta)}$ | $-\frac{2(\kappa-4)}{\kappa^4(1+\rho_i\mu\theta)(2+\rho_i\mu\theta)}$ | $\frac{8}{\kappa^4(1+\rho_i\mu\theta)(2+\rho_i\mu\theta)}$ |
| $c_6$ | $\frac{8}{\kappa^4(1+\rho_i\mu\theta)(2+\rho_i\mu\theta)}$ | $-\frac{2(\kappa-4)}{\kappa^4(1+\rho_i\mu\theta)(2+\rho_i\mu\theta)}$ | $\frac{8}{\kappa^4(1+\rho_i\mu\theta)(2+\rho_i\mu\theta)}$ |
| $c_7$ | $\frac{8(\kappa-2)}{\kappa^4(1+\rho_i\mu\theta)(2+\rho_i\mu\theta)^2}$ | $\frac{4(\kappa-4)}{\kappa^4(1+\rho_i\mu\theta)(2+\rho_i\mu\theta)^2}$ | $-\frac{16}{\kappa^4(1+\rho_i\mu\theta)(2+\rho_i\mu\theta)^2}$ |
| $c_8$ | $\frac{2\rho_i\mu^2\theta^2(\kappa+\rho_i\mu\theta)}{\kappa^4(1+\rho_i\mu\theta)^2(2+\rho_i\mu\theta)^2(3+\rho_i\mu\theta)}$ | $-\frac{\rho_i\mu^2\theta^2(\kappa+(\kappa-2)\rho_i\mu\theta)}{\kappa^4(1+\rho_i\mu\theta)^2(2+\rho_i\mu\theta)^2(3+\rho_i\mu\theta)}$ | $\frac{2\rho_i\mu^3\theta^3}{\kappa^4(1+\rho_i\mu\theta)^2(2+\rho_i\mu\theta)^2(3+\rho_i\mu\theta)}$ |

One checks that for all observations $\sigma_i = i_1 i_2 i_3 i_4$, $i_j \in [\kappa]$

$$\sum_i p_{\sigma_i|(S,\tau)}^{\rho_i,*} = \kappa p_{xxxx}^{\rho_i,*} + 2\kappa(\kappa-1)p_{xxxy}^{\rho_i,*} + 2\kappa(\kappa-1)p_{xyxx}^{\rho_i,*} + \kappa(\kappa-1)p_{xxyy}^{\rho_i,*}$$
$$+ 2\kappa(\kappa-1)p_{xyxy}^{\rho_i,*} + \kappa(\kappa-1)(\kappa-2)p_{xxyz}^{\rho_i,*} + 4\kappa(\kappa-1)(\kappa-2)p_{xyxz}^{\rho_i,*}$$
$$+ \kappa(\kappa-1)(\kappa-2)p_{yzxx}^{\rho_i,*} + \kappa(\kappa-1)(\kappa-2)(\kappa-3)p_{xyzw}^{\rho_i,*} = 1.$$

## 1.3 Asymmetric four-leaf species tree

Each site pattern probability for a asymmetric species 4-leaf tree will be of the following form:

$$p_{i_1 i_2 i_3 i_4|(S,\tau)}^{\rho_i,*} = c_0 + c_1 x_1^{2\rho_i\mu} + c_2 x_2^{2\rho_i\mu} + c_3 x_1^{\rho_i\mu} x_2^{2\rho_i\mu} + c_4 x_3^{2\rho_i\mu} + c_5 x_1^{\rho_i\mu} x_3^{2\rho_i\mu} + c_6 x_1^{2\rho_i\mu} x_3^{2\rho_i\mu}$$
$$+ c_7 x_2^{\rho_i\mu} x_3^{2\rho_i\mu} + c_8 x_1^{\rho_i\mu} x_2^{\rho_i\mu} x_3^{2\rho_i\mu} + c_9 x_1^{\frac{-2}{\theta}} x_2^{2\left(\rho_i\mu+\frac{1}{\theta}\right)} x_3^{2\rho_i\mu},$$

where $x_i = e^{-\tau_i}$, $i \in \{1,2,3\}$, $\kappa$ is the number of states, $\theta > 0$ is the effective population size and $\mu > 0$ is the instantaneous rate of any transition between states. Table 2 lists all the coefficients $c_i$ for each pattern. Under the molecular clock assumption, rooted asymmetric 4-leaf species tree $(a,(b,(c,d)))$ will have 11 distinct site patterns probabilities

$$p_{xxxx}^{\rho_i,*}, p_{xxxy}^{\rho_i,*} = p_{xxyx}^{\rho_i,*}, p_{xyxx}^{\rho_i,*}, p_{yxxx}^{\rho_i,*}, p_{xyxy}^{\rho_i,*} = p_{yxxy}^{\rho_i,*}, p_{xxyy}^{\rho_i,*},$$
$$p_{xyxz}^{\rho_i,*} = p_{xyzx}^{\rho_i,*}, p_{yxxz}^{\rho_i,*} = p_{yxzx}^{\rho_i,*}, p_{xxyz}^{\rho_i,*}, p_{yzxx}^{\rho_i,*}, p_{xyzw}^{\rho_i,*}.$$

Table 2: Coefficients for site pattern probabilities: asymmetric species 4-leaf tree

|   | xxxx | xxxy | xyxx | yxxx |
|---|---|---|---|---|
| $c_0$ | $\frac{1}{\kappa^4}$ | $\frac{1}{\kappa^4}$ | $\frac{1}{\kappa^4}$ | $\frac{1}{\kappa^4}$ |
| $c_1$ | $\frac{(\kappa-1)}{\kappa^4(1+\rho_i\mu\theta)}$ | $-\frac{1}{\kappa^4(1+\rho_i\mu\theta)}$ | $\frac{(\kappa-1)}{\kappa^4(1+\rho_i\mu\theta)}$ | $\frac{(\kappa-1)}{\kappa^4(1+\rho_i\mu\theta)}$ |
| $c_2$ | $\frac{2(\kappa-1)}{\kappa^4(1+\rho_i\mu\theta)}$ | $\frac{(\kappa-2)}{\kappa^4(1+\rho_i\mu\theta)}$ | $-\frac{2}{\kappa^4(1+\rho_i\mu\theta)}$ | $\frac{2(\kappa-1)}{\kappa^4(1+\rho_i\mu\theta)}$ |
| $c_3$ | $\frac{2(\kappa-1)(\kappa-2)}{\kappa^4(1+\rho_i\mu\theta)(2+\rho_i\mu\theta)}$ | $-\frac{2(\kappa-2)}{\kappa^4(1+\rho_i\mu\theta)(2+\rho_i\mu\theta)}$ | $-\frac{2(\kappa-2)}{\kappa^4(1+\rho_i\mu\theta)(2+\rho_i\mu\theta)}$ | $\frac{2(\kappa-1)(\kappa-2)}{\kappa^4(1+\rho_i\mu\theta)(2+\rho_i\mu\theta)}$ |
| $c_4$ | $\frac{3(\kappa-1)}{\kappa^4(1+\rho_i\mu\theta)}$ | $\frac{(2\kappa-3)}{\kappa^4(1+\rho_i\mu\theta)}$ | $\frac{(2\kappa-3)}{\kappa^4(1+\rho_i\mu\theta)}$ | $-\frac{3}{\kappa^4(1+\rho_i\mu\theta)}$ |
| $c_5$ | $\frac{2(\kappa-2)(\kappa-1)}{\kappa^4(1+\rho_i\mu\theta)(2+\rho_i\mu\theta)}$ | $-\frac{2(\kappa-2)}{\kappa^4(1+\rho_i\mu\theta)(2+\rho_i\mu\theta)}$ | $\frac{2(\kappa-2)(\kappa-1)}{\kappa^4(1+\rho_i\mu\theta)(2+\rho_i\mu\theta)}$ | $-\frac{2(\kappa-2)}{\kappa^4(1+\rho_i\mu\theta)(2+\rho_i\mu\theta)}$ |
| $c_6$ | $\frac{(\kappa-1)^2}{\kappa^4(1+\rho_i\mu\theta)^2}$ | $-\frac{(\kappa-1)}{\kappa^4(1+\rho_i\mu\theta)^2}$ | $-\frac{(\kappa-1)}{\kappa^4(1+\rho_i\mu\theta)^2}$ | $\frac{(\kappa-1)}{\kappa^4(1+\rho_i\mu\theta)^2}$ |
| $c_7$ | $\frac{4(\kappa-2)(\kappa-1)}{\kappa^4(1+\rho_i\mu\theta)(2+\rho_i\mu\theta)}$ | $\frac{2(\kappa-2)^2}{\kappa^4(1+\rho_i\mu\theta)(2+\rho_i\mu\theta)}$ | $-\frac{4(\kappa-2)}{\kappa^4(1+\rho_i\mu\theta)(2+\rho_i\mu\theta)}$ | $-\frac{4(\kappa-2)}{\kappa^4(1+\rho_i\mu\theta)(2+\rho_i\mu\theta)}$ |
| $c_8$ | $\frac{4(\kappa-1)(\kappa-2)^2}{\kappa^4(1+\rho_i\mu\theta)(2+\rho_i\mu\theta)^2}$ | $-\frac{4(\kappa-2)^2}{\kappa^4(1+\rho_i\mu\theta)(2+\rho_i\mu\theta)^2}$ | $-\frac{4(\kappa-2)^2}{\kappa^4(1+\rho_i\mu\theta)(2+\rho_i\mu\theta)^2}$ | $-\frac{4(\kappa-2)^2}{\kappa^4(1+\rho_i\mu\theta)(2+\rho_i\mu\theta)^2}$ |
| $c_9$ | $\frac{2(\kappa-1)\rho_i\mu\theta(\kappa+\rho_i\mu\theta)(\kappa+(\kappa-1)\rho_i\mu\theta)}{\kappa^4(1+\rho_i\mu\theta)^2(2+\rho_i\mu\theta)^2(3+\rho_i\mu\theta)}$ | $-\frac{2\rho_i\mu\theta(\kappa+\rho_i\mu\theta)(\kappa+(\kappa-1)\rho_i\mu\theta)}{\kappa^4(1+\rho_i\mu\theta)^2(2+\rho_i\mu\theta)^2(3+\rho_i\mu\theta)}$ | $-\frac{2\rho_i\mu\theta(\kappa+\rho_i\mu\theta)(\kappa+(\kappa-1)\rho_i\mu\theta)}{\kappa^4(1+\rho_i\mu\theta)^2(2+\rho_i\mu\theta)^2(3+\rho_i\mu\theta)}$ | $-\frac{2\rho_i\mu\theta(\kappa+\rho_i\mu\theta)(\kappa+(\kappa-1)\rho_i\mu\theta)}{\kappa^4(1+\rho_i\mu\theta)^2(2+\rho_i\mu\theta)^2(3+\rho_i\mu\theta)}$ |

|       | xxyy | xyxy | xxyz | yzxx |
|-------|------|------|------|------|
| $c_0$ | $\frac{1}{\kappa^4}$ | $\frac{1}{\kappa^4}$ | $\frac{1}{\kappa^4}$ | $\frac{1}{\kappa^4}$ |
| $c_1$ | $\frac{(\kappa-1)}{\kappa^4(1+\rho_i\mu\theta)}$ | $-\frac{1}{\kappa^4(1+\rho_i\mu\theta)}$ | $-\frac{1}{\kappa^4(1+\rho_i\mu\theta)}$ | $\frac{(\kappa-1)}{\kappa^4(1+\rho_i\mu\theta)}$ |
| $c_2$ | $-\frac{2}{\kappa^4(1+\rho_i\mu\theta)}$ | $\frac{(\kappa-2)}{\kappa^4(1+\rho_i\mu\theta)}$ | $-\frac{2}{\kappa^4(1+\rho_i\mu\theta)}$ | $-\frac{2}{\kappa^4(1+\rho_i\mu\theta)}$ |
| $c_3$ | $-\frac{2(\kappa-2)}{\kappa^4(1+\rho_i\mu\theta)(2+\rho_i\mu\theta)}$ | $-\frac{2(\kappa-2)}{\kappa^4(1+\rho_i\mu\theta)(2+\rho_i\mu\theta)}$ | $\frac{4}{\kappa^4(1+\rho_i\mu\theta)(2+\rho_i\mu\theta)}$ | $-\frac{2(\kappa-2)}{\kappa^4(1+\rho_i\mu\theta)(2+\rho_i\mu\theta)}$ |
| $c_4$ | $\frac{(\kappa-3)}{\kappa^4(1+\rho_i\mu\theta)}$ | $\frac{(\kappa-3)}{\kappa^4(1+\rho_i\mu\theta)}$ | $\frac{(\kappa-3)}{\kappa^4(1+\rho_i\mu\theta)}$ | $-\frac{3}{\kappa^4(1+\rho_i\mu\theta)}$ |
| $c_5$ | $-\frac{2(\kappa-2)}{\kappa^4(1+\rho_i\mu\theta)(2+\rho_i\mu\theta)}$ | $-\frac{2(\kappa-2)}{\kappa^4(1+\rho_i\mu\theta)(2+\rho_i\mu\theta)}$ | $\frac{4}{\kappa^4(1+\rho_i\mu\theta)(2+\rho_i\mu\theta)}$ | $-\frac{2(\kappa-2)}{\kappa^4(1+\rho_i\mu\theta)(2+\rho_i\mu\theta)}$ |
| $c_6$ | $\frac{(\kappa-1)^2}{\kappa^4(1+\rho_i\mu\theta)^2}$ | $\frac{1}{\kappa^4(1+\rho_i\mu\theta)^2}$ | $-\frac{(\kappa-1)}{\kappa^4(1+\rho_i\mu\theta)^2}$ | $-\frac{(\kappa-1)}{\kappa^4(1+\rho_i\mu\theta)^2}$ |
| $c_7$ | $-\frac{4(\kappa-2)}{\kappa^4(1+\rho_i\mu\theta)(2+\rho_i\mu\theta)}$ | $-\frac{4(\kappa-2)}{\kappa^4(1+\rho_i\mu\theta)(2+\rho_i\mu\theta)}$ | $-\frac{4(\kappa-2)}{\kappa^4(1+\rho_i\mu\theta)(2+\rho_i\mu\theta)}$ | $\frac{8}{\kappa^4(1+\rho_i\mu\theta)(2+\rho_i\mu\theta)}$ |
| $c_8$ | $-\frac{4(\kappa-2)^2}{\kappa^4(1+\rho_i\mu\theta)(2+\rho_i\mu\theta)^2}$ | $\frac{8(\kappa-2)}{\kappa^4(1+\rho_i\mu\theta)(2+\rho_i\mu\theta)^2}$ | $\frac{8(\kappa-2)}{\kappa^4(1+\rho_i\mu\theta)(2+\rho_i\mu\theta)^2}$ | $\frac{8(\kappa-2)^2}{\kappa^4(1+\rho_i\mu\theta)(2+\rho_i\mu\theta)^2}$ |
| $c_9$ | $\frac{2\rho_i\mu\theta(\kappa+\rho_i\mu\theta)^2}{\kappa^4(1+\rho_i\mu\theta)^2(2+\rho_i\mu\theta)^2(3+\rho_i\mu\theta)}$ | $\frac{\rho_i\mu\theta(2\kappa^2+\kappa(3\kappa-2)\rho_i\mu\theta+(2+(\kappa-2)\kappa)\rho_i\mu^2\theta^2)}{\kappa^4(1+\rho_i\mu\theta)^2(2+\rho_i\mu\theta)^2(3+\rho_i\mu\theta)}$ | $\frac{2\rho_i\mu^2\theta^2(\kappa+\rho_i\mu\theta)}{\kappa^4(1+\rho_i\mu\theta)^2(2+\rho_i\mu\theta)^2(3+\rho_i\mu\theta)}$ | $\frac{2\rho_i\mu^2\theta^2(\kappa+\rho_i\mu\theta)}{\kappa^4(1+\rho_i\mu\theta)^2(2+\rho_i\mu\theta)^2(3+\rho_i\mu\theta)}$ |

|       | xyxz | yxxz | xyzw |
|-------|------|------|------|
| $c_0$ | $\frac{1}{\kappa^4}$ | $\frac{1}{\kappa^4}$ | $\frac{1}{\kappa^4}$ |
| $c_1$ | $-\frac{1}{\kappa^4(1+\rho_i\mu\theta)}$ | $-\frac{1}{\kappa^4(1+\rho_i\mu\theta)}$ | $-\frac{1}{\kappa^4(1+\rho_i\mu\theta)}$ |
| $c_2$ | $-\frac{2}{\kappa^4(1+\rho_i\mu\theta)}$ | $\frac{(\kappa-2)}{\kappa^4(1+\rho_i\mu\theta)}$ | $-\frac{2}{\kappa^4(1+\rho_i\mu\theta)}$ |
| $c_3$ | $\frac{4}{\kappa^4(1+\rho_i\mu\theta)(2+\rho_i\mu\theta)}$ | $-\frac{2(\kappa-2)}{\kappa^4(1+\rho_i\mu\theta)(2+\rho_i\mu\theta)}$ | $\frac{4}{\kappa^4(1+\rho_i\mu\theta)(2+\rho_i\mu\theta)}$ |
| $c_4$ | $\frac{(\kappa-3)}{\kappa^4(1+\rho_i\mu\theta)}$ | $-\frac{3}{\kappa^4(1+\rho_i\mu\theta)}$ | $-\frac{3}{\kappa^4(1+\rho_i\mu\theta)}$ |
| $c_5$ | $-\frac{2(\kappa-2)}{\kappa^4(1+\rho_i\mu\theta)(2+\rho_i\mu\theta)}$ | $\frac{4}{\kappa^4(1+\rho_i\mu\theta)(2+\rho_i\mu\theta)}$ | $\frac{4}{\kappa^4(1+\rho_i\mu\theta)(2+\rho_i\mu\theta)}$ |
| $c_6$ | $\frac{1}{\kappa^4(1+\rho_i\mu\theta)^2}$ | $\frac{1}{\kappa^4(1+\rho_i\mu\theta)^2}$ | $\frac{1}{\kappa^4(1+\rho_i\mu\theta)^2}$ |
| $c_7$ | $-\frac{2(\kappa-4)}{\kappa^4(1+\rho_i\mu\theta)(2+\rho_i\mu\theta)}$ | $-\frac{2(\kappa-4)}{\kappa^4(1+\rho_i\mu\theta)(2+\rho_i\mu\theta)}$ | $\frac{8}{\kappa^4(1+\rho_i\mu\theta)(2+\rho_i\mu\theta)}$ |
| $c_8$ | $\frac{4(\kappa-4)}{\kappa^4(1+\rho_i\mu\theta)(2+\rho_i\mu\theta)^2}$ | $\frac{4(\kappa-4)}{\kappa^4(1+\rho_i\mu\theta)(2+\rho_i\mu\theta)^2}$ | $-\frac{16}{\kappa^4(1+\rho_i\mu\theta)(2+\rho_i\mu\theta)^2}$ |
| $c_9$ | $-\frac{\rho_i\mu^2\theta^2(\kappa+(\kappa-2)\rho_i\mu\theta)}{\kappa^4(1+\rho_i\mu\theta)^2(2+\rho_i\mu\theta)^2(3+\rho_i\mu\theta)}$ | $-\frac{\rho_i\mu^2\theta^2(\kappa+(\kappa-2)\rho_i\mu\theta)}{\kappa^4(1+\rho_i\mu\theta)^2(2+\rho_i\mu\theta)^2(3+\rho_i\mu\theta)}$ | $\frac{2\rho_i\mu^3\theta^3}{\kappa^4(1+\rho_i\mu\theta)^2(2+\rho_i\mu\theta)^2(3+\rho_i\mu\theta)}$ |

One checks that for all observations $\sigma_i = i_1 i_2 i_3 i_4$, $i_j \in [\kappa]$

$$\sum_i p^{\rho_i,*}_{\sigma_i|(S,\boldsymbol{\tau})} = \kappa p^{\rho_i,*}_{xxxx} + 2\kappa(\kappa-1)p^{\rho_i,*}_{xxxy} + \kappa(\kappa-1)p^{\rho_i,*}_{xyxx} + \kappa(\kappa-1)p^{\rho_i,*}_{yxxx} + \kappa(\kappa-1)p^{\rho_i,*}_{xxyy}$$
$$+ 2\kappa(\kappa-1)p^{\rho_i,*}_{xyxy} + \kappa(\kappa-1)(\kappa-2)p^{\rho_i,*}_{xxyz} + 2\kappa(\kappa-1)(\kappa-2)p^{\rho_i,*}_{xyxz}$$
$$+ 2\kappa(\kappa-1)(\kappa-2)p^{\rho_i,*}_{yxxz} + \kappa(\kappa-1)(\kappa-2)p^{\rho_i,*}_{yzxx} + \kappa(\kappa-1)(\kappa-2)(\kappa-3)p^{\rho_i,*}_{xyzw} = 1.$$

7

## *Mathematica* Supplement

Main article title: Identifiability of the unrooted species tree topology under the coalescent model with time-reversible substitution processes, site-specific rate variation, and invariable sites.

Julia Chifman, Wake Forest School of Medicine
Laura Kubatko, The Ohio State University

- This file serves as support for Theorem 5.1 of the main article. Refer to the main article for definitions and terminology.

### Site Pattern probabilities for generalized JC69 under the coalescent model for the symmetric 4-leaf species tree with the rate of evolution at a site $\rho_i > 0$ for a specific category associated with the discrete gamma distribution.

For ease of notation we will use $\rho$ in place of $\rho_i$. In addition, recall that in the proof of Theorem 5.1 we set $\tau_3 = \tau_2$. This implies that the asymmetric species tree will coincide with the symmetric species tree, i.e. site patterns will be the same. Thus, below we list only site patterns probabilities for a symmetric species tree. With parameter choices made in our proof and $\rho > 0$, all site patterns $P_{ijkl} > 0$.

$$
\begin{aligned}
\text{Pxxxx} :={}& \frac{1}{k^4} + \frac{(e^{-\tau_1})^{2\mu\rho}(e^{-\tau_2})^{2\mu\rho}(-1+k)^2}{k^4(1+\theta\mu\rho)^2} + \frac{(e^{-\tau_1})^{2\mu\rho}(-1+k)}{k^4(1+\theta\mu\rho)} + \frac{(e^{-\tau_2})^{2\mu\rho}(-1+k)}{k^4(1+\theta\mu\rho)} + \\
& \frac{4(e^{-\tau_3})^{2\mu\rho}(-1+k)}{k^4(1+\theta\mu\rho)} + \frac{4(e^{-\tau_1})^{\mu\rho}(e^{-\tau_2})^{\mu\rho}(e^{-\tau_3})^{2\mu\rho}(-2+k)^2(-1+k)}{k^4(1+\theta\mu\rho)(2+\theta\mu\rho)^2} + \\
& \frac{4(e^{-\tau_1})^{\mu\rho}(e^{-\tau_3})^{2\mu\rho}(-2+k)(-1+k)}{k^4(1+\theta\mu\rho)(2+\theta\mu\rho)} + \frac{4(e^{-\tau_2})^{\mu\rho}(e^{-\tau_3})^{2\mu\rho}(-2+k)(-1+k)}{k^4(1+\theta\mu\rho)(2+\theta\mu\rho)} + \\
& \left(2(e^{-\tau_1})^{-2/\theta}(e^{-\tau_2})^{-2/\theta}(e^{-\tau_3})^{4\left(\frac{1}{\theta}+\mu\rho\right)}(-1+k)\theta\mu\rho(k+\theta\mu\rho)(k+(-1+k)\theta\mu\rho)\right) \Big/ \\
& \left(k^4(1+\theta\mu\rho)^2(2+\theta\mu\rho)^2(3+\theta\mu\rho)\right)
\end{aligned}
$$

$$
\begin{aligned}
\text{Pxxxy} :={}& \frac{1}{k^4} - \frac{(e^{-\tau_1})^{2\mu\rho}(e^{-\tau_2})^{2\mu\rho}(-1+k)}{k^4(1+\theta\mu\rho)^2} - \frac{(e^{-\tau_1})^{2\mu\rho}}{k^4(1+\theta\mu\rho)} + \\
& \frac{2(e^{-\tau_3})^{2\mu\rho}(-2+k)}{k^4(1+\theta\mu\rho)} + \frac{(e^{-\tau_2})^{2\mu\rho}(-1+k)}{k^4(1+\theta\mu\rho)} - \frac{4(e^{-\tau_1})^{\mu\rho}(e^{-\tau_2})^{\mu\rho}(e^{-\tau_3})^{2\mu\rho}(-2+k)^2}{k^4(1+\theta\mu\rho)(2+\theta\mu\rho)^2} - \\
& \frac{4(e^{-\tau_1})^{\mu\rho}(e^{-\tau_3})^{2\mu\rho}(-2+k)}{k^4(1+\theta\mu\rho)(2+\theta\mu\rho)} + \frac{2(e^{-\tau_2})^{\mu\rho}(e^{-\tau_3})^{2\mu\rho}(-2+k)^2}{k^4(1+\theta\mu\rho)(2+\theta\mu\rho)} - \\
& \left(2(e^{-\tau_1})^{-2/\theta}(e^{-\tau_2})^{-2/\theta}(e^{-\tau_3})^{4\left(\frac{1}{\theta}+\mu\rho\right)}\theta\mu\rho(k+\theta\mu\rho)(k+(-1+k)\theta\mu\rho)\right) \Big/ \\
& \left(k^4(1+\theta\mu\rho)^2(2+\theta\mu\rho)^2(3+\theta\mu\rho)\right)
\end{aligned}
$$

$$\text{Pxyxx} := \frac{1}{k^4} - \frac{(e^{-\tau_1})^{2\mu\rho}(e^{-\tau_2})^{2\mu\rho}(-1+k)}{k^4(1+\theta\mu\rho)^2} - \frac{(e^{-\tau_2})^{2\mu\rho}}{k^4(1+\theta\mu\rho)} +$$
$$\frac{2(e^{-\tau_3})^{2\mu\rho}(-2+k)}{k^4(1+\theta\mu\rho)} + \frac{(e^{-\tau_1})^{2\mu\rho}(-1+k)}{k^4(1+\theta\mu\rho)} - \frac{4(e^{-\tau_1})^{\mu\rho}(e^{-\tau_2})^{\mu\rho}(e^{-\tau_3})^{2\mu\rho}(-2+k)^2}{k^4(1+\theta\mu\rho)(2+\theta\mu\rho)^2} -$$
$$\frac{4(e^{-\tau_2})^{\mu\rho}(e^{-\tau_3})^{2\mu\rho}(-2+k)}{k^4(1+\theta\mu\rho)(2+\theta\mu\rho)} + \frac{2(e^{-\tau_1})^{\mu\rho}(e^{-\tau_3})^{2\mu\rho}(-2+k)^2}{k^4(1+\theta\mu\rho)(2+\theta\mu\rho)} -$$
$$\left(2(e^{-\tau_1})^{-2/\theta}(e^{-\tau_2})^{-2/\theta}(e^{-\tau_3})^{4\left(\frac{1}{\theta}+\mu\rho\right)}\theta\mu\rho(k+\theta\mu\rho)(k+(-1+k)\theta\mu\rho)\right) \Big/$$
$$\left(k^4(1+\theta\mu\rho)^2(2+\theta\mu\rho)^2(3+\theta\mu\rho)\right)$$

$$\text{Pxyxy} := \frac{1}{k^4} + \frac{(e^{-\tau_1})^{2\mu\rho}(e^{-\tau_2})^{2\mu\rho}}{k^4(1+\theta\mu\rho)^2} - \frac{(e^{-\tau_1})^{2\mu\rho}}{k^4(1+\theta\mu\rho)} -$$
$$\frac{(e^{-\tau_2})^{2\mu\rho}}{k^4(1+\theta\mu\rho)} + \frac{2(e^{-\tau_3})^{2\mu\rho}(-2+k)}{k^4(1+\theta\mu\rho)} + \frac{8(e^{-\tau_1})^{\mu\rho}(e^{-\tau_2})^{\mu\rho}(e^{-\tau_3})^{2\mu\rho}(-2+k)}{k^4(1+\theta\mu\rho)(2+\theta\mu\rho)^2} -$$
$$\frac{4(e^{-\tau_1})^{\mu\rho}(e^{-\tau_3})^{2\mu\rho}(-2+k)}{k^4(1+\theta\mu\rho)(2+\theta\mu\rho)} - \frac{4(e^{-\tau_2})^{\mu\rho}(e^{-\tau_3})^{2\mu\rho}(-2+k)}{k^4(1+\theta\mu\rho)(2+\theta\mu\rho)} +$$
$$\left((e^{-\tau_1})^{-2/\theta}(e^{-\tau_2})^{-2/\theta}(e^{-\tau_3})^{4\left(\frac{1}{\theta}+\mu\rho\right)}\theta\mu\rho\left(2k^2+k(-2+3k)\theta\mu\rho+(2+(-2+k)k)\theta^2\mu^2\rho^2\right)\right) \Big/$$
$$\left(k^4(1+\theta\mu\rho)^2(2+\theta\mu\rho)^2(3+\theta\mu\rho)\right)$$

$$\text{Pxxyy} := \frac{1}{k^4} + \frac{(e^{-\tau_1})^{2\mu\rho}(e^{-\tau_2})^{2\mu\rho}(-1+k)^2}{k^4(1+\theta\mu\rho)^2} - \frac{4(e^{-\tau_3})^{2\mu\rho}}{k^4(1+\theta\mu\rho)} + \frac{(e^{-\tau_1})^{2\mu\rho}(-1+k)}{k^4(1+\theta\mu\rho)} +$$
$$\frac{(e^{-\tau_2})^{2\mu\rho}(-1+k)}{k^4(1+\theta\mu\rho)} - \frac{4(e^{-\tau_1})^{\mu\rho}(e^{-\tau_2})^{\mu\rho}(e^{-\tau_3})^{2\mu\rho}(-2+k)^2}{k^4(1+\theta\mu\rho)(2+\theta\mu\rho)^2} - \frac{4(e^{-\tau_1})^{\mu\rho}(e^{-\tau_3})^{2\mu\rho}(-2+k)}{k^4(1+\theta\mu\rho)(2+\theta\mu\rho)} -$$
$$\frac{4(e^{-\tau_2})^{\mu\rho}(e^{-\tau_3})^{2\mu\rho}(-2+k)}{k^4(1+\theta\mu\rho)(2+\theta\mu\rho)} + \frac{2(e^{-\tau_1})^{-2/\theta}(e^{-\tau_2})^{-2/\theta}(e^{-\tau_3})^{4\left(\frac{1}{\theta}+\mu\rho\right)}\theta\mu\rho(k+\theta\mu\rho)^2}{k^4(1+\theta\mu\rho)^2(2+\theta\mu\rho)^2(3+\theta\mu\rho)}$$

$$\text{Pxxyz} := \frac{1}{k^4} - \frac{(e^{-\tau_1})^{2\mu\rho}(e^{-\tau_2})^{2\mu\rho}(-1+k)}{k^4(1+\theta\mu\rho)^2} - \frac{(e^{-\tau_1})^{2\mu\rho}}{k^4(1+\theta\mu\rho)} - \frac{4(e^{-\tau_3})^{2\mu\rho}}{k^4(1+\theta\mu\rho)} +$$
$$\frac{(e^{-\tau_2})^{2\mu\rho}(-1+k)}{k^4(1+\theta\mu\rho)} + \frac{8(e^{-\tau_1})^{\mu\rho}(e^{-\tau_2})^{\mu\rho}(e^{-\tau_3})^{2\mu\rho}(-2+k)}{k^4(1+\theta\mu\rho)(2+\theta\mu\rho)^2} + \frac{8(e^{-\tau_1})^{\mu\rho}(e^{-\tau_3})^{2\mu\rho}}{k^4(1+\theta\mu\rho)(2+\theta\mu\rho)} -$$
$$\frac{4(e^{-\tau_2})^{\mu\rho}(e^{-\tau_3})^{2\mu\rho}(-2+k)}{k^4(1+\theta\mu\rho)(2+\theta\mu\rho)} + \frac{2(e^{-\tau_1})^{-2/\theta}(e^{-\tau_2})^{-2/\theta}(e^{-\tau_3})^{4\left(\frac{1}{\theta}+\mu\rho\right)}\theta^2\mu^2\rho^2(k+\theta\mu\rho)}{k^4(1+\theta\mu\rho)^2(2+\theta\mu\rho)^2(3+\theta\mu\rho)}$$

$$\text{Pyzxx} := \frac{1}{k^4} - \frac{(e^{-\tau_1})^{2\mu\rho}(e^{-\tau_2})^{2\mu\rho}(-1+k)}{k^4(1+\theta\mu\rho)^2} - \frac{(e^{-\tau_2})^{2\mu\rho}}{k^4(1+\theta\mu\rho)} - \frac{4(e^{-\tau_3})^{2\mu\rho}}{k^4(1+\theta\mu\rho)} +$$
$$\frac{(e^{-\tau_1})^{2\mu\rho}(-1+k)}{k^4(1+\theta\mu\rho)} + \frac{8(e^{-\tau_1})^{\mu\rho}(e^{-\tau_2})^{\mu\rho}(e^{-\tau_3})^{2\mu\rho}(-2+k)}{k^4(1+\theta\mu\rho)(2+\theta\mu\rho)^2} + \frac{8(e^{-\tau_2})^{\mu\rho}(e^{-\tau_3})^{2\mu\rho}}{k^4(1+\theta\mu\rho)(2+\theta\mu\rho)} -$$
$$\frac{4(e^{-\tau_1})^{\mu\rho}(e^{-\tau_3})^{2\mu\rho}(-2+k)}{k^4(1+\theta\mu\rho)(2+\theta\mu\rho)} + \frac{2(e^{-\tau_1})^{-2/\theta}(e^{-\tau_2})^{-2/\theta}(e^{-\tau_3})^{4\left(\frac{1}{\theta}+\mu\rho\right)}\theta^2\mu^2\rho^2(k+\theta\mu\rho)}{k^4(1+\theta\mu\rho)^2(2+\theta\mu\rho)^2(3+\theta\mu\rho)}$$

$$\text{Pxyxz} := \frac{1}{k^4} + \frac{(e^{-\tau_1})^{2\mu\rho}(e^{-\tau_2})^{2\mu\rho}}{k^4(1+\theta\mu\rho)^2} - \frac{(e^{-\tau_1})^{2\mu\rho}}{k^4(1+\theta\mu\rho)} - \frac{(e^{-\tau_2})^{2\mu\rho}}{k^4(1+\theta\mu\rho)} + \frac{(e^{-\tau_3})^{2\mu\rho}(-4+k)}{k^4(1+\theta\mu\rho)} +$$
$$\frac{4(e^{-\tau_1})^{\mu\rho}(e^{-\tau_2})^{\mu\rho}(e^{-\tau_3})^{2\mu\rho}(-4+k)}{k^4(1+\theta\mu\rho)(2+\theta\mu\rho)^2} - \frac{2(e^{-\tau_1})^{\mu\rho}(e^{-\tau_3})^{2\mu\rho}(-4+k)}{k^4(1+\theta\mu\rho)(2+\theta\mu\rho)} -$$
$$\frac{2(e^{-\tau_2})^{\mu\rho}(e^{-\tau_3})^{2\mu\rho}(-4+k)}{k^4(1+\theta\mu\rho)(2+\theta\mu\rho)} - \left((e^{-\tau_1})^{-2/\theta}(e^{-\tau_2})^{-2/\theta}(e^{-\tau_3})^{4\left(\frac{1}{\theta}+\mu\rho\right)}\theta^2\mu^2\rho^2(k+(-2+k)\theta\mu\rho)\right) \Big/$$
$$\left(k^4(1+\theta\mu\rho)^2(2+\theta\mu\rho)^2(3+\theta\mu\rho)\right)$$

$$\text{Pxyzw} := \frac{1}{k^4} + \frac{(e^{-\tau_1})^{2\mu\rho}(e^{-\tau_2})^{2\mu\rho}}{k^4(1+\theta\mu\rho)^2} - \frac{(e^{-\tau_1})^{2\mu\rho}}{k^4(1+\theta\mu\rho)} - \frac{(e^{-\tau_2})^{2\mu\rho}}{k^4(1+\theta\mu\rho)} -$$
$$\frac{4(e^{-\tau_3})^{2\mu\rho}}{k^4(1+\theta\mu\rho)} - \frac{16(e^{-\tau_1})^{\mu\rho}(e^{-\tau_2})^{\mu\rho}(e^{-\tau_3})^{2\mu\rho}}{k^4(1+\theta\mu\rho)(2+\theta\mu\rho)^2} + \frac{8(e^{-\tau_1})^{\mu\rho}(e^{-\tau_3})^{2\mu\rho}}{k^4(1+\theta\mu\rho)(2+\theta\mu\rho)} +$$
$$\frac{8(e^{-\tau_2})^{\mu\rho}(e^{-\tau_3})^{2\mu\rho}}{k^4(1+\theta\mu\rho)(2+\theta\mu\rho)} + \frac{2(e^{-\tau_1})^{-2/\theta}(e^{-\tau_2})^{-2/\theta}(e^{-\tau_3})^{4\left(\frac{1}{\theta}+\mu\rho\right)}\theta^3\mu^3\rho^3}{k^4(1+\theta\mu\rho)^2(2+\theta\mu\rho)^2(3+\theta\mu\rho)}$$

# Computing Principal Minors for k = 2.

Computations below demonstrate that all principal minors of a matrix $F_\rho$ are positive for $\rho > 0$, which establishes that $F_\rho$ is positive definite. (See main article and the proof of Theorem 5.1 for the definition of the matrix $F_\rho$)

```
k := 2
θ := 1 / 10
μ := 1 / 10
τ_3 := 1
τ_2 := 1
τ_1 := 1 / 10
```

Define principal submatrices of the matrix M4 = $F_\rho$.

```
M1 := Pxxxx
M2 := {
   {Pxxxx, Pxxxy},
   {Pxxxy, Pxxyy}}
M3 := {
   {Pxxxx, Pxxxy, Pxyxx},
   {Pxxxy, Pxxyy, Pxyxy},
   {Pxyxx, Pxyxy, Pxxyy}}
M4 := {
   {Pxxxx, Pxxxy, Pxyxx, Pxyxy},
   {Pxxxy, Pxxyy, Pxyxy, Pxyxx},
   {Pxyxx, Pxyxy, Pxxyy, Pxxxy},
   {Pxyxy, Pxyxx, Pxxxy, Pxxxx}}
```

First we compute principal minors M1, M2, M3 and M4.

**M1**

$$\frac{1}{16} + \frac{e^{-11\rho/50}}{16\left(1+\frac{\rho}{100}\right)^2} + \frac{5\,e^{-\rho/5}}{16\left(1+\frac{\rho}{100}\right)} + \frac{e^{-\rho/50}}{16\left(1+\frac{\rho}{100}\right)} + \frac{e^{22-4\left(10+\frac{\rho}{10}\right)}\rho}{800\left(1+\frac{\rho}{100}\right)^2\left(3+\frac{\rho}{100}\right)}$$

**FullSimplify[Det[M2]]**

$$\frac{1}{16(100+\rho)^3(300+\rho)}\,25\,e^{-\frac{3}{5}(30+\rho)}\left(20000\,\rho + 10000\,e^{18+\frac{9\rho}{50}}(300+\rho) + 200\,e^{18+\frac{19\rho}{50}}(100+\rho)(300+\rho) + e^{18+\frac{29\rho}{50}}(100+\rho)^2(300+\rho) - 200\,e^{\rho/5}(100+\rho)\left(-\rho + 2\,e^{18}(300+\rho)\right)\right)$$

**FullSimplify[Det[M3]]**

$$\frac{1}{16(100+\rho)^4(300+\rho)}\,625\,e^{-18-\frac{4\rho}{5}}\left(-1 + e^{9\rho/50}\right)\left(20000\,\rho + 10000\,e^{18+\frac{9\rho}{50}}(300+\rho) + 100\,e^{18+\frac{19\rho}{50}}(100+\rho)(300+\rho) + e^{18+\frac{2\rho}{5}}(100+\rho)^2(300+\rho) - 100\,e^{\rho/5}(100+\rho)\left(-2\rho + 3\,e^{18}(300+\rho)\right)\right)$$

**FullSimplify[Det[M4]]**

$$\left(390625\,e^{-18-\frac{4\rho}{5}}\left(-1 + e^{9\rho/50}\right)\left(2\rho + e^{18}\left(-1 + e^{9\rho/50}\right)(300+\rho)\right)\right) \Big/ \left((100+\rho)^4(300+\rho)\right)$$

From the above computations it is clear that M1 = $P_{xxxx} > 0$ and det(M4) > 0 for any real $\rho > 0$.

Next, we rewrite determinants for M2 and M3.

**Collect$\left[20000\,\rho + 10000\,e^{18+\frac{9\rho}{50}}(300+\rho) + 200\,e^{18+\frac{19\rho}{50}}(100+\rho)(300+\rho) + e^{18+\frac{29\rho}{50}}(100+\rho)^2(300+\rho) - 200\,e^{\rho/5}(100+\rho)\left(-\rho + 2\,e^{18}(300+\rho)\right), \rho, \text{Simplify}\right]$**

$$3000000\,e^{18+\frac{9\rho}{50}}\left(1 - 4\,e^{\rho/50} + 2\,e^{\rho/5} + e^{2\rho/5}\right) + 10000\left(2 + e^{18+\frac{9\rho}{50}} - 16\,e^{18+\frac{\rho}{5}} + 8\,e^{18+\frac{19\rho}{50}} + 7\,e^{18+\frac{29\rho}{50}} + 2\,e^{\rho/5}\right)\rho + 100\,e^{\rho/5}\left(2 - 4\,e^{18} + 2\,e^{18+\frac{9\rho}{50}} + 5\,e^{18+\frac{19\rho}{50}}\right)\rho^2 + e^{18+\frac{29\rho}{50}}\rho^3$$

```
Collect[20 000 ρ + 10 000 e^(18 + 9ρ/50) (300 + ρ) + 100 e^(18 + 19ρ/50) (100 + ρ) (300 + ρ) +
  e^(18 + 2ρ/5) (100 + ρ)^2 (300 + ρ) - 100 e^(ρ/5) (100 + ρ) (-2 ρ + 3 e^18 (300 + ρ)), ρ, Simplify]
```

$$3\,000\,000\, e^{18+\frac{9\rho}{50}} \left(1 - 3\, e^{\rho/50} + e^{\rho/5} + e^{11\rho/50}\right) +$$
$$10\,000 \left(2 + e^{18+\frac{9\rho}{50}} - 12\, e^{18+\frac{\rho}{5}} + 4\, e^{18+\frac{19\rho}{50}} + 7\, e^{18+\frac{2\rho}{5}} + 2\, e^{\rho/5}\right) \rho +$$
$$100\, e^{\rho/5} \left(2 - 3\, e^{18} + e^{18+\frac{9\rho}{50}} + 5\, e^{18+\frac{\rho}{5}}\right) \rho^2 + e^{18+\frac{2\rho}{5}} \rho^3$$

Notice that each of the expressions above will be positive for all real numbers $\rho > 0$ if terms in parentheses, which are the sums of exponential functions, are all greater than zero on the same interval.
It is a straightforward exercise to see that terms in parenthesis are all positive for $\rho > 0$.

$$1 - 4\, e^{\rho/50} + 2\, e^{\rho/5} + e^{2\rho/5} > 0$$
$$2 + e^{18+\frac{9\rho}{50}} - 16\, e^{18+\frac{\rho}{5}} + 8\, e^{18+\frac{19\rho}{50}} + 7\, e^{18+\frac{29\rho}{50}} + 2\, e^{\rho/5} > 0$$
$$2 - 4\, e^{18} + 2\, e^{18+\frac{9\rho}{50}} + 5\, e^{18+\frac{19\rho}{50}} > 0$$
$$1 - 3\, e^{\rho/50} + e^{\rho/5} + e^{11\rho/50} > 0$$
$$2 + e^{18+\frac{9\rho}{50}} - 12\, e^{18+\frac{\rho}{5}} + 4\, e^{18+\frac{19\rho}{50}} + 7\, e^{18+\frac{2\rho}{5}} + 2\, e^{\rho/5} > 0$$
$$2 - 3\, e^{18} + e^{18+\frac{9\rho}{50}} + 5\, e^{18+\frac{\rho}{5}} > 0$$

To be more convincing we also use Reduce[] function to compute intervals on which these terms are positive.

```
N[Reduce[1 - 4 e^(ρ/50) + 2 e^(ρ/5) + e^(2ρ/5) > 0, {ρ}, Reals]]
N[Reduce[2 + e^(18 + 9ρ/50) - 16 e^(18 + ρ/5) + 8 e^(18 + 19ρ/50) + 7 e^(18 + 29ρ/50) + 2 e^(ρ/5) > 0, {ρ}, Reals]]
N[Reduce[2 - 4 e^18 + 2 e^(18 + 9ρ/50) + 5 e^(18 + 19ρ/50) > 0, {ρ}, Reals]]
N[Reduce[1 - 3 e^(ρ/50) + e^(ρ/5) + e^(11ρ/50) > 0, {ρ}, Reals]]
N[Reduce[2 + e^(18 + 9ρ/50) - 12 e^(18 + ρ/5) + 4 e^(18 + 19ρ/50) + 7 e^(18 + 2ρ/5) + 2 e^(ρ/5) > 0, {ρ}, Reals]]
N[Reduce[2 - 3 e^18 + e^(18 + 9ρ/50) + 5 e^(18 + ρ/5) > 0, {ρ}, Reals]]
```

$\rho < -69.3146 \,||\, \rho > 0.$

$\rho < -97.5041 \,||\, \rho > -1.49314 \times 10^{-8}$

$\rho > -1.77506$

$\rho < -54.9295 \,||\, \rho > 0.$

$\rho < -94.8988 \,||\, \rho > -2.90095 \times 10^{-8}$

$\rho > -3.52626$

Now we can rewrite det(M2) and det(M3) as follows:

$$d2 := \left(25\, e^{-\frac{3}{5}(30+\rho)} \left(3\,000\,000\, e^{18+\frac{9\rho}{50}} \left(1 - 4\, e^{\rho/50} + 2\, e^{\rho/5} + e^{2\rho/5}\right) + \right.\right.$$
$$10\,000 \left(2 + e^{18+\frac{9\rho}{50}} - 16\, e^{18+\frac{\rho}{5}} + 8\, e^{18+\frac{19\rho}{50}} + 7\, e^{18+\frac{29\rho}{50}} + 2\, e^{\rho/5}\right) \rho + $$
$$\left.\left. 100\, e^{\rho/5} \left(2 - 4\, e^{18} + 2\, e^{18+\frac{9\rho}{50}} + 5\, e^{18+\frac{19\rho}{50}}\right) \rho^2 + e^{18+\frac{29\rho}{50}} \rho^3\right)\right) \Big/ \left(16\,(100+\rho)^3\,(300+\rho)\right)$$

$$d3 := \left(625\, e^{-18-\frac{4\rho}{5}} \left(-1 + e^{9\rho/50}\right) \left(3\,000\,000\, e^{18+\frac{9\rho}{50}} \left(1 - 3\, e^{\rho/50} + e^{\rho/5} + e^{11\rho/50}\right) + \right.\right.$$
$$10\,000 \left(2 + e^{18+\frac{9\rho}{50}} - 12\, e^{18+\frac{\rho}{5}} + 4\, e^{18+\frac{19\rho}{50}} + 7\, e^{18+\frac{2\rho}{5}} + 2\, e^{\rho/5}\right) \rho + $$
$$\left.\left. 100\, e^{\rho/5} \left(2 - 3\, e^{18} + e^{18+\frac{9\rho}{50}} + 5\, e^{18+\frac{\rho}{5}}\right) \rho^2 + e^{18+\frac{2\rho}{5}} \rho^3\right)\right) \Big/ \left(16\,(100+\rho)^4\,(300+\rho)\right)$$

Now it is quite obvious that d2 and d3 are strictly positive for any real $\rho > 0$.

One also checks that d2 and d3 are equivalent to det(M2) and det(M3) respectively:

```
Simplify[d2 - Det[M2]]
Simplify[d3 - Det[M3]]
```

0

0

```
Plot[M1, {ρ, 0, 50}, PlotTheme → "Detailed",
 PlotStyle → RGBColor[1.`, 0.18`, 0.31`], AxesOrigin → {0, 0}]
```

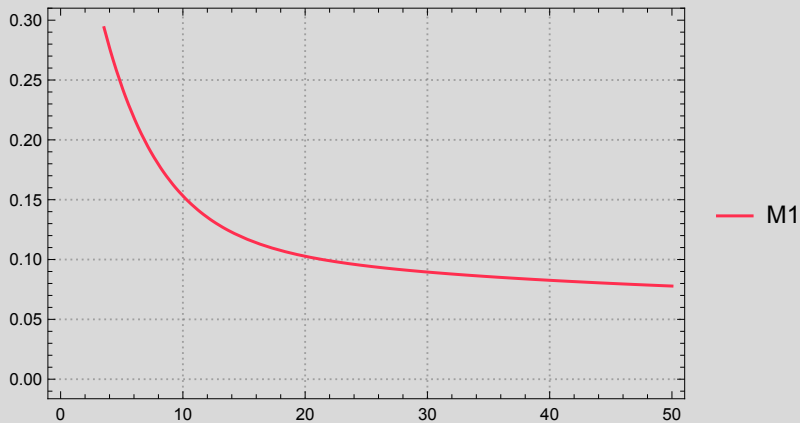

```
Plot[Det[M2], {ρ, 0, 70}, PlotTheme → "Detailed",
 PlotStyle → RGBColor[1.`, 0.18`, 0.31`], AxesOrigin → {0, 0}]
```

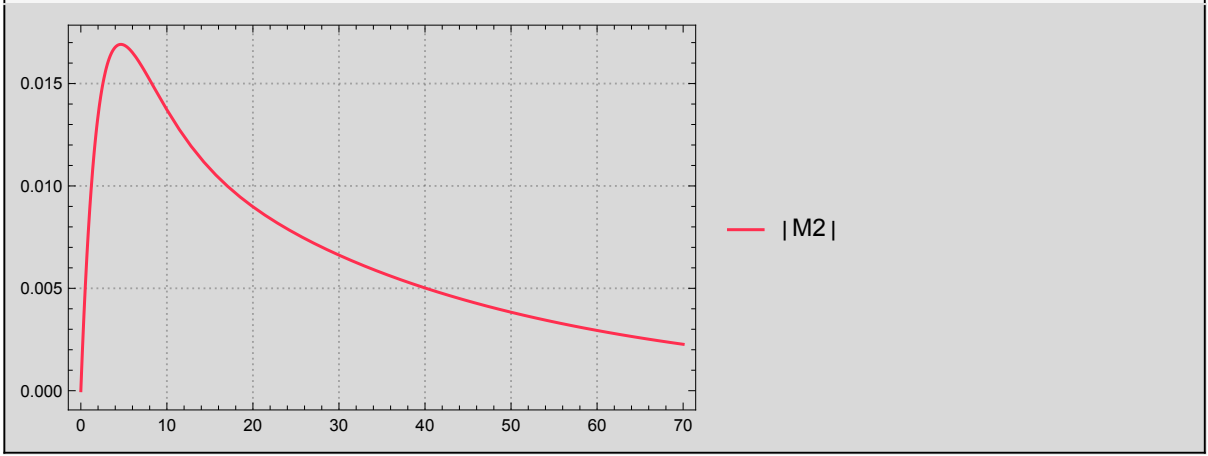

```
Plot[Det[M3], {ρ, 0, 50}, PlotTheme → "Detailed",
 PlotStyle → RGBColor[1.`, 0.18`, 0.31`], AxesOrigin → {0, 0}]
```

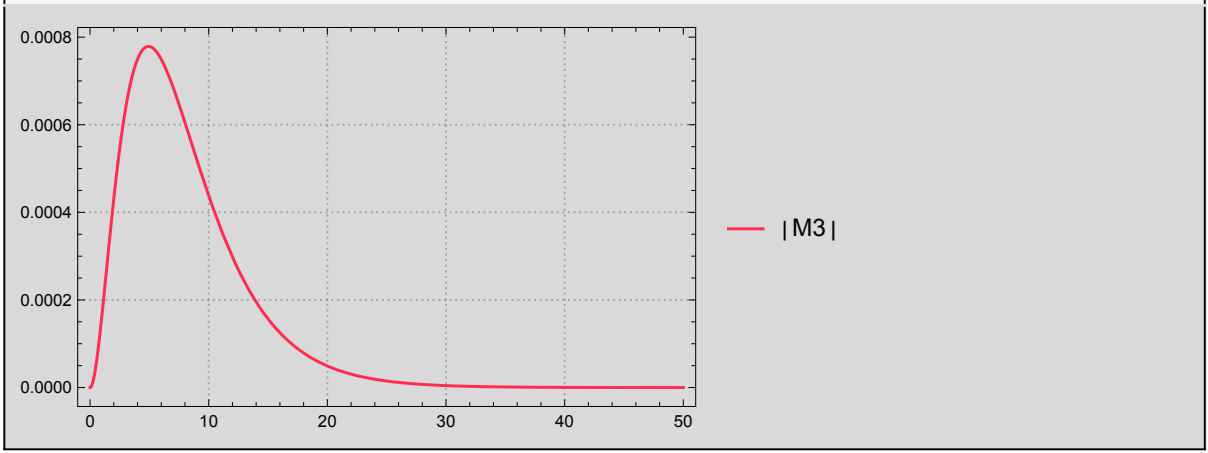

```
Plot[Det[M4], {ρ, 0, 50}, PlotTheme → "Detailed",
 PlotStyle → RGBColor[1.`, 0.18`, 0.31`], AxesOrigin → {0, 0}, PlotRange → Full]
```

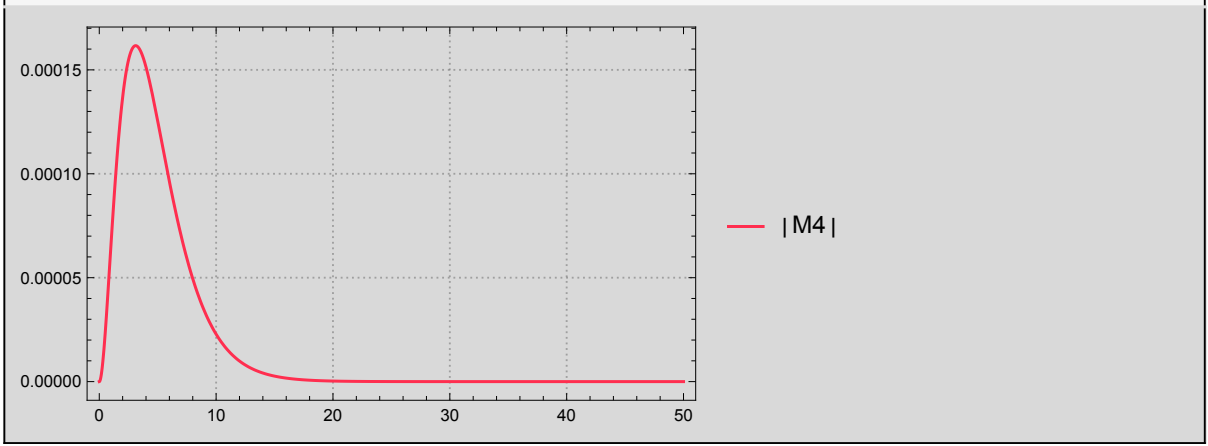

# Computing Principal Minors for k = 3.

Computations below demonstrate that all principal minors of a matrix $F^{**}_\rho$ are positive for $\rho > 0$, which establishes that $F^{**}_\rho$ is positive definite.
(See main article and the proof of Theorem 5.1 for the definition of the matrix $F^{**}_\rho$)

```
k := 3
θ := 1 / 10
μ := 1 / 10
τ₃ := 1
τ₂ := 1
τ₁ := 1 / 10
```

Define matrix $F^{**}_\rho$, call it A.
Principal 6 x 6 submatrix $F^*_\rho$ is highlighted in blue, call it B.

$$A := \begin{pmatrix} \text{Pxxyy} & \text{Pxxyz} & \text{Pxyxy} & \text{Pxyxz} & \text{Pxyxz} & \text{Pyzxx} & \text{Pxyxz} \\ \text{Pxxyz} & \text{Pxxyy} & \text{Pxyxz} & \text{Pyzxx} & \text{Pxyxy} & \text{Pxyxz} & \text{Pxyxx} \\ \text{Pxyxy} & \text{Pxyxz} & \text{Pxxyy} & \text{Pxxyz} & \text{Pyzxx} & \text{Pxyxz} & \text{Pxyxz} \\ \text{Pxyxz} & \text{Pyzxx} & \text{Pxxyz} & \text{Pxxyy} & \text{Pxyxz} & \text{Pxyxy} & \text{Pxyxx} \\ \text{Pxyxz} & \text{Pxyxy} & \text{Pyzxx} & \text{Pxyxz} & \text{Pxxyy} & \text{Pxxyz} & \text{Pxxxy} \\ \text{Pyzxx} & \text{Pxyxz} & \text{Pxyxz} & \text{Pxyxy} & \text{Pxxyz} & \text{Pxxyy} & \text{Pxxxy} \\ \text{Pxyxz} & \text{Pxyxx} & \text{Pxyxz} & \text{Pxyxx} & \text{Pxxxy} & \text{Pxxxy} & \text{Pxxxx} \end{pmatrix}$$

$$B := \begin{pmatrix} \text{Pxxyy} & \text{Pxxyz} & \text{Pxyxy} & \text{Pxyxz} & \text{Pxyxz} & \text{Pyzxx} \\ \text{Pxxyz} & \text{Pxxyy} & \text{Pxyxz} & \text{Pyzxx} & \text{Pxyxy} & \text{Pxyxz} \\ \text{Pxyxy} & \text{Pxyxz} & \text{Pxxyy} & \text{Pxxyz} & \text{Pyzxx} & \text{Pxyxz} \\ \text{Pxyxz} & \text{Pyzxx} & \text{Pxxyz} & \text{Pxxyy} & \text{Pxyxz} & \text{Pxyxy} \\ \text{Pxyxz} & \text{Pxyxy} & \text{Pyzxx} & \text{Pxyxz} & \text{Pxxyy} & \text{Pxxyz} \\ \text{Pyzxx} & \text{Pxyxz} & \text{Pxyxz} & \text{Pxyxy} & \text{Pxxyz} & \text{Pxxyy} \end{pmatrix}$$

Next, we show that B is strictly diagonally dominant. All row sums of matrix B are the same.

```
h := Pxxyy - (Pxxyz + Pxyxy + 2 Pxyxz + Pyzxx)
```

```
Collect[h, Exp[ρ], Simplify]
```

$$-\frac{4}{81} + \frac{50\,000\, e^{-11\rho/50}}{81\,(100+\rho)^2} + \frac{800\, e^{-\rho/5}}{81\,(100+\rho)} + \frac{400\, e^{-\rho/50}}{81\,(100+\rho)} - \frac{50\,000\, e^{-18-\frac{2\rho}{5}}\rho^2}{81\,(100+\rho)^2\,(200+\rho)^2} - \frac{20\,000\,000\, e^{-31\rho/100}}{81\,(100+\rho)\,(200+\rho)^2} - \frac{80\,000\, e^{-3\rho/10}}{81\,(100+\rho)\,(200+\rho)} - \frac{80\,000\, e^{-21\rho/100}}{81\,(100+\rho)\,(200+\rho)}$$

$$h2[\rho\_] := -\frac{4}{81} + \frac{50\,000\,e^{-11\rho/50}}{81\,(100+\rho)^2} + \frac{800\,e^{-\rho/5}}{81\,(100+\rho)} + \frac{400\,e^{-\rho/50}}{81\,(100+\rho)} - \frac{50\,000\,e^{-18-\frac{2\rho}{5}}\,\rho^2}{81\,(100+\rho)^2\,(200+\rho)^2} - \frac{20\,000\,000\,e^{-31\rho/100}}{81\,(100+\rho)\,(200+\rho)^2} - \frac{80\,000\,e^{-3\rho/10}}{81\,(100+\rho)\,(200+\rho)} - \frac{80\,000\,e^{-21\rho/100}}{81\,(100+\rho)\,(200+\rho)}$$

We would like to show that h2(ρ) > 0 on some interval 0 < ρ < n for positive real number n.
Notice that h2(ρ) is continuous for all real numbers ρ > 0.
The computation of real zeros of h2(ρ) and the graph of the function indicate that h2(ρ) > 0 for 0 < ρ < 8.

```
NSolve[h2[ρ] == 0, ρ, Reals]
```

{{ρ → 0.}, {ρ → 8.20656}}

```
Plot[h2[ρ], {ρ, 0, 8.5}, PlotTheme → "Detailed",
 PlotStyle → RGBColor[1.`, 0.18`, 0.31`], AxesOrigin → {0, 0}]
```

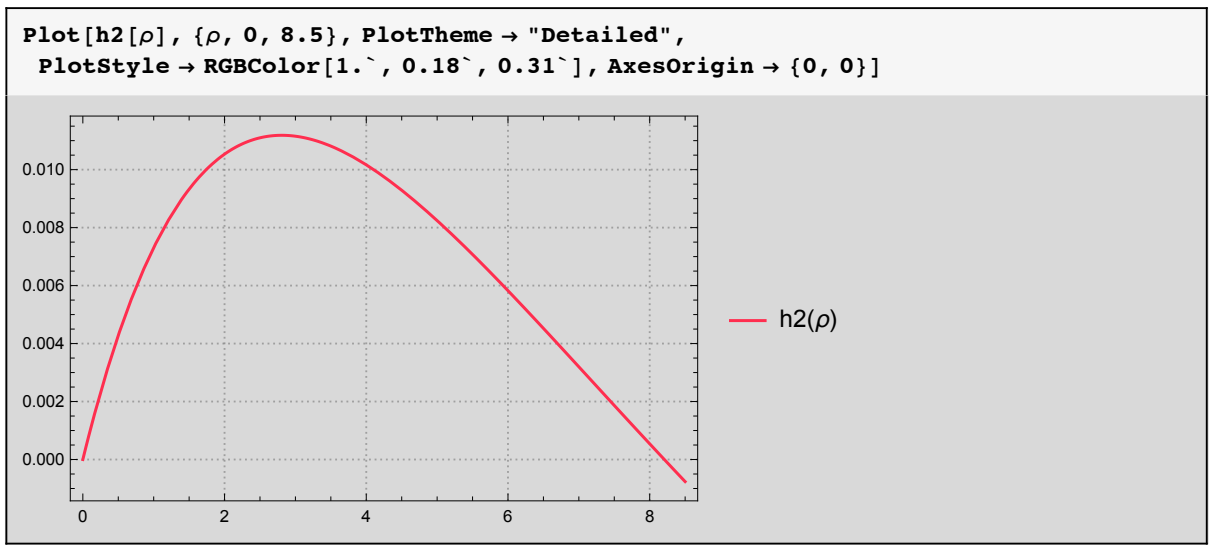

Next we compute det(A) symbolically and then substitute site pattern probabilities.

```
FullSimplify[Det[A]]
```

```
detA := (Pxxyy - Pxxyz - Pxyxy + 2 Pxyxz - Pyzxx)
   (-Pxxyz² - Pxyxy² + (Pxxyy - Pxyxz)² + Pxyxy Pyzxx - Pyzxx² + Pxxyz (Pxyxy + Pyzxx))
   (-2 Pxxxy² (-(Pxxyz + Pxyxz)² + (Pxxyy + Pxyxy) (Pxxyy + Pyzxx)) +
     Pxxxx (Pxxyy + Pxxyz + Pxyxy + 2 Pxyxz + Pyzxx) (-Pxxyz² - Pxyxy² + (Pxxyy - Pxyxz)² +
        Pxyxy Pyzxx - Pyzxx² + Pxxyz (Pxyxy + Pyzxx)) - 2 (-2 Pxxyz² Pxyxx Pxyxz +
        Pxxyy² (Pxyxx² + Pxyxz²) + Pxxyz (Pxyxx² Pxyxy - 2 Pxyxx Pxyxz² + Pxyxz² Pyzxx) +
        Pxxyy (Pxyxx² Pxyxy + Pxxyz (Pxyxx - Pxyxz)² - 2 Pxyxx Pxyxz² + Pxyxz² Pyzxx) -
        (Pxyxz (Pxyxy + Pxyxz) - Pxyxx (Pxyxz + Pyzxx))²) +
     4 Pxxxy (Pxyxz (-(Pxxyz + Pxyxz) (Pxyxy + Pxyxz) + Pxyxz Pyzxx + Pyzxx²) +
        Pxxyy (Pxyxx (Pxyxy + Pxyxz) + Pxyxz (Pxyxz + Pyzxx)) +
        Pxyxx (Pxyxy² + Pxyxy Pxyxz - (Pxxyz + Pxyxz) (Pxyxz + Pyzxx))))
```

**Simplify[detA]**

$$\frac{1}{129\,140\,163\,(100+\rho)^{11}\,(200+\rho)^{8}}$$

$$10\,000\,000\,000\,000\,000\,e^{-72-\frac{11\rho}{5}}\left(-1+e^{9\rho/100}\right)^{2}\left(-\rho^{2}-400\,e^{18+\frac{9\rho}{100}}(100+\rho)+e^{18+\frac{9\rho}{50}}(200+\rho)^{2}\right)$$

$$\left(40\,000\,e^{18}(100+\rho)-40\,000\,e^{18+\frac{9\rho}{100}}(100+\rho)+100\,\rho(200+\rho)+100\,e^{9\rho/100}\rho(200+\rho)+\right.$$
$$100\,e^{18+\frac{9\rho}{50}}(200+\rho)^{2}+100\,e^{18+\frac{27\rho}{100}}(200+\rho)^{2}+e^{18+\frac{\rho}{5}}(100+\rho)(200+\rho)^{2}+e^{18+\frac{29\rho}{100}}$$
$$\left.(100+\rho)(200+\rho)^{2}-400\,e^{18+\frac{\rho}{10}}(20\,000+300\,\rho+\rho^{2})-400\,e^{18+\frac{19\rho}{100}}(20\,000+300\,\rho+\rho^{2})\right)$$

$$\left(-6\,400\,000\,000\,e^{36+\frac{9\rho}{50}}(100+\rho)^{2}-8\,000\,000\,e^{36+\frac{27\rho}{100}}(100+\rho)(200+\rho)^{2}+\right.$$
$$80\,000\,e^{18+\frac{27\rho}{100}}\rho(150+\rho)(200+\rho)^{2}+800\,e^{36+\frac{3\rho}{10}}(100+\rho)^{2}(200+\rho)^{3}-$$
$$400\,e^{36+\frac{57\rho}{100}}(100+\rho)^{2}(200+\rho)^{3}+20\,000\,e^{36+\frac{9\rho}{20}}(200+\rho)^{4}-1000\,e^{36+\frac{47\rho}{100}}(100+\rho)$$
$$(200+\rho)^{4}+100\,e^{36+\frac{13\rho}{20}}(100+\rho)(200+\rho)^{4}-2\,e^{36+\frac{2\rho}{5}}(100+\rho)^{2}(200+\rho)^{4}+$$
$$5\,e^{36+\frac{67\rho}{100}}(100+\rho)^{2}(200+\rho)^{4}-40\,000\,000\,e^{18+\frac{9\rho}{100}}\rho\left(12\,000+220\,\rho+\rho^{2}\right)+$$
$$400\,e^{18+\frac{47\rho}{100}}\rho(200+\rho)^{2}\left(15\,000+250\,\rho+\rho^{2}\right)-80\,000\,e^{36+\frac{\rho}{5}}\left(20\,000+300\,\rho+\rho^{2}\right)^{2}+$$
$$e^{18+\frac{2\rho}{5}}\rho(1200+7\,\rho)\left(20\,000+300\,\rho+\rho^{2}\right)^{2}+e^{18+\frac{49\rho}{100}}\rho(1200+7\,\rho)\left(20\,000+300\,\rho+\rho^{2}\right)^{2}+$$
$$5\,e^{36+\frac{29\rho}{50}}\left(20\,000+300\,\rho+\rho^{2}\right)^{2}\left(48\,000+480\,\rho+\rho^{2}\right)+20\,000\,e^{36+\frac{9\rho}{25}}(200+\rho)^{2}$$
$$\left(80\,000+800\,\rho+\rho^{2}\right)-2\,e^{36+\frac{49\rho}{100}}\left(20\,000+300\,\rho+\rho^{2}\right)^{2}\left(180\,000+1800\,\rho+\rho^{2}\right)+$$
$$20\,000\,\rho^{2}\left(80\,000+1000\,\rho+3\,\rho^{2}\right)+20\,000\,e^{9\rho/100}\rho^{2}\left(80\,000+1000\,\rho+3\,\rho^{2}\right)-$$
$$400\,e^{18+\frac{3\rho}{10}}\rho(100+\rho)^{2}\left(80\,000+1000\,\rho+3\,\rho^{2}\right)-400\,e^{18+\frac{39\rho}{100}}\rho(100+\rho)^{2}$$
$$\left(80\,000+1000\,\rho+3\,\rho^{2}\right)+80\,000\,e^{36+\frac{29\rho}{100}}(100+\rho)^{2}\left(380\,000+3800\,\rho+9\,\rho^{2}\right)+$$
$$800\,e^{36+\frac{39\rho}{100}}(100+\rho)^{2}\left(4\,000\,000+60\,000\,\rho+400\,\rho^{2}+\rho^{3}\right)+$$
$$80\,000\,e^{18+\frac{9\rho}{50}}\rho\left(8\,000\,000+130\,000\,\rho+650\,\rho^{2}+\rho^{3}\right)+$$
$$100\,e^{36+\frac{14\rho}{25}}(200+\rho)^{2}\left(8\,000\,000+160\,000\,\rho+900\,\rho^{2}+\rho^{3}\right)-$$
$$400\,e^{36+\frac{12\rho}{25}}(100+\rho)^{2}\left(16\,000\,000+240\,000\,\rho+1000\,\rho^{2}+\rho^{3}\right)+$$
$$100\,e^{\rho/5}\rho^{2}\left(8\,000\,000+180\,000\,\rho+1300\,\rho^{2}+3\,\rho^{3}\right)+$$
$$100\,e^{29\rho/100}\rho^{2}\left(8\,000\,000+180\,000\,\rho+1300\,\rho^{2}+3\,\rho^{3}\right)-$$
$$1000\,e^{18+\frac{\rho}{5}}\rho\left(640\,000\,000+15\,600\,000\,\rho+140\,000\,\rho^{2}+580\,\rho^{3}+\rho^{4}\right)+$$
$$400\,e^{18+\frac{19\rho}{50}}\rho\left(800\,000\,000+21\,000\,000\,\rho+195\,000\,\rho^{2}+750\,\rho^{3}+\rho^{4}\right)-$$
$$1000\,e^{18+\frac{29\rho}{100}}\rho\left(880\,000\,000+22\,400\,000\,\rho+204\,000\,\rho^{2}+780\,\rho^{3}+\rho^{4}\right)-$$
$$\left.1000\,e^{36+\frac{19\rho}{50}}\left(160\,000\,000\,000+4\,800\,000\,000\,\rho+55\,600\,000\,\rho^{2}+312\,000\,\rho^{3}+860\,\rho^{4}+\rho^{5}\right)\right)$$

There are several products of functions in the above expression. Two of them are clearly positive for any real number $\rho > 0$. They are listed below.

$$\frac{1}{129\,140\,163\,(100+\rho)^{11}\,(200+\rho)^{8}}\,10\,000\,000\,000\,000\,000\,e^{-72-\frac{11\rho}{5}} > 0 \quad \text{and} \quad \left(-1+e^{9\rho/100}\right)^{2} > 0$$

We have to check if the other three products are positive. First we define them as p0, p1 and p2.

```
p0 := -ρ^2 - 400 e^(18 + 9ρ/100) (100 + ρ) + e^(18 + 9ρ/50) (200 + ρ)^2

p1 := 40000 e^18 (100 + ρ) - 40000 e^(18 + 9ρ/100) (100 + ρ) +
    100 ρ (200 + ρ) + 100 e^(9ρ/100) ρ (200 + ρ) + 100 e^(18 + 9ρ/50) (200 + ρ)^2 +
    100 e^(18 + 27ρ/100) (200 + ρ)^2 + e^(18 + ρ/5) (100 + ρ) (200 + ρ)^2 + e^(18 + 29ρ/100) (100 + ρ) (200 + ρ)^2 -
    400 e^(18 + ρ/10) (20000 + 300 ρ + ρ^2) - 400 e^(18 + 19ρ/100) (20000 + 300 ρ + ρ^2)

p2 := -6400000000 e^(36 + 9ρ/50) (100 + ρ)^2 -
    8000000 e^(36 + 27ρ/100) (100 + ρ) (200 + ρ)^2 + 80000 e^(18 + 27ρ/100) ρ (150 + ρ) (200 + ρ)^2 +
    800 e^(36 + 3ρ/10) (100 + ρ)^2 (200 + ρ)^3 - 400 e^(36 + 57ρ/100) (100 + ρ)^2 (200 + ρ)^3 +
    20000 e^(36 + 9ρ/20) (200 + ρ)^4 - 1000 e^(36 + 47ρ/100) (100 + ρ) (200 + ρ)^4 +
    100 e^(36 + 13ρ/20) (100 + ρ) (200 + ρ)^4 - 2 e^(36 + 2ρ/5) (100 + ρ)^2 (200 + ρ)^4 +
    5 e^(36 + 67ρ/100) (100 + ρ)^2 (200 + ρ)^4 - 40000000 e^(18 + 9ρ/100) ρ (12000 + 220 ρ + ρ^2) +
    400 e^(18 + 47ρ/100) ρ (200 + ρ)^2 (15000 + 250 ρ + ρ^2) - 80000 e^(36 + ρ/5) (20000 + 300 ρ + ρ^2)^2 +
    e^(18 + 2ρ/5) ρ (1200 + 7 ρ) (20000 + 300 ρ + ρ^2)^2 + e^(18 + 49ρ/100) ρ (1200 + 7 ρ) (20000 + 300 ρ + ρ^2)^2 +
    5 e^(36 + 29ρ/50) (20000 + 300 ρ + ρ^2)^2 (48000 + 480 ρ + ρ^2) +
    20000 e^(36 + 9ρ/25) (200 + ρ)^2 (80000 + 800 ρ + ρ^2) -
    2 e^(36 + 49ρ/100) (20000 + 300 ρ + ρ^2)^2 (180000 + 1800 ρ + ρ^2) + 20000 ρ^2 (80000 + 1000 ρ + 3 ρ^2) +
    20000 e^(9ρ/100) ρ^2 (80000 + 1000 ρ + 3 ρ^2) - 400 e^(18 + 3ρ/10) ρ (100 + ρ)^2 (80000 + 1000 ρ + 3 ρ^2) -
    400 e^(18 + 39ρ/100) ρ (100 + ρ)^2 (80000 + 1000 ρ + 3 ρ^2) +
    80000 e^(36 + 29ρ/100) (100 + ρ)^2 (380000 + 3800 ρ + 9 ρ^2) +
    800 e^(36 + 39ρ/100) (100 + ρ)^2 (4000000 + 60000 ρ + 400 ρ^2 + ρ^3) +
    80000 e^(18 + 9ρ/50) ρ (8000000 + 130000 ρ + 650 ρ^2 + ρ^3) +
    100 e^(36 + 14ρ/25) (200 + ρ)^2 (8000000 + 160000 ρ + 900 ρ^2 + ρ^3) -
    400 e^(36 + 12ρ/25) (100 + ρ)^2 (16000000 + 240000 ρ + 1000 ρ^2 + ρ^3) +
    100 e^(ρ/5) ρ^2 (8000000 + 180000 ρ + 1300 ρ^2 + 3 ρ^3) +
    100 e^(29ρ/100) ρ^2 (8000000 + 180000 ρ + 1300 ρ^2 + 3 ρ^3) -
    1000 e^(18 + ρ/5) ρ (640000000 + 15600000 ρ + 140000 ρ^2 + 580 ρ^3 + ρ^4) +
    400 e^(18 + 19ρ/50) ρ (800000000 + 21000000 ρ + 195000 ρ^2 + 750 ρ^3 + ρ^4) -
    1000 e^(18 + 29ρ/100) ρ (880000000 + 22400000 ρ + 204000 ρ^2 + 780 ρ^3 + ρ^4) -
    1000 e^(36 + 19ρ/50) (160000000000 + 4800000000 ρ + 55600000 ρ^2 + 312000 ρ^3 + 860 ρ^4 + ρ^5)
```

By rewriting p0 one can easily see that p0 > 0 for any real number $\rho$ > 0.

```
Collect[p0, ρ, Simplify]
```

$40\,000\, e^{18+\frac{9\rho}{100}} \left(-1 + e^{9\rho/100}\right) + 400\, e^{18+\frac{9\rho}{100}} \left(-1 + e^{9\rho/100}\right) \rho + \left(-1 + e^{18+\frac{9\rho}{50}}\right) \rho^2$

Next, collect and simplify p1.
Notice that expression below will be positive for all real numbers $\rho$ > 0 if terms in parentheses, which are the sums of exponential functions, are all greater than zero on the same interval.
It is easy to see that these terms are positive for $\rho$ > 0.

```
Collect[p1, ρ, Simplify]
```

$4\,000\,000\, e^{18} \left(1 - e^{9\rho/100} - 2\, e^{\rho/10} + e^{9\rho/50} - 2\, e^{19\rho/100} + e^{\rho/5} + e^{27\rho/100} + e^{29\rho/100}\right) + 20\,000$
$\left(1 + 2\, e^{18} - 2\, e^{18+\frac{9\rho}{100}} - 6\, e^{18+\frac{\rho}{10}} + 2\, e^{18+\frac{9\rho}{50}} - 6\, e^{18+\frac{19\rho}{100}} + 4\, e^{18+\frac{\rho}{5}} + 2\, e^{18+\frac{27\rho}{100}} + 4\, e^{18+\frac{29\rho}{100}} + e^{9\rho/100}\right) \rho +$
$100 \left(1 - 4\, e^{18+\frac{\rho}{10}} + e^{18+\frac{9\rho}{50}} + 5\, e^{18+\frac{\rho}{5}}\right) \left(1 + e^{9\rho/100}\right) \rho^2 + e^{18+\frac{\rho}{5}} \left(1 + e^{9\rho/100}\right) \rho^3$

$1 - e^{9\rho/100} - 2\, e^{\rho/10} + e^{9\rho/50} - 2\, e^{19\rho/100} + e^{\rho/5} + e^{27\rho/100} + e^{29\rho/100} > 0$

$1 + 2\, e^{18} - 2\, e^{18+\frac{9\rho}{100}} - 6\, e^{18+\frac{\rho}{10}} + 2\, e^{18+\frac{9\rho}{50}} - 6\, e^{18+\frac{19\rho}{100}} + 4\, e^{18+\frac{\rho}{5}} + 2\, e^{18+\frac{27\rho}{100}} + 4\, e^{18+\frac{29\rho}{100}} + e^{9\rho/100} > 0$

$100 \left(1 - 4\, e^{18+\frac{\rho}{10}} + e^{18+\frac{9\rho}{50}} + 5\, e^{18+\frac{\rho}{5}}\right) \left(1 + e^{9\rho/100}\right) > 0$

$1 + e^{9\rho/100} > 0$

To be more convincing we also use Reduce[] function to compute intervals on which these terms are positive.

```
N[Reduce[1 - e^(9ρ/100) - 2 e^(ρ/10) + e^(9ρ/50) - 2 e^(19ρ/100) + e^(ρ/5) + e^(27ρ/100) + e^(29ρ/100) > 0, {ρ}, Reals]]
N[Reduce[1 + 2 e^18 - 2 e^(18+9ρ/100) - 6 e^(18+ρ/10) + 2 e^(18+9ρ/50) -
    6 e^(18+19ρ/100) + 4 e^(18+ρ/5) + 2 e^(18+27ρ/100) + 4 e^(18+29ρ/100) + e^(9ρ/100) > 0, {ρ}, Reals]]
N[Reduce[100 (1 - 4 e^(18+ρ/10) + e^(18+9ρ/50) + 5 e^(18+ρ/5)) (1 + e^(9ρ/100)) > 0, {ρ}, Reals]]
N[Reduce[1 + e^(9ρ/100) > 0, {ρ}, Reals]]
```

$\rho < -10.245 \,||\, \rho > 0.$

$\rho < -13.6935 \,||\, \rho > -3.24042 \times 10^{-8}$

$\rho < -193.863 \,||\, \rho > -4.19963$

True

Thus, we see that p1 > 0 for any real $\rho$ > 0.

Next, we collect and simplify p2 and investigate if functions in parentheses, which are the sums of exponential functions, are positive for all real numbers $\rho > 0$.

```
Collect[p2, ρ, Simplify]
```

$16\,000\,000\,000\,000\, e^{36+\frac{9\rho}{50}} \left(-4 - 2 e^{\rho/50} - 2 e^{9\rho/100} + 19 e^{11\rho/100} + 4 e^{3\rho/25} + 4 e^{9\rho/50} - 10 e^{\rho/5} + 2 e^{21\rho/100} - 2 e^{11\rho/50} + 2 e^{27\rho/100} - 10 e^{29\rho/100} - 4 e^{3\rho/10} - 9 e^{31\rho/100} + 2 e^{19\rho/50} - 2 e^{39\rho/100} + 6 e^{2\rho/5} + e^{47\rho/100} + 5 e^{49\rho/100}\right) +$

$80\,000\,000\,000\, e^{18+\frac{9\rho}{100}} \left(-6 - 16 e^{18+\frac{9\rho}{100}} - 12 e^{18+\frac{11\rho}{100}} - 8 e^{18+\frac{9\rho}{50}} + 114 e^{18+\frac{\rho}{5}} + 28 e^{18+\frac{21\rho}{100}} + 16 e^{18+\frac{27\rho}{100}} - 60 e^{18+\frac{29\rho}{100}} + 14 e^{18+\frac{3\rho}{10}} - 16 e^{18+\frac{31\rho}{100}} + 8 e^{18+\frac{9\rho}{25}} - 60 e^{18+\frac{19\rho}{50}} - 28 e^{18+\frac{39\rho}{100}} - 72 e^{18+\frac{2\rho}{5}} + 12 e^{18+\frac{47\rho}{100}} - 14 e^{18+\frac{12\rho}{25}} + 48 e^{18+\frac{49\rho}{100}} + 6 e^{18+\frac{14\rho}{25}} + 40 e^{18+\frac{29\rho}{50}} + 8 e^{9\rho/100} - 8 e^{11\rho/100} + 6 e^{9\rho/50} - 11 e^{\rho/5} - 4 e^{21\rho/100} + 4 e^{29\rho/100} - 4 e^{3\rho/10} + 6 e^{31\rho/100} + 3 e^{19\rho/50} + 6 e^{2\rho/5}\right) \rho +$

$400\,000\,000 \left(4 - 22 e^{18+\frac{9\rho}{100}} + 26 e^{18+\frac{9\rho}{50}} - 16 e^{36+\frac{9\rho}{50}} - 39 e^{18+\frac{\rho}{5}} - 26 e^{36+\frac{\rho}{5}} + 20 e^{18+\frac{27\rho}{100}} - 10 e^{36+\frac{27\rho}{100}} - 56 e^{18+\frac{29\rho}{100}} + 246 e^{36+\frac{29\rho}{100}} - 26 e^{18+\frac{3\rho}{10}} + 76 e^{36+\frac{3\rho}{10}} + 22 e^{36+\frac{9\rho}{25}} + 21 e^{18+\frac{19\rho}{50}} - 139 e^{36+\frac{19\rho}{50}} - 26 e^{18+\frac{39\rho}{100}} + 40 e^{36+\frac{39\rho}{100}} + 43 e^{18+\frac{2\rho}{5}} - 52 e^{36+\frac{2\rho}{5}} + 12 e^{36+\frac{9\rho}{20}} + 16 e^{18+\frac{47\rho}{100}} - 140 e^{36+\frac{47\rho}{100}} - 74 e^{36+\frac{12\rho}{25}} + 43 e^{18+\frac{49\rho}{100}} - 227 e^{36+\frac{49\rho}{100}} + 27 e^{36+\frac{14\rho}{25}} - 38 e^{36+\frac{57\rho}{100}} + 155 e^{36+\frac{29\rho}{50}} + 14 e^{36+\frac{13\rho}{20}} + 130 e^{36+\frac{67\rho}{100}} + 4 e^{9\rho/100} + 2 e^{\rho/5} + 2 e^{29\rho/100}\right) \rho^2 +$

$2\,000\,000 \left(10 - 20 e^{18+\frac{9\rho}{100}} + 26 e^{18+\frac{9\rho}{50}} - 70 e^{18+\frac{\rho}{5}} - 24 e^{36+\frac{\rho}{5}} + 22 e^{18+\frac{27\rho}{100}} - 4 e^{36+\frac{27\rho}{100}} - 102 e^{18+\frac{29\rho}{100}} + 224 e^{36+\frac{29\rho}{100}} - 62 e^{18+\frac{3\rho}{10}} + 100 e^{36+\frac{3\rho}{10}} + 12 e^{36+\frac{9\rho}{25}} + 39 e^{18+\frac{19\rho}{50}} - 156 e^{36+\frac{19\rho}{50}} - 62 e^{18+\frac{39\rho}{100}} + 60 e^{36+\frac{39\rho}{100}} + 120 e^{18+\frac{2\rho}{5}} - 88 e^{36+\frac{2\rho}{5}} + 8 e^{36+\frac{9\rho}{20}} + 31 e^{18+\frac{47\rho}{100}} - 160 e^{36+\frac{47\rho}{100}} - 90 e^{36+\frac{12\rho}{25}} + 120 e^{18+\frac{49\rho}{100}} - 354 e^{36+\frac{49\rho}{100}} + 28 e^{36+\frac{14\rho}{25}} - 50 e^{36+\frac{57\rho}{100}} + 258 e^{36+\frac{29\rho}{50}} + 16 e^{36+\frac{13\rho}{20}} + 220 e^{36+\frac{67\rho}{100}} + 10 e^{9\rho/100} + 9 e^{\rho/5} + 9 e^{29\rho/100}\right) \rho^3 +$

$10\,000 \left(6 + 8 e^{18+\frac{9\rho}{50}} - 58 e^{18+\frac{\rho}{5}} - 8 e^{36+\frac{\rho}{5}} + 8 e^{18+\frac{27\rho}{100}} - 78 e^{18+\frac{29\rho}{100}} + 72 e^{36+\frac{29\rho}{100}} - 64 e^{18+\frac{3\rho}{10}} + 64 e^{36+\frac{3\rho}{10}} + 2 e^{36+\frac{9\rho}{25}} + 30 e^{18+\frac{19\rho}{50}} - 86 e^{36+\frac{19\rho}{50}} - 64 e^{18+\frac{39\rho}{100}} + 48 e^{36+\frac{39\rho}{100}} + 163 e^{18+\frac{2\rho}{5}} - 82 e^{36+\frac{2\rho}{5}} + 2 e^{36+\frac{9\rho}{20}} + 26 e^{18+\frac{47\rho}{100}} - 90 e^{36+\frac{47\rho}{100}} - 48 e^{36+\frac{12\rho}{25}} + 163 e^{18+\frac{49\rho}{100}} - 278 e^{36+\frac{49\rho}{100}} + 13 e^{36+\frac{14\rho}{25}} - 32 e^{36+\frac{57\rho}{100}} + 233 e^{36+\frac{29\rho}{50}} + 9 e^{36+\frac{13\rho}{20}} + 205 e^{36+\frac{67\rho}{100}} + 6 e^{9\rho/100} + 13 e^{\rho/5} + 13 e^{29\rho/100}\right) \rho^4 +$

$100\, e^{\rho/5} \left(3 - 10 e^{18} - 10 e^{18+\frac{9\rho}{100}} - 12 e^{18+\frac{\rho}{10}} + 8 e^{36+\frac{\rho}{10}} + 4 e^{18+\frac{9\rho}{50}} - 10 e^{36+\frac{9\rho}{50}} - 12 e^{18+\frac{19\rho}{100}} + 8 e^{36+\frac{19\rho}{100}} + 54 e^{18+\frac{\rho}{5}} - 20 e^{36+\frac{\rho}{5}} + 4 e^{18+\frac{27\rho}{100}} - 10 e^{36+\frac{27\rho}{100}} - 4 e^{36+\frac{7\rho}{25}} + 54 e^{18+\frac{29\rho}{100}} - 48 e^{36+\frac{29\rho}{100}} + e^{36+\frac{9\rho}{25}} - 4 e^{36+\frac{37\rho}{100}} + 54 e^{36+\frac{19\rho}{50}} + e^{36+\frac{9\rho}{20}} + 50 e^{36+\frac{47\rho}{100}} + 3 e^{9\rho/100}\right) \rho^5 +$

$e^{18+\frac{2\rho}{5}} \left(7 - 2 e^{18} - 2 e^{18+\frac{9\rho}{100}} + 5 e^{18+\frac{9\rho}{50}} + 5 e^{18+\frac{27\rho}{100}} + 7 e^{9\rho/100}\right) \rho^6$

Use Reduce[] function to compute intervals on which these terms are positive.

```mathematica
N[Reduce[-4 - 2 e^(ρ/50) - 2 e^(9 ρ/100) + 19 e^(11 ρ/100) + 4 e^(3 ρ/25) + 4 e^(9 ρ/50) -
    10 e^(ρ/5) + 2 e^(21 ρ/100) - 2 e^(11 ρ/50) + 2 e^(27 ρ/100) - 10 e^(29 ρ/100) - 4 e^(3 ρ/10) - 9 e^(31 ρ/100) +
    2 e^(19 ρ/50) - 2 e^(39 ρ/100) + 6 e^(2 ρ/5) + e^(47 ρ/100) + 5 e^(49 ρ/100) > 0, {ρ}, Reals]]
N[Reduce[-6 - 16 e^(18+9ρ/100) - 12 e^(18+11ρ/100) - 8 e^(18+9ρ/50) + 114 e^(18+ρ/5) + 28 e^(18+21ρ/100) + 16 e^(18+27ρ/100) -
    60 e^(18+29ρ/100) + 14 e^(18+3ρ/10) - 16 e^(18+31ρ/100) + 8 e^(18+9ρ/25) - 60 e^(18+19ρ/50) - 28 e^(18+39ρ/100) - 72 e^(18+2ρ/5) + 12 e^(18+47ρ/100) -
    14 e^(18+12ρ/25) + 48 e^(18+49ρ/100) + 6 e^(18+14ρ/25) + 40 e^(18+29ρ/50) + 8 e^(9ρ/100) - 8 e^(11ρ/100) + 6 e^(9ρ/50) - 11 e^(ρ/5) -
    4 e^(21ρ/100) + 4 e^(29ρ/100) - 4 e^(3ρ/10) + 6 e^(31ρ/100) + 3 e^(19ρ/50) + 6 e^(2ρ/5) > 0, {ρ}, Reals]]
N[Reduce[4 - 22 e^(18+9ρ/100) + 26 e^(18+9ρ/50) - 16 e^(36+9ρ/50) - 39 e^(18+ρ/5) - 26 e^(36+ρ/5) + 20 e^(18+27ρ/100) -
    10 e^(36+27ρ/100) - 56 e^(18+29ρ/100) + 246 e^(36+29ρ/100) - 26 e^(18+3ρ/10) + 76 e^(36+3ρ/10) + 22 e^(36+9ρ/25) + 21 e^(18+19ρ/50) -
    139 e^(36+19ρ/50) - 26 e^(18+39ρ/100) + 40 e^(36+39ρ/100) + 43 e^(18+2ρ/5) - 52 e^(36+2ρ/5) + 12 e^(36+9ρ/20) + 16 e^(18+47ρ/100) -
    140 e^(36+47ρ/100) - 74 e^(36+12ρ/25) + 43 e^(18+49ρ/100) - 227 e^(36+49ρ/100) + 27 e^(36+14ρ/25) - 38 e^(36+57ρ/100) + 155 e^(36+29ρ/50) +
    14 e^(36+13ρ/20) + 130 e^(36+67ρ/100) + 4 e^(9ρ/100) + 2 e^(ρ/5) + 2 e^(29ρ/100) > 0, {ρ}, Reals]]
N[Reduce[10 - 20 e^(18+9ρ/100) + 26 e^(18+9ρ/50) - 70 e^(18+ρ/5) - 24 e^(36+ρ/5) + 22 e^(18+27ρ/100) - 4 e^(36+27ρ/100) -
    102 e^(18+29ρ/100) + 224 e^(36+29ρ/100) - 62 e^(18+3ρ/10) + 100 e^(36+3ρ/10) + 12 e^(36+9ρ/25) + 39 e^(18+19ρ/50) - 156 e^(36+19ρ/50) -
    62 e^(18+39ρ/100) + 60 e^(36+39ρ/100) + 120 e^(18+2ρ/5) - 88 e^(36+2ρ/5) + 8 e^(36+9ρ/20) + 31 e^(18+47ρ/100) - 160 e^(36+47ρ/100) -
    90 e^(36+12ρ/25) + 120 e^(18+49ρ/100) - 354 e^(36+49ρ/100) + 28 e^(36+14ρ/25) - 50 e^(36+57ρ/100) + 258 e^(36+29ρ/50) +
    16 e^(36+13ρ/20) + 220 e^(36+67ρ/100) + 10 e^(9ρ/100) + 9 e^(ρ/5) + 9 e^(29ρ/100) > 0, {ρ}, Reals]]
N[Reduce[6 + 8 e^(18+9ρ/50) - 58 e^(18+ρ/5) - 8 e^(36+ρ/5) + 8 e^(18+27ρ/100) - 78 e^(18+29ρ/100) + 72 e^(36+29ρ/100) -
    64 e^(18+3ρ/10) + 64 e^(36+3ρ/10) + 2 e^(36+9ρ/25) + 30 e^(18+19ρ/50) - 86 e^(36+19ρ/50) - 64 e^(18+39ρ/100) +
    48 e^(36+39ρ/100) + 163 e^(18+2ρ/5) - 82 e^(36+2ρ/5) + 2 e^(36+9ρ/20) + 26 e^(18+47ρ/100) - 90 e^(36+47ρ/100) -
    48 e^(36+12ρ/25) + 163 e^(18+49ρ/100) - 278 e^(36+49ρ/100) + 13 e^(36+14ρ/25) - 32 e^(36+57ρ/100) + 233 e^(36+29ρ/50) +
    9 e^(36+13ρ/20) + 205 e^(36+67ρ/100) + 6 e^(9ρ/100) + 13 e^(ρ/5) + 13 e^(29ρ/100) > 0, {ρ}, Reals]]
N[Reduce[3 - 10 e^18 - 10 e^(18+9ρ/100) - 12 e^(18+ρ/10) + 8 e^(36+ρ/10) + 4 e^(18+9ρ/50) - 10 e^(36+9ρ/50) - 12 e^(18+19ρ/100) +
    8 e^(36+19ρ/100) + 54 e^(18+ρ/5) - 20 e^(36+ρ/5) + 4 e^(18+27ρ/100) - 10 e^(36+27ρ/100) - 4 e^(36+7ρ/25) + 54 e^(18+29ρ/100) - 48 e^(36+29ρ/100) +
    e^(36+9ρ/25) - 4 e^(36+37ρ/100) + 54 e^(36+19ρ/50) + e^(36+9ρ/20) + 50 e^(36+47ρ/100) + 3 e^(9ρ/100) > 0, {ρ}, Reals]]
N[Reduce[7 - 2 e^18 - 2 e^(18+9ρ/100) + 5 e^(18+9ρ/50) + 5 e^(18+27ρ/100) + 7 e^(9ρ/100) > 0, {ρ}, Reals]]
```

$-8.64807 < ρ < 0.\ ||\ ρ > 0.$

$-12.9296 < ρ < -2.47642 \times 10^{-8}\ ||\ ρ > 0.$

$ρ < -220.205\ ||\ -18.6172 < ρ < -0.886122\ ||\ ρ > 0.$

$ρ < -207.805\ ||\ -27.2215 < ρ < -3.09527\ ||\ ρ > -2.23033 \times 10^{-8}$

$$\rho < -181.438 \,||\, -29.3002 < \rho < -7.47858 \,||\, \rho > -0.543002$$

$$-177.769 < \rho < -15.2936 \,||\, \rho > -2.06978$$

$$\rho > -5.0905$$

This computations show that p2 > 0 for $\rho$ > 0.
This implies that det(A) > 0 for any real $\rho$ > 0.
Since the matrix B is diagonally dominant and it is symmetric with positive entries for 0 < $\rho$ < 8, then it is positive definite. This means that all principal minors of B are positive on the interval (0,8).
In addition, det(A) > 0 for any real $\rho$ > 0. Therefore, we conclude that the matrix A is positive definite and invertible for 0 < $\rho$ < 8.

# Diagonal Dominance calculation for k >= 4.

We show that for our choice of parameters
Pxxyy - (Pxyxy + (k - 2)(Pxxyz + 2*Pxyxz + Pyzxx + (k - 3)Pxyzw)) > 0,
which establishes that $F^*$ is strictly diagonally dominant and hence generically invertible.

```
k := k
θ := 1 / 10
μ := 1 / 10
τ₃ := 1
τ₂ := 1
τ₁ := 1 / 10
```

```
Collect[Pxxyy - (Pxyxy + (k - 2) (Pxxyz + 2 * Pxyxz + Pyzxx + (k - 3) Pxyzw)), k, Simplify]
```

$$\frac{1}{k^3}\left(1 - \frac{50000\, e^{-11\rho/50}}{(100+\rho)^2} + \frac{200\, e^{-\rho/50}}{100+\rho} + \frac{40000000\, e^{-31\rho/100}}{(100+\rho)(200+\rho)^2} - \frac{80000\, e^{-3\rho/10}}{(100+\rho)(200+\rho)} - \frac{80000\, e^{-21\rho/100}}{(100+\rho)(200+\rho)} + \frac{10000000\, e^{-\frac{2}{5}(45+\rho)}\, \rho^2}{(100+\rho)^2(200+\rho)^2(300+\rho)}\right) +$$

$$\frac{1}{k^4}\left(2 + \frac{20000\, e^{-11\rho/50}}{(100+\rho)^2} - \frac{1000\, e^{-\rho/5}}{100+\rho} - \frac{200\, e^{-\rho/50}}{100+\rho} - \frac{32000000\, e^{-31\rho/100}}{(100+\rho)(200+\rho)^2} + \frac{160000\, e^{-3\rho/10}}{(100+\rho)(200+\rho)} + \frac{160000\, e^{-21\rho/100}}{(100+\rho)(200+\rho)} + \frac{40000\, e^{-\frac{2}{5}(45+\rho)}\, \rho^3}{(100+\rho)^2(200+\rho)^2(300+\rho)}\right) + \frac{1}{k^2}$$

$$\left(-1 + \frac{20000\, e^{-11\rho/50}}{(100+\rho)^2} + \frac{200\, e^{-\rho/5}}{100+\rho} - \frac{12000000\, e^{-31\rho/100}}{(100+\rho)(200+\rho)^2} - \frac{10000\, e^{-\frac{2}{5}(45+\rho)}\, \rho^2(500+\rho)}{(100+\rho)^2(200+\rho)^2(300+\rho)}\right)$$

One can see that the expression simplifies into the following form: $f1(\rho)/k^2 + f2(\rho)/k^3 + f3(\rho)/k^4$

$$f1 := -1 + \frac{20\,000\, e^{-11\rho/50}}{(100+\rho)^2} + \frac{200\, e^{-\rho/5}}{100+\rho} - \frac{12\,000\,000\, e^{-31\rho/100}}{(100+\rho)(200+\rho)^2} - \frac{10\,000\, e^{-\frac{2}{5}(45+\rho)} \rho^2 (500+\rho)}{(100+\rho)^2 (200+\rho)^2 (300+\rho)}$$

$$f2 := 1 - \frac{50\,000\, e^{-11\rho/50}}{(100+\rho)^2} + \frac{200\, e^{-\rho/50}}{100+\rho} + \frac{40\,000\,000\, e^{-31\rho/100}}{(100+\rho)(200+\rho)^2} -$$

$$\frac{80\,000\, e^{-3\rho/10}}{(100+\rho)(200+\rho)} - \frac{80\,000\, e^{-21\rho/100}}{(100+\rho)(200+\rho)} + \frac{10\,000\,000\, e^{-\frac{2}{5}(45+\rho)} \rho^2}{(100+\rho)^2 (200+\rho)^2 (300+\rho)}$$

$$f3 := 2 + \frac{20\,000\, e^{-11\rho/50}}{(100+\rho)^2} - \frac{1000\, e^{-\rho/5}}{100+\rho} - \frac{200\, e^{-\rho/50}}{100+\rho} - \frac{32\,000\,000\, e^{-31\rho/100}}{(100+\rho)(200+\rho)^2} +$$

$$\frac{160\,000\, e^{-3\rho/10}}{(100+\rho)(200+\rho)} + \frac{160\,000\, e^{-21\rho/100}}{(100+\rho)(200+\rho)} + \frac{40\,000\, e^{-\frac{2}{5}(45+\rho)} \rho^3}{(100+\rho)^2 (200+\rho)^2 (300+\rho)}$$

We would like to show that there is an interval (0, n) for some positive real number n on which $f_i > 0$. Notice, that $f_i$ are continuous for all real numbers $\rho > 0$.

The computation of real zeros of f1, the interval on which it is positive and the graph of the function indicate that f1 > 0 for 0 < $\rho$ < 2.

```
NSolve[f1 == 0, ρ, Reals]
```

$\{\{\rho \to -500.\}, \{\rho \to 0.\}, \{\rho \to 2.02268\}\}$

```
Reduce[N[f1] > 0 && ρ > 0, {ρ}, Reals]
```

Reduce::ratnz : Reduce was unable to solve the system with inexact coefficients. The answer was obtained by solving a corresponding exact system and numericizing the result. »

$0 < \rho < 2.02268$

```
Plot[f1, {ρ, 0, 2.1}, PlotTheme → "Detailed",
 PlotStyle → RGBColor[1.`, 0.18`, 0.31`], AxesOrigin → {0, 0}]
```

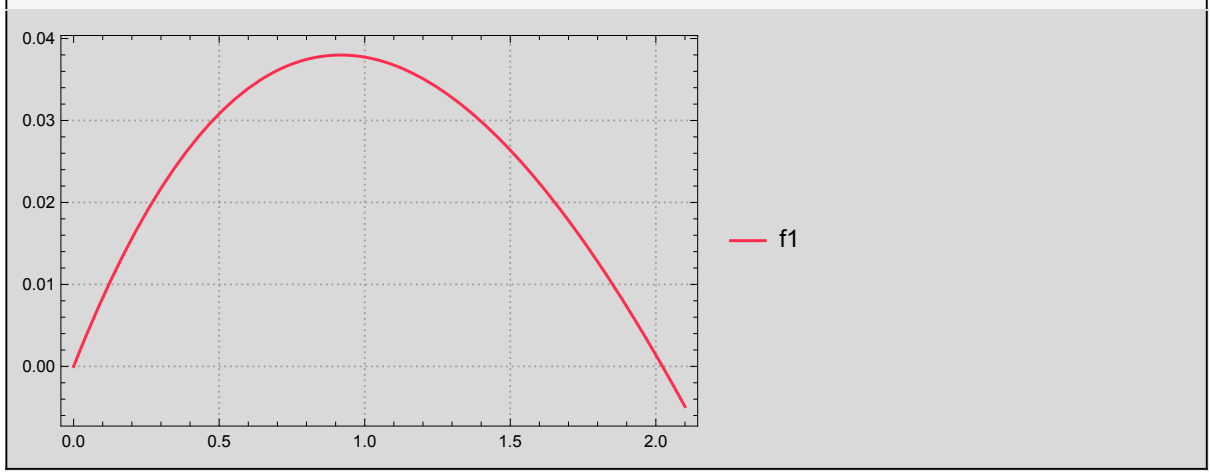

Next, we simplify and collect like terms of f2 and investigate if functions in parentheses are positive for all real numbers $\rho > 0$.

```
Together[f2]
```

$$\frac{1}{(100+\rho)^2 (200+\rho)^2 (300+\rho)}$$
$$e^{-18-\frac{2\rho}{5}} \left(1\,200\,000\,000\,000\, e^{18+\frac{9\rho}{100}} - 480\,000\,000\,000\, e^{18+\frac{\rho}{10}} - 600\,000\,000\,000\, e^{18+\frac{9\rho}{50}} - \right.$$
$$480\,000\,000\,000\, e^{18+\frac{19\rho}{100}} + 240\,000\,000\,000\, e^{18+\frac{19\rho}{50}} + 120\,000\,000\,000\, e^{18+\frac{2\rho}{5}} +$$
$$16\,000\,000\,000\, e^{18+\frac{9\rho}{100}} \rho - 8\,800\,000\,000\, e^{18+\frac{\rho}{10}} \rho - 8\,000\,000\,000\, e^{18+\frac{9\rho}{50}} \rho -$$
$$8\,800\,000\,000\, e^{18+\frac{19\rho}{100}} \rho + 5\,600\,000\,000\, e^{18+\frac{19\rho}{50}} \rho + 4\,000\,000\,000\, e^{18+\frac{2\rho}{5}} \rho +$$
$$10\,000\,000\, \rho^2 + 40\,000\,000\, e^{18+\frac{9\rho}{100}} \rho^2 - 48\,000\,000\, e^{18+\frac{\rho}{10}} \rho^2 - 35\,000\,000\, e^{18+\frac{9\rho}{50}} \rho^2 -$$
$$48\,000\,000\, e^{18+\frac{19\rho}{100}} \rho^2 + 46\,000\,000\, e^{18+\frac{19\rho}{50}} \rho^2 + 51\,000\,000\, e^{18+\frac{2\rho}{5}} \rho^2 -$$
$$80\,000\, e^{18+\frac{\rho}{10}} \rho^3 - 50\,000\, e^{18+\frac{9\rho}{50}} \rho^3 - 80\,000\, e^{18+\frac{19\rho}{100}} \rho^3 + 160\,000\, e^{18+\frac{19\rho}{50}} \rho^3 +$$
$$\left. 310\,000\, e^{18+\frac{2\rho}{5}} \rho^3 + 200\, e^{18+\frac{19\rho}{50}} \rho^4 + 900\, e^{18+\frac{2\rho}{5}} \rho^4 + e^{18+\frac{2\rho}{5}} \rho^5 \right)$$

```
Collect[1 200 000 000 000 e^(18+9ρ/100) - 480 000 000 000 e^(18+ρ/10) - 600 000 000 000 e^(18+9ρ/50) -
  480 000 000 000 e^(18+19ρ/100) + 240 000 000 000 e^(18+19ρ/50) + 120 000 000 000 e^(18+2ρ/5) +
  16 000 000 000 e^(18+9ρ/100) ρ - 8 800 000 000 e^(18+ρ/10) ρ - 8 000 000 000 e^(18+9ρ/50) ρ -
  8 800 000 000 e^(18+19ρ/100) ρ + 5 600 000 000 e^(18+19ρ/50) ρ + 4 000 000 000 e^(18+2ρ/5) ρ +
  10 000 000 ρ^2 + 40 000 000 e^(18+9ρ/100) ρ^2 - 48 000 000 e^(18+ρ/10) ρ^2 - 35 000 000 e^(18+9ρ/50) ρ^2 -
  48 000 000 e^(18+19ρ/100) ρ^2 + 46 000 000 e^(18+19ρ/50) ρ^2 + 51 000 000 e^(18+2ρ/5) ρ^2 -
  80 000 e^(18+ρ/10) ρ^3 - 50 000 e^(18+9ρ/50) ρ^3 - 80 000 e^(18+19ρ/100) ρ^3 + 160 000 e^(18+19ρ/50) ρ^3 +
  310 000 e^(18+2ρ/5) ρ^3 + 200 e^(18+19ρ/50) ρ^4 + 900 e^(18+2ρ/5) ρ^4 + e^(18+2ρ/5) ρ^5, ρ, Simplify]
```

$$120\,000\,000\,000\, e^{18+\frac{9\rho}{100}} \left(10 - 4\, e^{\rho/100} - 5\, e^{9\rho/100} - 4\, e^{\rho/10} + 2\, e^{29\rho/100} + e^{31\rho/100}\right) +$$
$$800\,000\,000\, e^{18+\frac{9\rho}{100}} \left(20 - 11\, e^{\rho/100} - 10\, e^{9\rho/100} - 11\, e^{\rho/10} + 7\, e^{29\rho/100} + 5\, e^{31\rho/100}\right) \rho +$$
$$1\,000\,000 \left(10 + 40\, e^{18+\frac{9\rho}{100}} - 48\, e^{18+\frac{\rho}{10}} - 35\, e^{18+\frac{9\rho}{50}} - 48\, e^{18+\frac{19\rho}{100}} + 46\, e^{18+\frac{19\rho}{50}} + 51\, e^{18+\frac{2\rho}{5}}\right) \rho^2 +$$
$$10\,000\, e^{18+\frac{\rho}{10}} \left(-8 - 5\, e^{2\rho/25} - 8\, e^{9\rho/100} + 16\, e^{7\rho/25} + 31\, e^{3\rho/10}\right) \rho^3 +$$
$$100\, e^{18+\frac{19\rho}{50}} \left(2 + 9\, e^{\rho/50}\right) \rho^4 + e^{18+\frac{2\rho}{5}} \rho^5$$

The above equivalent expression for f2 is a sum of functions. It is clear that f2 will be positive if the terms in parentheses, which are the sums of exponential functions, are positive for $\rho > 0$.
But it is easy to see that those terms are greater than zero for any real number $\rho > 0$.

$10 - 4\, e^{\rho/100} - 5\, e^{9\rho/100} - 4\, e^{\rho/10} + 2\, e^{29\rho/100} + e^{31\rho/100} > 0$

$20 - 11\, e^{\rho/100} - 10\, e^{9\rho/100} - 11\, e^{\rho/10} + 7\, e^{29\rho/100} + 5\, e^{31\rho/100} > 0$

$10 + 40\, e^{18+\frac{9\rho}{100}} - 48\, e^{18+\frac{\rho}{10}} - 35\, e^{18+\frac{9\rho}{50}} - 48\, e^{18+\frac{19\rho}{100}} + 46\, e^{18+\frac{19\rho}{50}} + 51\, e^{18+\frac{2\rho}{5}} > 0$

$-8 - 5\, e^{2\rho/25} - 8\, e^{9\rho/100} + 16\, e^{7\rho/25} + 31\, e^{3\rho/10} > 0$

$2 + 9\, e^{\rho/50} > 0$

To be more convincing we also use Reduce[] function to compute intervals on which these terms are positive.

```
N[Reduce[10 - 4 e^(ρ/100) - 5 e^(9 ρ/100) - 4 e^(ρ/10) + 2 e^(29 ρ/100) + e^(31 ρ/100) > 0, {ρ}, Reals]]
N[Reduce[20 - 11 e^(ρ/100) - 10 e^(9 ρ/100) - 11 e^(ρ/10) + 7 e^(29 ρ/100) + 5 e^(31 ρ/100) > 0, {ρ}, Reals]]
N[Reduce[
   10 + 40 e^(18 + 9 ρ/100) - 48 e^(18 + ρ/10) - 35 e^(18 + 9 ρ/50) - 48 e^(18 + 19 ρ/100) + 46 e^(18 + 19 ρ/50) + 51 e^(18 + 2 ρ/5) > 0, {ρ}, Reals]]
N[Reduce[-8 - 5 e^(2 ρ/25) - 8 e^(9 ρ/100) + 16 e^(7 ρ/25) + 31 e^(3 ρ/10) > 0, {ρ}, Reals]]
```

ρ ≠ 0.

ρ < -6.33611 || ρ > 0.

ρ < -30.3637 || ρ > -0.307214

ρ > -3.31953

This computations show that f2 > 0 for any real number ρ > 0.
Below is the graph of f2.

```
Plot[f2, {ρ, 0, 150}, PlotTheme → "Detailed",
 PlotStyle → RGBColor[1.`, 0.18`, 0.31`], AxesOrigin → {0, 0}]
```

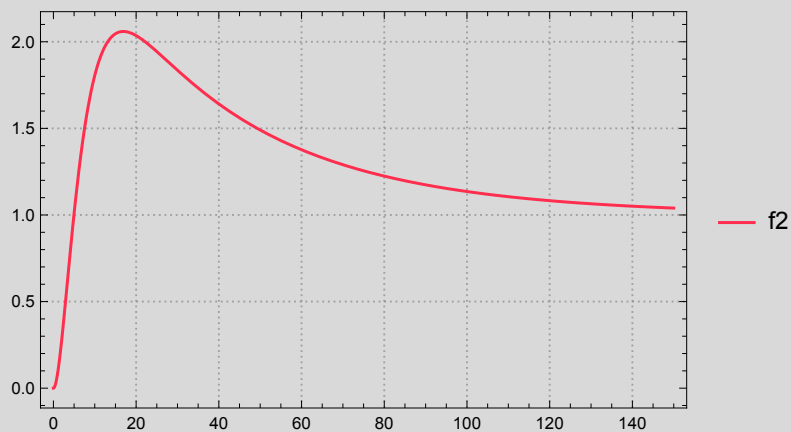

```
Limit[f2, ρ → Infinity]
```

1

Now, we repeat same procedure for f3.

```
Together[f3]
```

$$\frac{1}{(100+\rho)^2 (200+\rho)^2 (300+\rho)}$$
$$2 e^{-18-\frac{2\rho}{5}} \left(-480\,000\,000\,000\, e^{18+\frac{9\rho}{100}} + 480\,000\,000\,000\, e^{18+\frac{\rho}{10}} + 120\,000\,000\,000\, e^{18+\frac{9\rho}{50}} + \right.$$
$$480\,000\,000\,000\, e^{18+\frac{19\rho}{100}} - 600\,000\,000\,000\, e^{18+\frac{\rho}{5}} - 120\,000\,000\,000\, e^{18+\frac{19\rho}{50}} +$$
$$120\,000\,000\,000\, e^{18+\frac{2\rho}{5}} - 6\,400\,000\,000\, e^{18+\frac{9\rho}{100}} \rho + 8\,800\,000\,000\, e^{18+\frac{\rho}{10}} \rho +$$
$$1\,600\,000\,000\, e^{18+\frac{9\rho}{50}} \rho + 8\,800\,000\,000\, e^{18+\frac{19\rho}{100}} \rho - 14\,000\,000\,000\, e^{18+\frac{\rho}{5}} \rho -$$
$$2\,800\,000\,000\, e^{18+\frac{19\rho}{50}} \rho + 4\,000\,000\,000\, e^{18+\frac{2\rho}{5}} \rho - 16\,000\,000\, e^{18+\frac{9\rho}{100}} \rho^2 +$$
$$48\,000\,000\, e^{18+\frac{\rho}{10}} \rho^2 + 7\,000\,000\, e^{18+\frac{9\rho}{50}} \rho^2 + 48\,000\,000\, e^{18+\frac{19\rho}{100}} \rho^2 - 115\,000\,000\, e^{18+\frac{\rho}{5}} \rho^2 -$$
$$23\,000\,000\, e^{18+\frac{19\rho}{50}} \rho^2 + 51\,000\,000\, e^{18+\frac{2\rho}{5}} \rho^2 + 20\,000\, \rho^3 + 80\,000\, e^{18+\frac{\rho}{10}} \rho^3 +$$
$$10\,000\, e^{18+\frac{9\rho}{50}} \rho^3 + 80\,000\, e^{18+\frac{19\rho}{100}} \rho^3 - 400\,000\, e^{18+\frac{\rho}{5}} \rho^3 - 80\,000\, e^{18+\frac{19\rho}{50}} \rho^3 +$$
$$\left. 310\,000\, e^{18+\frac{2\rho}{5}} \rho^3 - 500\, e^{18+\frac{\rho}{5}} \rho^4 - 100\, e^{18+\frac{19\rho}{50}} \rho^4 + 900\, e^{18+\frac{2\rho}{5}} \rho^4 + e^{18+\frac{2\rho}{5}} \rho^5\right)$$

```
Collect[-480 000 000 000 e^(18+9ρ/100) + 480 000 000 000 e^(18+ρ/10) + 120 000 000 000 e^(18+9ρ/50) +
  480 000 000 000 e^(18+19ρ/100) - 600 000 000 000 e^(18+ρ/5) - 120 000 000 000 e^(18+19ρ/50) +
  120 000 000 000 e^(18+2ρ/5) - 6 400 000 000 e^(18+9ρ/100) ρ + 8 800 000 000 e^(18+ρ/10) ρ + 1 600 000 000 e^(18+9ρ/50) ρ +
  8 800 000 000 e^(18+19ρ/100) ρ - 14 000 000 000 e^(18+ρ/5) ρ - 2 800 000 000 e^(18+19ρ/50) ρ + 4 000 000 000 e^(18+2ρ/5) ρ -
  16 000 000 e^(18+9ρ/100) ρ^2 + 48 000 000 e^(18+ρ/10) ρ^2 + 7 000 000 e^(18+9ρ/50) ρ^2 + 48 000 000 e^(18+19ρ/100) ρ^2 -
  115 000 000 e^(18+ρ/5) ρ^2 - 23 000 000 e^(18+19ρ/50) ρ^2 + 51 000 000 e^(18+2ρ/5) ρ^2 + 20 000 ρ^3 +
  80 000 e^(18+ρ/10) ρ^3 + 10 000 e^(18+9ρ/50) ρ^3 + 80 000 e^(18+19ρ/100) ρ^3 - 400 000 e^(18+ρ/5) ρ^3 - 80 000 e^(18+19ρ/50) ρ^3 +
  310 000 e^(18+2ρ/5) ρ^3 - 500 e^(18+ρ/5) ρ^4 - 100 e^(18+19ρ/50) ρ^4 + 900 e^(18+2ρ/5) ρ^4 + e^(18+2ρ/5) ρ^5, ρ, Simplify]
```

$$120\,000\,000\,000\, e^{18+\frac{9\rho}{100}} \left(-4 + 4\, e^{\rho/100} + e^{9\rho/100} + 4\, e^{\rho/10} - 5\, e^{11\rho/100} - e^{29\rho/100} + e^{31\rho/100}\right) +$$
$$400\,000\,000\, e^{18+\frac{9\rho}{100}}$$
$$\left(-16 + 22\, e^{\rho/100} + 4\, e^{9\rho/100} + 22\, e^{\rho/10} - 35\, e^{11\rho/100} - 7\, e^{29\rho/100} + 10\, e^{31\rho/100}\right) \rho + 1\,000\,000$$
$$e^{18+\frac{9\rho}{100}} \left(-16 + 48\, e^{\rho/100} + 7\, e^{9\rho/100} + 48\, e^{\rho/10} - 115\, e^{11\rho/100} - 23\, e^{29\rho/100} + 51\, e^{31\rho/100}\right) \rho^2 +$$
$$10\,000 \left(2 + 8\, e^{18+\frac{\rho}{10}} + e^{18+\frac{9\rho}{50}} + 8\, e^{18+\frac{19\rho}{100}} - 40\, e^{18+\frac{\rho}{5}} - 8\, e^{18+\frac{19\rho}{50}} + 31\, e^{18+\frac{2\rho}{5}}\right) \rho^3 +$$
$$100\, e^{18+\frac{\rho}{5}} \left(-5 - e^{9\rho/50} + 9\, e^{\rho/5}\right) \rho^4 + e^{18+\frac{2\rho}{5}} \rho^5$$

The above equivalent expression for f3 is a sum of functions. It is clear that f3 will be positive if the terms in parentheses, which are the sums of exponential functions, are positive for $\rho > 0$. It is easy to see that those terms are greater than zero for $\rho>0$.

$-4 + 4\, e^{\rho/100} + e^{9\rho/100} + 4\, e^{\rho/10} - 5\, e^{11\rho/100} - e^{29\rho/100} + e^{31\rho/100} > 0$

$-16 + 22\, e^{\rho/100} + 4\, e^{9\rho/100} + 22\, e^{\rho/10} - 35\, e^{11\rho/100} - 7\, e^{29\rho/100} + 10\, e^{31\rho/100} > 0$

$-16 + 48\, e^{\rho/100} + 7\, e^{9\rho/100} + 48\, e^{\rho/10} - 115\, e^{11\rho/100} - 23\, e^{29\rho/100} + 51\, e^{31\rho/100} > 0$

$2 + 8\, e^{18+\frac{\rho}{10}} + e^{18+\frac{9\rho}{50}} + 8\, e^{18+\frac{19\rho}{100}} - 40\, e^{18+\frac{\rho}{5}} - 8\, e^{18+\frac{19\rho}{50}} + 31\, e^{18+\frac{2\rho}{5}} > 0$

$-5 - e^{9\rho/50} + 9\, e^{\rho/5} > 0$

To be more convincing we also use Reduce[] function to compute intervals on which these terms are positive.

```
N[Reduce[-4 + 4 e^(ρ/100) + e^(9 ρ/100) + 4 e^(ρ/10) - 5 e^(11 ρ/100) - e^(29 ρ/100) + e^(31 ρ/100) > 0, {ρ}, Reals]]
N[Reduce[
    -16 + 22 e^(ρ/100) + 4 e^(9 ρ/100) + 22 e^(ρ/10) - 35 e^(11 ρ/100) - 7 e^(29 ρ/100) + 10 e^(31 ρ/100) > 0, {ρ}, Reals]]
N[Reduce[-16 + 48 e^(ρ/100) + 7 e^(9 ρ/100) + 48 e^(ρ/10) - 115 e^(11 ρ/100) - 23 e^(29 ρ/100) + 51 e^(31 ρ/100) > 0,
    {ρ}, Reals]]
N[Reduce[2 + 8 e^(18+ρ/10) + e^(18+9ρ/50) + 8 e^(18+19ρ/100) - 40 e^(18+ρ/5) - 8 e^(18+19ρ/50) + 31 e^(18+2ρ/5) > 0, {ρ}, Reals]]
N[Reduce[-5 - e^(9 ρ/50) + 9 e^(ρ/5) > 0, {ρ}, Reals]]
```

ρ > 0.

-32.3883 < ρ < 0. || ρ > 0.

-109.864 < ρ < -3.55154 || ρ > 0.

ρ < -12.5104 || ρ > -7.89118 × 10⁻⁹

ρ > -2.32024

This computations show that f3 > 0 for any real number ρ > 0. Below is the graph f3.

```
Plot[f3, {ρ, 0, 200}, PlotTheme → "Detailed",
 PlotStyle → RGBColor[1.`, 0.18`, 0.31`], AxesOrigin → {0, 0}]
```

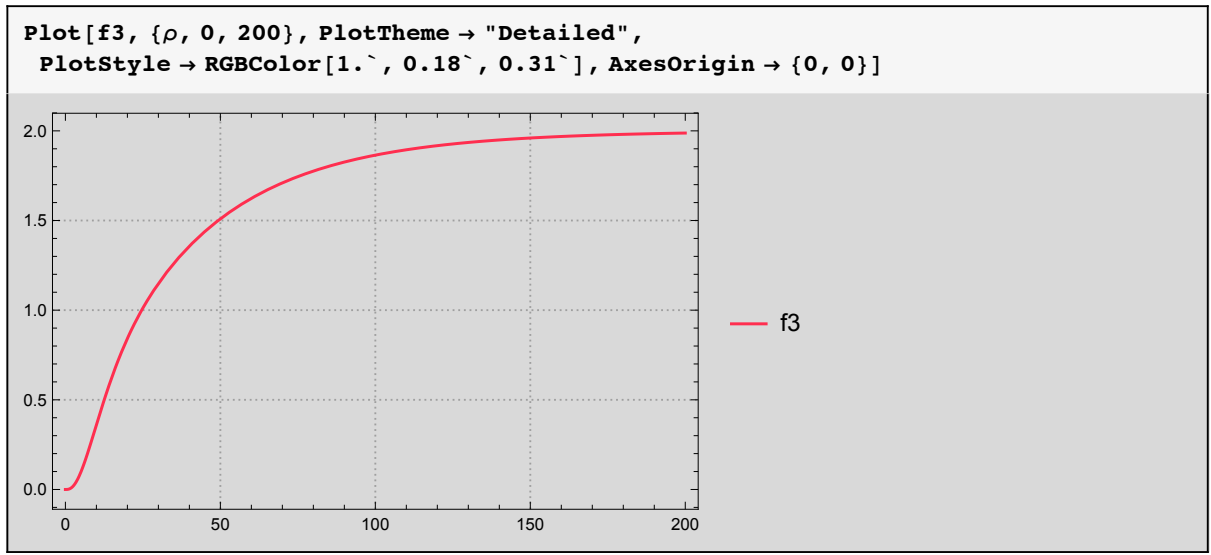

```
Limit[f3, ρ → Infinity]
```

2

Thus, we see that f1(ρ)/k² + f2(ρ)/k³ + f3(ρ)/k⁴ > 0 for any real number ρ in the (0,2) interval.
We conclude that *F*ˆ is strictly diagonally dominant and hence positive definite for 0 < ρ < 2.



\end{document}